\documentclass[11pt]{article}

\usepackage[margin=1in]{geometry}

\usepackage{fullpage}
\usepackage[utf8]{inputenc}
\usepackage[english]{babel}
\usepackage[shortlabels]{enumitem}
\usepackage{mathtools,
  booktabs,
  braket,
  mathrsfs,
  latexsym,
  epsfig,
  xcolor,
  url,
  setspace,
  graphics,
  lipsum,
  float,
  lineno,
  placeins,
  pifont
}

\usepackage{graphicx}
\usepackage{subcaption}
\usepackage[labelfont=bf]{caption}

\usepackage{float}
\usepackage{pdflscape}
\newcommand{\bland}{\begin{landscape}}
\newcommand{\eland}{\end{landscape}}

\usepackage[flushmargin]{footmisc}
\definecolor{CardinalRed}{cmyk}{0,1,0.65,0.34} 
\usepackage[colorlinks,linkcolor=CardinalRed,citecolor=CardinalRed,urlcolor = CardinalRed]{hyperref}
\usepackage{hyperref}

\makeatletter
\def\maxwidth{\ifdim\Gin@nat@width>\linewidth\linewidth\else\Gin@nat@width\fi}
\def\maxheight{\ifdim\Gin@nat@height>\textheight\textheight\else\Gin@nat@height\fi}
\makeatother

\setkeys{Gin}{width=\maxwidth,height=\maxheight,keepaspectratio}

\newcommand{\burl}[1]{\textcolor{blue}{\url{#1}}}

\usepackage[round]{natbib}

\usepackage[OT1]{fontenc}

\usepackage{todonotes}

\makeatletter
\providecommand\@dotsep{5}
\def\listtodoname{List of Todos}
\def\listoftodos{\@starttoc{tdo}\listtodoname}
\makeatother

\graphicspath{{./graphs/}}%

\usepackage[labelsep=period, font=sc, skip=0pt,justification=centering]{caption}
\usepackage[figuresright]{rotating}
\usepackage{titlesec}
\titlelabel{\thetitle.\ }

\usepackage{titling}
\usepackage{titlesec}
\titleformat{\section}
{\normalfont\fontsize{15}{15}\bfseries}{\thesection.}{0.5em}{}

\newenvironment{itemize*}%
  {\begin{itemize}%
    \setlength{\itemsep}{0pt}%
    \setlength{\parskip}{0pt}}%
  {\end{itemize}}
\newenvironment{enumerate*}%
  {\begin{enumerate}%
    \setlength{\itemsep}{0pt}%
    \setlength{\parskip}{0pt}}%
  {\end{enumerate}}

\usepackage{
  amscd,
  amsfonts,
  amsmath,
  amssymb,
  amsbsy,
  amsthm,
  bm, 
  dsfont,
  cancel, 
  latexsym,
  mathtools,
  graphicx,
  xcolor,
  xargs
}

\usepackage{longtable}

\newcommand{\beq}{\begin{equation}}
\newcommand{\eeq}{\end{equation}}

\newcommand*\Bigpar[1]{\left( #1 \right )}

\newcommand{\ba}{\begin{array}}
\newcommand{\ea}{\end{array}}
\newcommand{\be}{\begin{enumerate}}
\newcommand{\ee}{\end{enumerate}}
\newcommand{\bi}{\begin{itemize}}
\newcommand{\ei}{\end{itemize}}
\newcommand{\bs}{\begin{align}\begin{split}\nonumber}
\newcommand{\bsnumber}{\begin{align}\begin{split}}
\newcommand{\es}{\end{split}\end{align}}

\newcommandx{\deriv}[2][1=x,2=f]{\nabla \, #2 \Bigpar{ #1 } }

\newcommand{\abs}[1]{\left\vert {#1} \right\vert}

\renewcommand{\dim}{\mathrm{dim}}

\renewcommand{\to}{{\rightarrow}}

\newcommand{\R}{\ensuremath{\mathbb{R}}}

\newcommand\frakfamily{\usefont{U}{yfrak}{m}{n}}
\DeclareTextFontCommand{\textfrak}{\frakfamily}

\newcommand{\argmin}{\operatornamewithlimits{arg\,min}}

\renewcommand{\dim}{\mathrm{dim}}
\renewcommand{\to}{{\rightarrow}}

\providecommand{\argmin}{\mathop{\rm argmin}}

\providecommand{\diag}{\mathop{\rm diag}}

\providecommand{\abs}{\mathop{\rm abs}}

\newcommand{\hyp}[2]{
\ensuremath{H_0:#1 \ifhmode\quad\text{versus}\quad\fi\text{ vs. } H_1:#2}}

\newcommandx{\uniff}[1][1={a,b}]{\textrm{Unif}\left({#1}\right)}
\newcommandx{\unifd}[1][1={a,\ldots,b}]{\textrm{Unif}\left\{{#1}\right\}}
\newcommandx{\dunif}[3][1=x,2=a,3=b]{\frac{I(#2<#1<#3)}{#3-#2}}
\newcommandx{\dunifd}[3][1=x,2=a,3=b]{\frac{I(#2\le#1\le#3)}{#3-#2+1}}
\newcommandx{\punif}[3][1=x,2=a,3=b]{
\begin{cases} 0 & #1 < #2 \\ \frac{#1-#2}{#3-#2} & #2 < #1 < #3 \\ 1 & #1 > #3\\\end{cases}}
\newcommandx{\punifd}[3][1=x,2=a,3=b]{
\begin{cases} 0 & #1 < #2\\ \frac{\lfloor#1\rfloor-#2+1}{#3-#2} & #2 \le #1 \le #3 \\ 1 & #1 > #3\\ \end{cases}}

\newcommandx\bern[1][1=p]{\textrm{Bern}\left({#1}\right)}
\newcommandx\dbern[2][1=x,2=p]{#2^{#1} \left(1-#2\right)^{1-#1}}
\newcommandx\pbern[2][1=x,2=p]{\left(1-#2\right)^{1-#1}}

\newcommandx\bin[1][1={n,p}]{\textrm{Bin}\left(#1\right)}
\newcommandx\dbin[3][1=x,2=n,3=p]{\binom{#2}{#1}#3^#1\left(1-#3\right)^{#2-#1}}

\newcommandx\mult[1][1={n,p}]{\textrm{Mult}\left(#1\right)}
\newcommandx\dmult[3][1=x,2=n,3=p]{\frac{#2!}{#1_1!\ldots#1_k!}#3_1^{#1_1}\cdots#3_k^{#1_k}}

\newcommandx\hyper[1][1={N,m,n}]{\textrm{Hyp}\left({#1}\right)}
\newcommandx\dhyper[4][1=x,2=N,3=m,4=n]{\frac{\binom{#3}{#1}\binom{#2-#3}{#4-#1}}{\binom{#2}{#4}}}

\newcommandx\nbin[1][1={r,p}]{\textrm{NBin}\left({#1}\right)}
\newcommandx\dnbin[3][1=x,2=r,3=p]{\binom{#1+#2-1}{#2-1}#3^#2(1-#3)^#1}
\newcommandx\pnbin[3][1=x,2=r,3=p]{I_#3(#2,#1+1)}

\newcommandx\geo[1][1=p]{\textrm{Geo}\left(#1\right)}
\newcommandx\dgeo[2][1=x,2=p]{#2(1-#2)^{#1-1}}
\newcommandx\pgeo[2][1=x,2=p]{1-(1-#2)^#1}

\newcommandx\pois[1][1=\lambda]{\textrm{Po}\left({#1}\right)}
\newcommandx\dpois[2][1=x,2=\lambda]{\frac{#2^#1 e^{-#2}}{#1!}}
\newcommandx\ppois[2][1=x,2=\lambda]{e^{-#2}\sum_{i=0}^#1\frac{#2^i}{i!}}

\newcommandx\normall[1][1={\mu,\sigma^2}]{\mathcal{N}\left({#1}\right)}
\newcommandx\dnormall[3][1=x,2=\mu,3=\sigma]%
  {\frac{1}{#3\sqrt{2\pi}}\exp \Bigpar{-\frac{\left(#1-#2\right)^2}{2 #3^2}}}
\newcommandx\pnormall[1][1=x]{\Phi\left({#1}\right)}
\newcommandx\qnormall[1]{\Phi^{-1}\left({#1}\right)}

\newcommandx\mvn[1][1={\mu,\Sigma}]{\mathrm{MVN}\left({#1}\right)}

\newcommandx\ex[1][1=\lambda]{\textrm{Exp}\left(#1\right)}
\newcommandx\dex[2][1=x,2=\lambda]{#2e^{-#1 #2}}
\newcommandx\pex[2][1=x,2=\lambda]{1-e^{-#1 #2}}

\newcommandx\gam[1][1={\alpha,\lambda}]{\textrm{Gamma}\left({#1}\right)}
\newcommandx\dgamma[3][1=x,2=\alpha,3=\lambda]%
  {\frac{#3^{#2}}{\Gamma\left( #2 \right)} #1^{#2-1}e^{-#3#1}}

\newcommandx\invgamma[1][1={\alpha,\beta}]{\textrm{InvGamma}\left({#1}\right)}
\newcommandx\dinvgamma[3][1=x,2=\alpha,3=\beta]%
{\frac{#3^{#2}}{\Gamma\left(#2\right)}#1^{-#2-1}e^{-#3/#1}}
\newcommandx\pinvgamma[3][1=x,2=\alpha,3=\beta]%
{\frac{\Gamma\left(#2,\frac{#3}{#1}\right)}{\Gamma\left(#2\right)}}

\newcommandx\bet[1][1={\alpha,\beta}]{\textrm{Beta}\left(#1\right)}
\newcommandx\dbeta[3][1=x,2=\alpha,3=\beta]
{\frac{\Gamma\left(#2+#3\right)}{\Gamma\left(#2\right)\Gamma\left(#3\right)}#1^{#2-1}\left(1-#1\right)^{#3-1}}

\newcommandx\dir[1][1={\alpha}]{\textrm{Dir}\left(#1\right)}
\newcommandx\ddir[3][1=x,2=\alpha]{\frac{\Gamma\left(\sum_{i=1}^k #2_i\right)}{\prod_{i=1}^k\Gamma\left(#2_i\right)}\prod_{i=1}^k #1_i^{#2_i-1}}

\newcommandx\weibull[1][1={\alpha}]{\textrm{Dir}\left(#1\right)}
\newcommandx\dweibull[3][1=x,2=\lambda,3=k]{\frac{#3}{#2}
\left(\frac{#1}{#2}\right)^{#3-1} e^{-(#1/#2)^k}}

\newcommandx\chisq[1][1=k]{\chi_{#1}^2}

\newcommandx\zet[1][1=s]{\textrm{Zeta}\left(#1\right)}
\newcommandx\dzeta[2][1=x,2=s]{\frac{#1^{-#2}}{\zeta\left(#2\right)}}

\newtheoremstyle{mystyle}
  {12pt}
  {12pt}
  {}
  {}
  {\sffamily \bfseries }
  {.}
  {0.5em}
  {\thmname{#1}\thmnumber{ #2}\thmnote{ (#3)}}

\theoremstyle{mystyle}

\newtheorem{definition}{Definition}
\newtheorem{lemma}{Lemma}
\newtheorem{proposition}{Proposition}

\renewenvironment{proof}{\noindent{\bf Proof}\hspace*{1em}}{\qed\bigskip\\}
\newenvironment{proof-sketch}{\noindent{\bf Sketch of Proof}
  \hspace*{1em}}{\qed\bigskip\\}
\newenvironment{proof-idea}{\noindent{\bf Proof Idea}
  \hspace*{1em}}{\qed\bigskip\\}
\newenvironment{proof-of-lemma}[1][{}]{\noindent{\bf Proof of Lemma {#1}}
  \hspace*{1em}}{\qed\bigskip\\}
\newenvironment{proof-of-proposition}[1][{}]{\noindent{\bf
    Proof of Proposition {#1}}
  \hspace*{1em}}{\qed\bigskip\\}
\newenvironment{proof-of-theorem}[1][{}]{\noindent{\bf Proof of Theorem {#1}}
  \hspace*{1em}}{\qed\bigskip\\}
\newenvironment{inner-proof}{\noindent{\bf Proof}\hspace{1em}}{
  $\bigtriangledown$\medskip\\}
\newenvironment{proof-attempt}{\noindent{\bf Proof Attempt}
  \hspace*{1em}}{\qed\bigskip\\}

\usepackage{array}
\newcolumntype{L}[1]{>{\raggedright\let\newline\\\arraybackslash\hspace{0pt}}m{#1}}
\newcolumntype{C}[1]{>{\centering\let\newline\\\arraybackslash\hspace{0pt}}m{#1}}
\newcolumntype{R}[1]{>{\raggedleft\let\newline\\\arraybackslash\hspace{0pt}}m{#1}}

\usepackage{multirow}

\usepackage{makecell} 
\newcommand{\cR}{\mathcal{R}}
\newcommand{\1}{\mathbf{1}}
\newcommand{\pre}{\mathrm{pre}}
\newcommand{\post}{\mathrm{post}}
\newcommand{\oracle}{\mathrm{oracle}}
\newcommand{\Null}{\mathrm{Null}}

\graphicspath{{./graphs/}}%

\begin{document}


\title{\Large\bf The Harmonic Synthetic Control Method
\thanks{Ziyi Liu, PhD student, Haas School of Business, University of California, Berkeley. Email: \url{zyliu2023@berkeley.edu}. Yiqing Xu, Assistant Professor, Department of Political Science, Stanford University. Email: \url{yiqingxu@stanford.edu}. This project was inspired by discussions with Hongyu Mou and Yifan Sun regarding the replication of synthetic control studies. We thank David Bruns-Smith, Kevin Chen, Alex Hayes, Guido Imbens, Lihua Lei, David Ritzwoller, Sarah Vicol, and participants at the Stanford Metrics Lunch for their helpful comments. The authors used Claude Code Opus 4.7 as a research, coding, and writing assistant in preparing this manuscript. All errors remain solely the responsibility of the authors.
}
\\\bigskip}

\author{Ziyi Liu \\(UC Berkeley)\and Yiqing Xu\\(Stanford)}

\date{\bigskip \today
}


\maketitle

\vspace{-2em}
\begin{abstract}
\noindent Synthetic control methods can produce misleading counterfactual predictions when outcome series contain unit-specific stochastic trends, a common feature of nonstationary macroeconomic data. Existing remedies, such as pre-filtering or differencing, reduce spurious matching but may discard \emph{shared} nonstationary variation that helps estimate donor weights. We propose Harmonic Synthetic Control (HSC), which replaces this binary choice with a soft allocation mechanism. HSC jointly estimates donor weights and a treated-unit-specific smooth residual component, then extrapolates this component into post-treatment periods using a time-series forecaster. A tuning parameter, selected by rolling-origin cross-validation, governs the division between donor matching and forecasting. As it varies, HSC continuously interpolates between synthetic control applied to differenced outcomes and synthetic control applied to raw outcomes with an intercept or trend. We provide a spectral interpretation showing how HSC downweights low-frequency residual components in donor matching and assigns them to the forecasting branch. A prediction-error decomposition separates weight-estimation distortion from residual-forecasting error. Monte Carlo exercises show that HSC adapts across regimes, performing well when stochastic trends are predominantly common or idiosyncratic, while estimators fixed to one regime can fail in the other.

\end{abstract}

\thispagestyle{empty}  
\clearpage
\newpage
\doublespacing

\clearpage

\setcounter{page}{1}
\abovedisplayskip=5pt
\belowdisplayskip=5pt


\section{Introduction}
\label{sec:intro}

Synthetic control methods construct counterfactuals for treated units by finding weighted combinations of untreated donors that match the treated unit's pre-treatment outcome path \citep{abadie2003ecc,abadie2010synthetic, abadie2015comparative}. The logic is that if a weighted combination of donors can reproduce the treated unit's outcomes before treatment, the same combination should approximate what the treated unit's outcomes would have been in the absence of treatment. A difficulty arises when outcome series are nonstationary: the pre-treatment fit that synthetic control exploits may be spurious. Specifically, when units contain unit-specific stochastic trends, a convex combination of donors can closely track the treated unit's pre-treatment path through coincidental co-movement rather than shared structure, producing in-sample fit that breaks down out of sample and leading to biased estimation and distorted inference
\citep{phillips1986understanding,masini2021counterfactual,
masini2022counterfactual,shi2025synthetic}.

Existing responses to this problem face a tradeoff over whether to preprocess the data, such as detrending or differencing. These transformations turn nonstationary time series into stationary ones, thereby alleviating the spurious matching risk. However, these transformations also discard nonstationary variation that the treated unit potentially shares with donor units, which is the main source of identifying variation for donor-weight estimation in synthetic control \citep{ferman2021synthetic,abadie2021using}. Other variants of synthetic control, such as augmented or bias-corrected extensions, can reduce the residual imbalance left by imperfect pre-treatment fit \citep{ben2021augmented,arkhangelsky2021synthetic}, but because they still begin from weights chosen to match raw pre-treatment outcomes, they can inherit the same underlying weight distortion caused by spurious matching.

We formalize this tradeoff through a conceptual distinction. In macroeconomic panel data, nonstationarity typically takes the form of persistent stochastic components, such as random walks or other integrated processes, whose variance grows without bound over time. We categorize this stochastic trend variation into two sources. By a \emph{shared stochastic trend}, we mean a stochastic trend component whose innovations are shared across units; such a trend moves the treated unit and donors together, possibly with unit-specific responses. Synthetic control is designed to exploit this shared structure, using a weighted combination of donors to approximate the treated unit's trajectory. By an \emph{idiosyncratic stochastic trend}, we mean a stochastic trend component whose realizations are unit-specific and do not generate stable comovement across units, even though they may appear correlated by chance in any finite sample. This is the source of the spurious matching problem. We use \emph{stochastic trend} as an umbrella term for both.

\begin{figure}[!h]
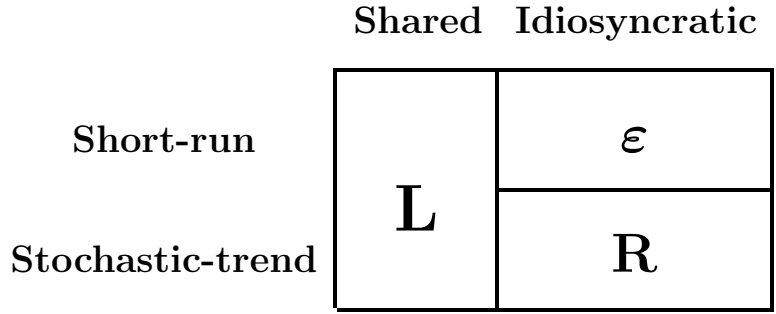

    \centering
    \setlength{\arrayrulewidth}{1.5pt}
    \renewcommand{\arraystretch}{2.5}
    \Large
    \begin{tabular}{c|c|c|}
         \multicolumn{1}{c}{} & \multicolumn{1}{c}{\textbf{Shared}} & \multicolumn{1}{c}{\textbf{Idiosyncratic}} \\ \cline{2-3}
        \textbf{Short-run} & \multirow{2}{*}{\centering \fontsize{70}{80}\selectfont \textbf{L}} & \huge $\boldsymbol{\varepsilon}$ \\ \cline{3-3}
        \textbf{Stochastic-trend} &  & \huge \textbf{R} \\ \cline{2-3}
    \end{tabular}\\ \medskip
    \caption{Schematic decomposition of untreated potential outcomes.}
    \label{fig:econometric_problem}
    \vspace{0.5em}
    \begin{minipage}{0.88\textwidth}
    \footnotesize
    \textit{Notes:} $L$ denotes the component governed by the shared factors, which may contain
    both short-run and stochastic trend latent factors. $\varepsilon$ denotes
    idiosyncratic short-run noise. $R$ denotes an idiosyncratic stochastic
    trend that is not governed by the shared factor structure.
    \end{minipage}
\end{figure}

Based on these concepts, we decompose untreated potential outcomes of the treated unit and all control donors into three components (Figure~\ref{fig:econometric_problem}). The first is a component $L$, possibly low-rank and driven by shared latent factors whose loadings vary across units. These factors may include both stochastic trend and short-run variations.\footnote{We avoid the labels stationary and nonstationary because what matters for our analysis is bounded long-run variance, not strict stationarity. By short-run we mean variation with bounded long-run variance, which includes strictly stationary processes as well as bounded-variance nonstationary processes.} The second is idiosyncratic short-run noise $\varepsilon$. It does not fundamentally distort donor-weight estimation: it averages out over long pre-treatment windows under standard moment conditions. The third is the idiosyncratic stochastic trend $R$, which is not governed by the shared factor structure. Unlike $\varepsilon$, $R$ can severely distort donor-weight estimation: independent stochastic trends produce realized correlations whose magnitude does not vanish as $T_0$ grows, so longer pre-treatment windows do not help \citep{phillips1986understanding}. Whether the stochastic trend variation in the treated unit's series is mostly shared(captured by $L$) or idiosyncratic (captured by $R$) is generally unknown to the researcher ex ante. In the first case, preprocessing removes the variation that synthetic control would otherwise use for donor-weight estimation. In the second case, allowing stochastic trends to enter weight estimation creates the risk of spurious matching. This tradeoff motivates us to propose a solution that adapts to both regimes.

This paper proposes \emph{harmonic synthetic control} (HSC), which replaces this binary choice of preprocessing with a soft, data-driven allocation. HSC jointly estimates donor weights and a treated-unit-specific smooth component $E$ that absorbs idiosyncratic trend-like variation not reproducible by a convex combination of donors. The roughness of $E$ is controlled by a penalty $\|D_q E\|_2^2$, where $D_q$ is the $q$th-order difference operator, with $q \in \{1, 2\}$ specifying the order of smoothness. The division of labor between donor matching and the smooth component is governed by a single tuning parameter $\rho \in [0,1]$, selected by rolling-origin cross-validation. When $\rho$ is close to $1$, HSC imposes a high penalty on the roughness of $E$ and recovers synthetic control with an intercept or with an intercept plus a linear trend, depending on whether $q=1$ or $q=2$. When $\rho$ is close to $0$, $E$ absorbs the entire discrepancy between the treated unit and the convex combination of donors, and the weight estimation in HSC approaches synthetic control applied to first or second differences of the raw outcomes (for $q=1$ or $q=2$, respectively). At intermediate $\rho$, HSC continuously interpolates between these two endpoints. In post-treatment periods, the smooth component is forecast by a time series forecaster, and the counterfactual is constructed by adding the forecast of $E$ to the donor matching component. HSC does not attempt to disentangle which portion of the stochastic trend variation is shared and which is idiosyncratic; instead, cross-validation selects the allocation between donor matching and the smooth component that yields the best out-of-sample predictive performance.

This soft allocation has a clean spectral interpretation. Any pre-treatment series can be decomposed into components at different frequencies: low-frequency content carries slowly varying (trend-like) variation, while high-frequency content carries short-run variation. The HSC weight estimation problem down-weights low-frequency components and amplifies high-frequency components, thus alleviating the spurious matching risk. The treated-unit-specific component $E$ absorbs the low-frequency residual variation that remains after donor matching, and the post-treatment forecast of $E$ extrapolates this low-frequency, trend-like component forward in time. Equivalently, HSC can be understood as applying synthetic control to a soft spectral transformation of the raw data that interpolates between two extremes: at $\rho = 0$, the transformation reduces to $q$-th order differencing, which strongly suppresses low-frequency content; at $\rho = 1$, the transformation removes only the null-space component (constants for $q=1$, constants plus linear trends for $q=2$), leaving other frequencies unchanged. The method's name reflects the form of this interpolation: at each frequency, the spectral gain of the HSC transformation is a weighted harmonic mean of the gains at the two endpoints, with weights $1-\rho$ and $\rho$.

We develop an envelope bound on the HSC counterfactual's prediction error. The error decomposes into a weight-estimation term and a forecasting term, each depending on $\rho$. The weight-estimation term reflects that the researcher observes $Y = L + R + \varepsilon$ rather than the shared component $L$ alone; the HSC-estimated donor weights therefore deviate from the oracle weights constructed from $L$ alone. This term captures the tradeoff between the risk of spurious matching at large $\rho$ and the downweighting of useful low-frequency variation in $L$ at small $\rho$. The forecasting term captures the error that would remain even with oracle weights; depending on the prediction quality of the time series forecaster, this term can be monotonically increasing, monotonically decreasing, or non-monotonic in $\rho$. Together these two terms determine the tradeoff that cross-validation aims to balance in the choice of $\rho$.

The synthetic control literature has produced many methods that match donors to the treated unit's pre-treatment outcomes, differing in weight constraints, bias-correction strategies, and temporal aggregation \citep{doudchenko2016balancing,ben2021augmented,arkhangelsky2021synthetic,sun2024temporal}. A growing subset of this literature addresses synthetic control specifically under nonstationarity. \citet{masini2021counterfactual,masini2022counterfactual} show that counterfactual estimation requires a cointegrating relationship between treated and donor units; without it, estimated effects diverge and inference suffers severe size distortion. \citet{harvey2021cointegration} model the common stochastic trend explicitly and propose stationarity tests on the pre-treatment difference as a diagnostic for donor selection. \citet{shi2025synthetic} decompose outcomes into trend and cycle via the Hamilton filter and restrict donor matching to the cyclical residual. These contributions diagnose the nonstationarity problem or propose hard filters that remove persistent variation before matching. HSC introduces the smooth component $E$ as a new degree of freedom; rather than diagnosing nonstationarity or applying a fixed filter, HSC offers a soft, data-driven alternative that retains shared stochastic trend while mitigating idiosyncratic spurious matching.

The estimator's mechanism draws on a different technical toolkit than is standard in the synthetic control literature. The roughness penalty $\|D_q E\|_2^2$ connects HSC to the Whittaker--Henderson smoothing framework \citep{whittaker1922new,henderson1924new}, the Hodrick--Prescott filter \citep{hodrick1997postwar}, and penalized spline formulations \citep{eilers1996flexible}. The spectral decomposition of the HSC metric, in which the tuning parameter $\rho$ acts as a frequency-dependent gain function on the pre-treatment residual, imports ideas from spectral analysis in time series into the synthetic control setting. This connection provides both interpretive clarity ($\rho$ governs a soft spectral partition between what is matched cross-sectionally and what is forecasted univariately) and computational tractability, since the profiled HSC objective reduces to a standard constrained quadratic program.

We illustrate the method on the canonical Hong Kong example, the path of per-capita GDP after the 1997 return of Hong Kong to Chinese sovereignty, studied by \citet{hsiao2012panel} and \citet{shi2025synthetic}. Cross-validation selects an interior allocation rather than either preprocessing extreme, confirming that the data prefer a soft partition between donor matching and the smooth component. HSC distributes donor weight broadly across the control pool, whereas level-matching and filter-based competitors either concentrate weight on a few donors or extrapolate the treated unit's own trend and overshoot the observed series. On a rolling-origin out-of-sample criterion HSC is the most accurate estimator among those we compare, and this ranking is stable across the cross-validation horizon and across two very different donor-selection philosophies. The application thus reproduces, in real data, the soft-allocation behavior that the theory and the Monte Carlo evidence predict.

The remainder of the paper is organized as follows. Section~\ref{sec:tradeoff} formalizes the outcome decomposition, establishes notation, and characterizes the two failure modes (spurious donor matching and over-filtering) that motivate the need for a soft allocation mechanism. Section~\ref{sec:estimator} introduces the HSC estimator, derives its profiled representation, and constructs the forecast operator that extrapolates the smooth component into post-treatment periods. Section~\ref{sec:spectral} develops the spectral interpretation, showing that the tuning parameter $\rho$ acts as a frequency-dependent gain function, and describes the cross-validation procedure for selecting $\rho$. Section~\ref{sec:predict_decompose} presents the prediction-error decomposition into weight-estimation and forecasting terms, and discusses the trade-off under the selection of $\rho$. Section~\ref{sec:simulation} reports Monte Carlo evidence for the HSC estimator. Section~\ref{sec:empirical} applies HSC to the 1997 Hong Kong handover and compares it with established alternatives. Section~\ref{sec:conclusion} concludes.

\section{A Tradeoff between Spurious Donor Matching and Over-Filtering}
\label{sec:tradeoff}

This section formalizes the allocation problem described in the Introduction. We first establish notation for the synthetic control setting, then provide the formal $L + R + \varepsilon$ decomposition and characterize two failure modes: spurious donor matching when the idiosyncratic stochastic trend $R$ dominates, and over-filtering when shared stochastic trend in $L$ is discarded. A simulated illustration shows that existing methods commit to one regime or the other.

\subsection{Setup and notation}
\label{subsec:setup}

We observe a panel of $N_0+1$ units over $T=T_0+T_{\post}$ periods.
Unit $i=1$ is the treated unit; units $i=2,\ldots,N_0+1$ form the
donor pool.  Treatment is imposed at the end of period $T_0$, so the
pre-treatment window is $t=1,\ldots,T_0$ and the post-treatment window
is $t=T_0+1,\ldots,T$.  Let $Y_{it}(0)$ denote the untreated potential
outcome for unit $i$ at time $t$.  Under the standard no-anticipation
assumption, we observe $Y_{1t}=Y_{1t}(0)$ for $t\le T_0$.  Let
$X_t=(Y_{2t},\ldots,Y_{N_0+1,t})'$ denote the
$N_0\times 1$ vector of donor outcomes at time $t$.  We write
$Y_{\pre}=(Y_{1,1},\ldots,Y_{1,T_0})'$ and
$X_{\pre}=(X_1,\ldots,X_{T_0})'$ for the $T_0\times 1$ and
$T_0\times N_0$ pre-treatment arrays.  $Y_{\post}$ and $X_{\post}$ are
defined similarly.

The synthetic control estimator constructs a counterfactual for the
treated unit as a weighted combination of donors.  The weights
$\hat\omega$ are chosen from the simplex
$\Delta=\{\omega\in\mathbb{R}^{N_0}:\omega\ge 0,\;\1'\omega=1\}$ by
minimizing the pre-treatment sum of squared residuals:
\begin{equation*}
\hat\omega \;=\; \argmin_{\omega\in\Delta}\;
\|Y_{\pre}-X_{\pre}\omega\|^2,
\end{equation*}
and the counterfactual at horizon $h\ge 1$ is
$\hat{Y}_{1,T_0+h}(0) = X_{T_0+h}'\hat\omega$.
A common simplex-constrained extension adds an intercept $\hat\alpha$ to absorb a constant
level shift between the treated unit and the weighted donors.  The
weights and intercept are estimated jointly:
\begin{equation*}
(\hat\omega,\hat\alpha) \;=\; \argmin_{\omega\in\Delta,\;\alpha\in\mathbb{R}}\;
\|Y_{\pre}-X_{\pre}\omega - \alpha\1_{T_0}\|^2,
\end{equation*}
which is equivalent to matching on demeaned pre-treatment outcomes
\citep{doudchenko2016balancing,ferman2021synthetic}.  The counterfactual
becomes $\hat{Y}_{1,T_0+h}(0) = X_{T_0+h}'\hat\omega + \hat\alpha$.

Both formulations share the same basic structure: donor weights are
chosen to minimize a pre-treatment loss computed on the observed
outcome series, and the counterfactual extrapolates those weights into
the post-treatment window.

\subsection{Decomposing untreated outcomes and the allocation problem}
\label{subsec:allocation}

We now formalize the decomposition introduced informally in Section~\ref{sec:intro} (Figure~\ref{fig:econometric_problem}). We decompose untreated potential outcomes into three components:
\begin{equation}
Y_{it}(0) \;=\; L_{it} \;+\; R_{it} \;+\; \varepsilon_{it}.
\end{equation}
Here $L_{it} = \Lambda_i'F_t$ is a shared low-rank component, where
$F_t$ collects latent factors and $\Lambda_i$ is the corresponding
vector of unit-specific loadings. The factors $F_t$ may contain
both short-run and stochastic trend movements. The defining property of
$L_{it}$ is that it is \emph{shared}: when a convex combination of donors matches the treated unit's loadings $\Lambda_1$, it reproduces $L_{1t}$ in every period. By contrast, $R_{it}$ and $\varepsilon_{it}$ are
idiosyncratic. The term $R_{it}$ denotes an idiosyncratic stochastic
trend, whose long-run variance grows without bound, whereas $\varepsilon_{it}$ denotes
idiosyncratic short-run noise with bounded long-run variance. Unlike the shared component, neither
$R_{it}$ nor $\varepsilon_{it}$ is governed by a shared factor
structure. As a result, observed outcomes may mask the low-rank
structure in $L_{it}$, with the main difficulty arising from the
idiosyncratic stochastic trend component $R_{it}$.

This decomposition refines the outcome framework of \citet{arkhangelsky2021synthetic}
by separating idiosyncratic stochastic trends from short-run idiosyncratic noise.
In \citeauthor{arkhangelsky2021synthetic}'s notation, untreated outcomes are represented
by a systematic component $\mathcal{L}_{it}$ and an idiosyncratic error $\epsilon_{it}$.
Their asymptotic analysis allows $\mathcal{L}$ to be an approximately low-rank systematic matrix,
without imposing a fixed known rank, while their Assumption~1 restricts the rows
$\epsilon_{i\cdot}$ of the noise matrix to be i.i.d.\ Gaussian vectors with covariance
matrix $\Sigma\in\mathbb R^{T\times T}$ whose eigenvalues are bounded and bounded
away from zero. This condition permits temporal dependence and some forms of
nonstationary covariance heterogeneity. In the growing-$T$ asymptotics
under which Assumption~1 is imposed, however, it excludes integrated
unit-specific stochastic trends: if $\epsilon_{it}$ were a random walk
with innovation variance $\sigma^2$, then
$\operatorname{Cov}(\epsilon_{is},\epsilon_{it})=\sigma^2\min\{s,t\}$,
so the largest eigenvalue of $\Sigma$ grows on the order of $T^2$ and
is not bounded uniformly in $T$, violating the bounded-eigenvalue
requirement of Assumption~1. Our decomposition therefore writes the idiosyncratic component
as $R_{it}+\varepsilon_{it}$, where $R_{it}$ captures the unit-specific stochastic trend
excluded by the bounded-eigenvalue condition and $\varepsilon_{it}$ denotes the
remaining short-run noise. This split makes explicit the type of persistent
idiosyncratic variation that we argue is central to the spurious matching problem.

A parallel restriction appears in \citet{ferman2021synthetic}, who study synthetic control in the fixed-$N$, large-$T_0$ regime under a linear factor model with common factors, unit-specific loadings, and idiosyncratic shocks $\epsilon_{it}$. Their Assumption~4 requires the pre-treatment moments to converge in probability to non-stochastic constants: $T_0^{-1}\sum_t \epsilon_{it}^2 \xrightarrow{p} \sigma_\epsilon^2$, with analogous conditions on the common factors and on cross-products of factors and noise. The substantive content matches \citeauthor{arkhangelsky2021synthetic}'s Assumption~1: the idiosyncratic noise must have bounded long-run variance with well-behaved pre-treatment sample moments. An idiosyncratic random walk in $\epsilon_{it}$ violates Assumption~4 in the same way it violates the bounded-eigenvalue condition.

As discussed in Section~\ref{sec:intro}, the fitting criterion in Section~\ref{subsec:setup} operates on observed outcomes and therefore does not distinguish among $L_{it}$, $R_{it}$, and $\varepsilon_{it}$. The resulting allocation problem, deciding how much stochastic trend variation to attribute to the shared component $L_{it}$ versus the idiosyncratic stochastic trend $R_{it}$, leads to two failure modes formalized in the following subsections.

\subsection{The risk of spurious donor matching}
\label{subsec:spurious}

When $R_{it}$ contributes substantially to $Y_{it}(0)$, the
pre-treatment optimization can achieve a close fit by assigning weight
to donors whose independent persistent movements happen to co-move with
the treated unit over the observed window. This is a manifestation of
spurious regression in the synthetic control setting: the in-sample fit
may appear excellent, yet it is driven by coincidental trending behavior
rather than a shared factor structure, and therefore does not extend
beyond the pre-treatment period. \citet{shi2025synthetic} formalize
this as the \emph{spurious synthetic control} problem, noting that
``even if a country's GDP can be closely approximated by a weighted
average of others over a given period, such a fit may arise purely from
coincidental trending behavior.'' The simplex constraints
$\omega\ge 0$, $\1'\omega=1$ narrow the feasible set but do not prevent
spurious fit: independent random walks can still produce close
pre-treatment matches within the simplex.

It is worth emphasizing that the central failure is \emph{weight
distortion}, not merely imperfect fit.  When $R_{it}$ is quantitatively
important, the optimizer is drawn toward donors whose idiosyncratic
trends happen to track $R_{1t}$ in the pre-treatment period, pulling
the weights away from the combination that would best reproduce the
shared component $L_{it}$. This distinction
between weight distortion and lack of fit is important because several
recent proposals, including augmented synthetic control
\citep{ben2021augmented} and synthetic difference-in-differences
\citep{arkhangelsky2021synthetic}, improve on the standard synthetic
control by augmenting it with an outcome model that corrects for the
residual imbalance left by imperfect pre-treatment fit.  Because these
methods still begin from the donor weights chosen to match pre-treatment
outcome levels, they can inherit the same underlying weight-distortion
channel: the bias-correction step operates conditional on already-distorted
weights, so it can at best mitigate but does not fully undo this distortion.

More fundamentally, the existing literature on synthetic control with
nonstationary data identifies cointegration between the treated unit
and the synthetic control as the key condition separating valid donor matching from spurious fit.  In the regression-based
framework studied by \citet{masini2021counterfactual}, counterfactual
estimation requires the data-generating process to admit a cointegrating
relationship involving the treated unit; without such a relationship,
the regression is spurious.  \citet{masini2022counterfactual} show that
in this case the estimated treatment effect diverges and that ignoring
the nonstationary nature of the data leads to severe over-rejection of
the null hypothesis of no effect.
\citet{harvey2021cointegration} reach a similar conclusion from a
structural time series perspective: they model the shared stochastic
trend explicitly and argue that a synthetic control is valid when the
target and control series share a common stochastic trend, proposing
stationarity tests on the pre-treatment difference as a diagnostic for
donor selection.

In the language of Section~\ref{subsec:allocation}, cointegration between the treated unit and the synthetic control therefore requires that the stochastic trend content of their difference be eliminated not only in the pre-treatment period but also out of sample. When the stochastic trend variation is entirely driven by shared factors in $L_{it}$, a cointegrating relationship among the units is guaranteed by the factor structure, and the donor weights that minimize the pre-treatment criterion can recover the common component both in pre- and post-treatment periods. When the treated unit also contains a quantitatively important idiosyncratic stochastic trend $R_{1t}$, the cointegrating relationship involving the treated unit may not exist, and the critiques from the cited literature then apply directly.

\subsection{The cost of hard filtering}
\label{subsec:filtering}

The previous subsection showed that level-based matching is vulnerable
to spurious donor matching when $R_{it}$ is large. Two natural
responses have been proposed, both of which remove stochastic trend
variation before constructing the synthetic control.

The first is explicit pre-filtering. \citet{shi2025synthetic} propose
decomposing the outcome into a trend and a cyclical component using the
Hamilton filter \citep{hamilton2018you}, forecasting the treated unit's trend from its own
lagged values, and restricting donor matching to the stationary cyclical
residual. Their strategy is designed to eliminate spurious matching from idiosyncratic stochastic trend by
construction under the maintained trend-cycle decomposition: the Hamilton
filter removes the trend component, and donor matching is then applied to
the resulting cyclical residual rather than to the raw persistent series.

The second is differencing. If $R_{it}$ is a random walk,
first-differencing yields
$\Delta Y_{it}=\Delta L_{it}+\Delta R_{it}+\Delta\varepsilon_{it}$,
where $\Delta R_{it}$ is now stationary. This logic underlies the
concern raised by \citet{masini2022counterfactual} that nonstationary
data can produce spurious counterfactual estimates, motivating
transformations such as differencing before applying counterfactual
methods, and is related to
\citeauthor{abadie2021using}'s~(\citeyear{abadie2021using}) observation
that matching in differences can help when levels are not credibly
matchable.

Both strategies address the spurious matching problem, but they share a
common cost: neither distinguishes between $R_{it}$ and the
stochastic trend part of $L_{it}$. If the common factor $F_t$ is a random
walk with heterogeneous loadings $\Lambda_i$, then first-differencing
removes the resulting stochastic trend in $\Lambda_i'F_t$ entirely. \citet{ferman2021synthetic} show that when
common factors include diverging nonstationary components, these
components dominate the pre-treatment fitting criterion as $T_0$ grows,
providing the primary source of identifying variation for donor weight
estimation. Mechanically removing them discards an important signal for donor
matching. What remains after filtering or differencing is a stationary
factor structure with potentially much lower signal-to-noise ratio. The resulting donor weights can become substantially more
imprecise and perform poorly in recovering the shared component of
$L_{it}$. \citet{abadie2021using} indeed notes that matching in first
differences can inflate the variance attributable to noise, potentially
inducing an increase in bias.

\subsection{Illustration of the tradeoff}
\label{subsec:design_goal}

The preceding two subsections identified a tension between two failure
modes: spurious donor matching when $R_{it}$ dominates, and
over-filtering when shared stochastic trend variation in $L_{it}$ is
discarded. The relative severity of these two risks depends on the
magnitude of $R_{it}$ relative to $L_{it}$.

Figure~\ref{fig:tradeoff} illustrates the tradeoff with a simulated
example. The data-generating process is
$Y_{it}(0)=\Lambda_i' F_t+\kappa\cdot R_{it}+\varepsilon_{it}$, where
$F_t$ is a shared random-walk factor with heterogeneous loadings
($\Lambda_i \sim N(1,1)$), $R_{it}$ are idiosyncratic random walks, and $\varepsilon_{it}$ is iid noise with
$\sigma_\varepsilon=2$. The common factor structure is the same across the two rows; the
only difference is the importance of the idiosyncratic stochastic trend,
governed by $\kappa$. The top row sets $\kappa=0$ (only shared stochastic trend); the bottom row sets $\kappa=2$ (an idiosyncratic
stochastic trend added on top of the same common structure). Each column applies a different
estimator. The first is synthetic control in levels with an intercept,
which constructs the counterfactual as
$\hat{Y}_{1,T_0+h}(0) = X_{T_0+h}'\hat\omega^{\mathrm{lev}}
+ \hat\alpha$, where $\hat\alpha$ absorbs a constant shift between the
treated unit and the weighted donors. The second is synthetic control in
first differences, which estimates donor weights on the differenced data
$\Delta Y_{it}$ and anchors the counterfactual at the last
pre-treatment observation:
$\hat{Y}_{1,T_0+h}(0) = X_{T_0+h}'\hat\omega^{\mathrm{dif}}
+ (Y_{1,T_0} - X_{T_0}'\hat\omega^{\mathrm{dif}})$.

\begin{figure}[!h]
\centering
\includegraphics[width=\textwidth]{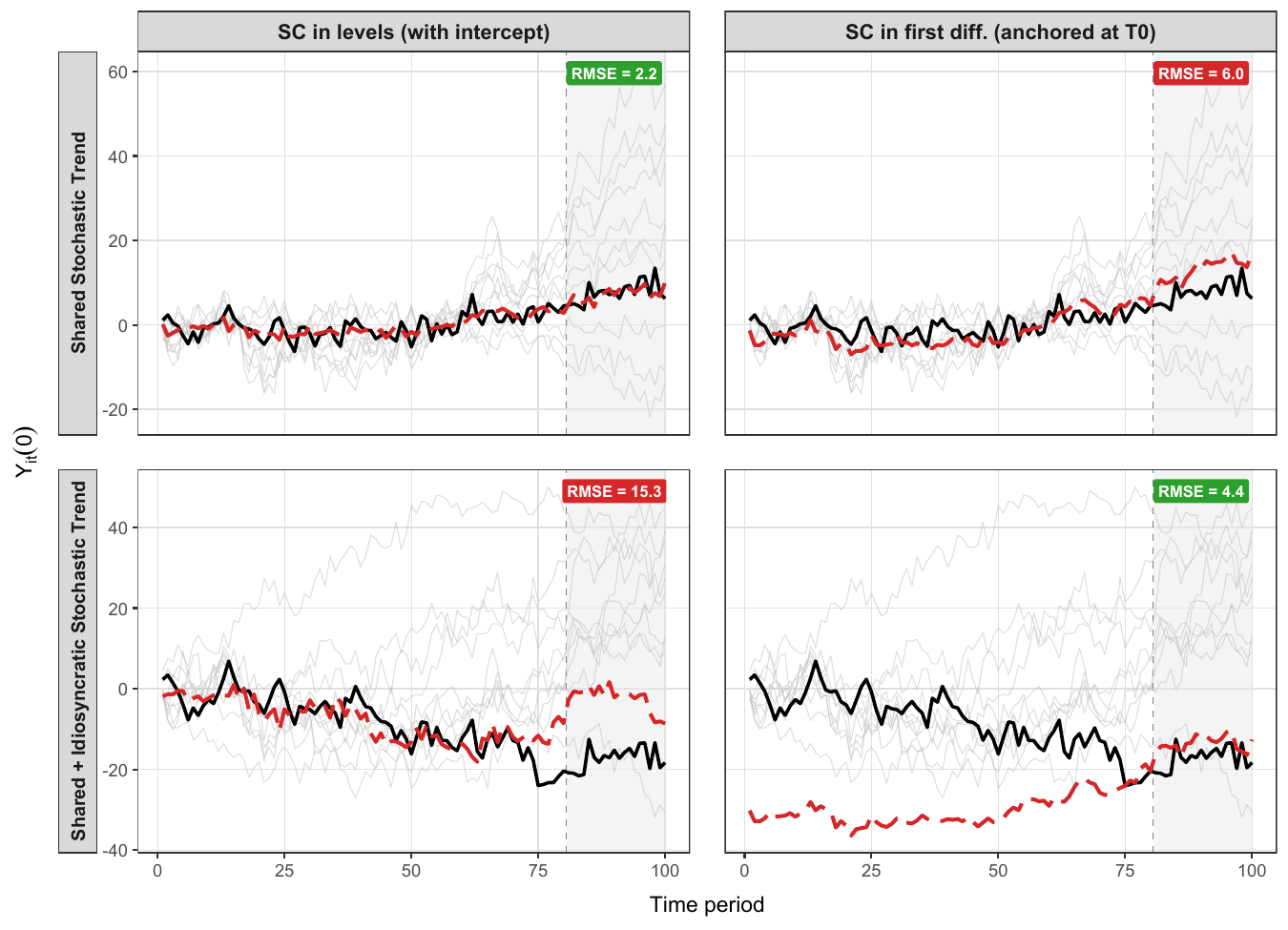}
\caption{Spurious donor matching versus over-filtering.}
\label{fig:tradeoff}

\vspace{4pt}
\begin{minipage}{0.92\textwidth}
\footnotesize
\textit{Notes:} Solid black lines show the treated unit's outcome path
(no treatment effect is imposed). Dashed red lines show the synthetic
control counterfactual. Thin grey lines show individual donor units.
The shaded region marks post-treatment periods. The data-generating
process is
$Y_{it}(0)=\Lambda_i' F_t+\kappa\cdot R_{it}+\varepsilon_{it}$ with one
common random-walk factor, loadings $\Lambda_i\sim N(1,1)$,
$\sigma_\varepsilon=2$, $N_0=10$ donors, $T_0=80$
pre-treatment periods, and $T_{\post}=20$ post-treatment periods. The
top row sets $\kappa=0$ (shared stochastic trend only); the bottom
row adds idiosyncratic stochastic trend ($\kappa=2$). RMSE is computed over
post-treatment periods.
\end{minipage}
\end{figure}

In the top-left panel, synthetic control with intercept captures the common factor
structure and tracks the treated unit closely in both the pre- and
post-treatment periods (RMSE~$=2.2$). In the top-right panel,
first-differencing weakens the dominant common signal; matching on
differenced data in the presence of large stationary noise
($\sigma_\varepsilon=2$) produces weights that recover the common
structure less precisely, and the reconstructed level path has a larger
discrepancy (RMSE~$=6.0$). In the bottom-left panel, synthetic control with intercept shows a large performance discrepancy between pre- and post- treatment periods, suggesting an instance of spurious donor matching
(RMSE~$=15.3$). In the bottom-right panel, first-differencing synthetic
control avoids the spurious match and produces a more accurate
counterfactual (RMSE~$=4.4$).

In practice, the researcher does not know how important the
idiosyncratic stochastic trend component is relative to the shared
component. A method that commits fully to either level matching or
hard filtering will fail in one regime or the other. This motivates the need for a \emph{soft} allocation mechanism rather
than a binary choice. The goal of this paper is to provide a synthetic control estimator that lets the data
determine how much of the stochastic trend variation in the treated unit should be allocated between donor matching and a smooth treated-specific component.

\section{Harmonic Synthetic Control}
\label{sec:estimator}

Section~\ref{sec:tradeoff} showed that synthetic control's pre-treatment fitting criterion does not distinguish between the shared low-rank component $L_{it}$ and the idiosyncratic stochastic trend $R_{it}$; level matching and hard filtering each fail in one regime or the other. To bridge these two regimes, we propose the harmonic synthetic control (HSC) estimator. The construction proceeds in three steps: weight estimation for donor matching, the treated-unit-specific time series forecaster, and the construction of the counterfactual. HSC does not attempt to disentangle $L_{it}$ and $R_{it}$ in the data-generating process. Instead, it seeks the optimal allocation between modeling stochastic trend variation as shared low-rank structure and modeling it as idiosyncratic stochastic trends. A single tuning parameter $\rho \in [0,1]$ governs such an allocation.

\subsection{Donor weights of HSC}\label{subsec:hsc_weights}

We use $q\in\{1,2\}$ to denote the smoothness order. The
difference operator $D_q$ is the $q$th-order difference operator: $D_1$ is the $(T_0-1)\times T_0$ first-difference operator with
rows $(D_1 x)_t = x_{t+1}-x_t$, and $D_2$ is the
$(T_0-2)\times T_0$ second-difference operator with rows
$(D_2 x)_t =x_{t+2}-2x_{t+1}+x_t$.\footnote{Explicitly, $D_1$ has
$-1$ on the main diagonal and $+1$ on the first superdiagonal; $D_2$
has entries $+1, -2, +1$ on three consecutive diagonals.} Let $K_q := D_q'D_q$, which is symmetric positive semidefinite
with null space $\Null(K_q)$ equal to $\mathrm{span}\{\mathbf{1}_{T_0}\}$ for
$q=1$ and $\mathrm{span}\{\mathbf{1}_{T_0}, t_\pre\}$ for
$q=2$, where $t_\pre:=(1,\dots,T_0)'$. Recall that
$X_\pre\in\R^{T_0\times N_0}$ denotes the pre-treatment donor outcome
matrix, $Y_{\pre}\in\R^{T_0}$ is the treated unit's pre-treatment
outcome vector, and
$\Delta_{N_0}:=\{\omega\in\R^{N_0}:\omega\ge 0,\;\mathbf{1}'\omega=1\}$
is the unit simplex.

\begin{definition}[Harmonic Synthetic Control, $\rho\in(0,1)$]
\label{def:hsc_interior}
For $\rho\in(0,1)$, the HSC estimator jointly solves
\begin{equation}
\bigl(\hat\omega(\rho,q),\,\hat E_\pre(\rho,q)\bigr)
\;\in\;
\argmin_{\omega\in\Delta_{N_0},\; E\in\R^{T_0}}
\left\{
\frac{1}{\rho}\bigl\|Y_{\pre}-X_\pre\omega - E\bigr\|_2^2
\;+\;\frac{1}{1-\rho}\|D_q E\|_2^2
\;+\;\zeta^2 T_0\|\omega\|_2^2
\right\},
\label{eq:hsc_primal}
\end{equation}
where $\zeta>0$ is a ridge regularization parameter.\footnote{Following
\citet{arkhangelsky2021synthetic}, we set the default value of $\zeta$
as $T_\post^{1/4}\hat\sigma$, where $\hat\sigma$ is the standard
deviation of all elements of the first-differenced donor matrix
$D_1 X_\pre$. This is the single-treated-unit specialization of the
formula
$\zeta = (N_{\mathrm{tr}}T_\post)^{1/4}\hat\sigma$ used for multiple
treated units.} The first term penalizes the discrepancy between the
treated unit and the donor-weighted combination after removing the
latent smooth component $E$; the second penalizes the roughness of $E$
through the $q$th-difference penalty $\|D_q E\|_2^2 = E'K_q E$; the
third is a ridge penalty that helps stabilize the donor weights.
\end{definition}

$D_q$ determines what counts as a smooth component.
When $q=1$, the penalty
$\|D_1E\|_2^2=\sum_{t=1}^{T_0-1}(E_{t+1}-E_t)^2$ penalizes local
changes in level, so smoother paths are those that vary less from one
period to the next.
Constant vectors, which are in the null space of $K_1$, receive no penalty. When $q=2$, the penalty
$\|D_2E\|_2^2=\sum_{t=1}^{T_0-2}(E_{t+2}-2E_{t+1}+E_t)^2$ penalizes
local changes in slope, so smoother paths are those with less curvature. Intercept-plus-linear-trend components, which are in the null space of $K_2$, receive no penalty. Thus, for $q=1$ the smooth
component is encouraged to be locally flat in levels, whereas for $q=2$
it is encouraged to be locally linear in time.

The tuning parameter $\rho$ governs the relative incentives in the joint
optimization over $(\omega,E)$. Conditional on donor weights $\omega$, the optimizer chooses $E$ to balance fidelity to the residual
$Y_{\pre}-X_\pre\omega$ against the smoothness restriction imposed by
$\|D_qE\|_2^2$. When $\rho$ is small, the fit term
$\frac{1}{\rho}\|Y_{\pre}-X_\pre\omega-E\|_2^2$ receives relatively
large weight compared with the roughness penalty. Therefore, conditional
on a given $\omega$, the optimizer is more willing to let $E$ track a
larger and potentially rougher portion of the residual, where roughness
is measured by $D_q$. When $\rho$ is large, the roughness penalty
$\frac{1}{1-\rho}\|D_qE\|_2^2$ receives relatively large weight, so
conditional on a given $\omega$, the optimizer is forced to choose a
smoother $E$, where smoothness is in the sense measured by $D_q$.

Because $\omega$ and $E$ are chosen jointly, the tuning parameter $\rho$ does not act on a fixed pre-treatment discrepancy. Instead, it determines how the joint optimizer splits the treated series $Y_\pre$ between the donor-matched component $X_\pre\hat\omega$ and the treated-unit-specific smooth component $\hat{E}$, with larger $\rho$ forcing $\hat{E}$ to be smoother.

\subsection{Profiling and the HSC metric}
\label{subsec:profiling}

The joint formulation in Definition~\ref{def:hsc_interior} is useful
for intuition, but the estimator becomes more transparent and
computationally tractable after profiling out the treated-unit-specific
smooth component \(E\). This yields an equivalent weight-estimation
problem in which the pre-treatment residual is measured under a
\(\rho\)- and \(q\)-dependent metric.\footnote{We use the term
``metric'' informally throughout: \(W_{\rho,q}\) is symmetric
positive semidefinite (not positive definite), and it annihilates
\(\Null(K_q)\), so the quadratic form
\(r'W_{\rho,q}r\) is a seminorm rather than a norm.  Residual
components in \(\Null(K_q)\) contribute zero to the HSC criterion
and are handled separately through the smooth component
\(\hat E_\pre\); see the remark after
Proposition~\ref{prop:endpoints}.}

For any candidate donor weight vector
\(\omega\in\Delta_{N_0}\), define the pre-treatment discrepancy/residual as
\begin{equation}
r_\pre(\omega):=Y_{\pre}-X_\pre\omega.
\label{eq:r_pre}
\end{equation}
Fix \(\rho\in(0,1)\). For each \(\omega\), the inner minimization over
\(E\) in~\eqref{eq:hsc_primal} is a strictly convex quadratic program:
\[
\min_{E\in\R^{T_0}}
\left\{
\frac{1}{\rho}\|r_\pre(\omega)-E\|_2^2
+\frac{1}{1-\rho}\|D_qE\|_2^2
\right\}.
\]
To express its solution, let
\begin{equation}
\lambda_\rho:=\frac{\rho}{1-\rho},
\qquad
S_{\rho,q}:=(I_{T_0}+\lambda_\rho K_q)^{-1}.
\label{eq:Srhoq_def}
\end{equation}

\begin{proposition}[Profiled representation]
\label{prop:hsc_profiled}
Fix \(q\in\{1,2\}\) and \(\rho\in(0,1)\). For every
\(\omega\in\Delta_{N_0}\), the inner problem in \(E\) has the unique
minimizer
\begin{equation}
\hat E_\pre(\omega;\rho,q)=S_{\rho,q}\,r_\pre(\omega).
\label{eq:Ehat_profile}
\end{equation}
Substituting this optimizer back into the criterion yields the profiled
objective
\begin{equation}
\hat\omega(\rho,q)
\;\in\;
\argmin_{\omega\in\Delta_{N_0}}
\left\{
r_\pre(\omega)'W_{\rho,q}r_\pre(\omega)
+\zeta^2T_0\|\omega\|_2^2
\right\},
\label{eq:hsc_profiled}
\end{equation}
where
\begin{equation}
W_{\rho,q}:=\frac{1}{\rho}(I_{T_0}-S_{\rho,q}),
\qquad
\rho\in(0,1).
\label{eq:W_rhoq}
\end{equation}
Equivalently, the fitted treated-unit-specific smooth component is
\begin{equation}
\hat E_\pre(\rho,q)
=
S_{\rho,q}\bigl(Y_{\pre}-X_\pre\hat\omega(\rho,q)\bigr).
\label{eq:Ehat_profiled}
\end{equation}
\end{proposition}

The proof is deferred to Appendix~\ref{app:proofs_sec3}. Proposition~\ref{prop:hsc_profiled} shows that donor weight estimation in HSC can be understood through the metric $W_{\rho,q}$, which re-weights the pre-treatment discrepancy $r_\pre$ in a standard ridge-regularized quadratic program on the simplex. The operator $S_{\rho,q}$ acts as a smoother that extracts the smooth part of $\hat r_\pre$ as the smooth component $\hat E_\pre$.

We now show that this family of metrics extends continuously to the boundary cases $\rho=0$ and $\rho=1$. Define
\begin{align}
S_{0,q} &:= I_{T_0}, & S_{1,q} &:= P_{0,q}, \label{eq:S_boundary} \\
W_{0,q} &:= K_q,    & W_{1,q} &:= I_{T_0}-P_{0,q}, \label{eq:W_boundary}
\end{align}
where $P_{0,q}$ denotes the orthogonal projector onto $\Null(K_q)$.

\begin{proposition}[Continuous extension of the HSC metric]
\label{prop:endpoints}
For each fixed \(q\in\{1,2\}\):
\begin{enumerate}[label=(\roman*)]
\item The interior families satisfy
\[
\lim_{\rho\downarrow 0}S_{\rho,q}=I_{T_0},
\qquad
\lim_{\rho\uparrow 1}S_{\rho,q}=P_{0,q},
\]
and
\[
\lim_{\rho\downarrow 0}W_{\rho,q}=K_q,
\qquad
\lim_{\rho\uparrow 1}W_{\rho,q}=I_{T_0}-P_{0,q}.
\]
Hence the boundary definitions~\eqref{eq:S_boundary}
and~\eqref{eq:W_boundary} are the unique continuous extensions of
\(\{S_{\rho,q}\}\) and \(\{W_{\rho,q}\}\) from \((0,1)\) to
\([0,1]\).
\item \(W_{\rho,q}\) is symmetric positive semidefinite for every
\(\rho\in[0,1]\).
\end{enumerate}
\end{proposition}

The proof is deferred to Appendix~\ref{app:proofs_sec3}. The key
implication of Proposition~\ref{prop:endpoints} is that the HSC weight
problem extends continuously from the interior \(\rho\in(0,1)\) to the
endpoint cases \(\rho=0\) and \(\rho=1\), where the two boundary values
correspond to familiar special cases.

At \(\rho=0\), \(S_{0,q}=I_{T_0}\) means the smoother assigns the
entire residual to the treated-unit-specific ``smooth'' component;
correspondingly, \(W_{0,q}=K_q=D_q'D_q\), so the metric measures only the
\(q\)th-order roughness of the residual:
\[
r_\pre(\omega)'W_{0,q}r_\pre(\omega)
=
\|D_qr_\pre(\omega)\|_2^2.
\]
Hence the profiled objective becomes
\begin{equation}
\hat\omega(0,q)
\;\in\;
\argmin_{\omega\in\Delta_{N_0}}
\left\{
\|D_qY_{\pre}-(D_qX_\pre)\omega\|_2^2
+\zeta^2T_0\|\omega\|_2^2
\right\},
\label{eq:endpoint_rho0}
\end{equation}
that is, synthetic control applied to the \(q\)th-differenced outcomes
with a ridge penalty.

At \(\rho=1\), because \(W_{1,q}=I_{T_0}-P_{0,q}\), the profiled
criterion penalizes only the component of the residual orthogonal to
\(\Null(K_q)\). Using the projection identity
\(r'(I_{T_0}-P_{0,q})r=\min_{\gamma}\|r-Z_{0,q}\gamma\|_2^2\), where
\(Z_{0,q}\) is a basis of \(\Null(K_q)\), the profiled objective at
\(\rho=1\) is equivalent to the simultaneous least-squares fit of donor
weights and unregularized null-space coefficients \(\gamma\). Thus, for
\(q=1\), since \(\Null(K_1)=\mathrm{span}\{\mathbf 1_{T_0}\}\), the
endpoint estimator is synthetic control with an intercept; for
\(q=2\), since \(\Null(K_2)=\mathrm{span}\{\mathbf 1_{T_0},t_\pre\}\),
the endpoint estimator is synthetic control with an intercept and a
linear trend in time.

By Proposition~\ref{prop:endpoints}, null-space components of the residual are not tuned by $\rho$. For every $v \in \Null(K_q)$ and every $\rho \in (0,1)$, $(I_{T_0} + \lambda_\rho K_q)v = v$, so $S_{\rho,q}v = v$ and $W_{\rho,q}v = 0$; the same identities hold at $\rho \in \{0,1\}$ by definitions. Thus, a constant component of the residual (for $q=1$) or any intercept-plus-linear-trend component (for $q=2$) is always absorbed entirely into $\hat E_\pre(\rho,q)$ and contributes nothing to the donor-matching criterion. The tuning parameter $\rho$ reallocates only the residual components outside $\Null(K_q)$ between donor matching and the smooth component; null-space components are always assigned to the smooth branch.

With these properties at $\rho=0$ and $\rho=1$, we can extend the
definition of the HSC weight estimation to $\rho \in [0,1]$.

\begin{definition}[Harmonic Synthetic Control,
\texorpdfstring{$\rho\in[0,1]$}{rho in [0,1]}]
\label{def:hsc}
For \(\rho\in[0,1]\) and \(q\in\{1,2\}\), define the HSC weight
estimator by
\begin{equation}
\hat\omega(\rho,q)
\;\in\;
\argmin_{\omega\in\Delta_{N_0}}
\left\{
r_\pre(\omega)'W_{\rho,q}r_\pre(\omega)
+\zeta^2T_0\|\omega\|_2^2
\right\},
\label{eq:hsc_extended}
\end{equation}
where \(W_{\rho,q}\) is the continuously extended metric family from
Proposition~\ref{prop:endpoints}. The associated treated-unit-specific
smooth component is
\begin{equation}
\hat E_\pre(\rho,q):=S_{\rho,q}\,\hat r_\pre(\rho,q),
\qquad
\hat r_\pre(\rho,q):=Y_{\pre}-X_\pre\hat\omega(\rho,q),
\label{eq:hsc_E_extended}
\end{equation}
with \(S_{\rho,q}\) continuously extended to \([0,1]\) as in
Proposition~\ref{prop:endpoints}.
\end{definition}

This profiled Definition~\ref{def:hsc} is the HSC weight estimator used throughout the remainder of the paper. Section~\ref{sec:spectral} discusses the interpretation of this metric from the spectral perspective and shows how $W_{\rho,q}$ amplifies or downweights components at different frequencies.

\subsection{Forecast operator and the HSC counterfactual}
\label{subsec:counterfactual}

The previous subsection develops the HSC weight estimator $\hat\omega(\rho,q)$ and shows that it arises from minimizing the pre-treatment residual norm under the metric $W_{\rho,q}$. However, the donor weights alone do not define the full counterfactual. The profiled representation produces a smooth component $\hat E_\pre(\rho,q)$, which captures the smooth part of the donor matching residual. To construct a post-treatment counterfactual, this smooth component is extrapolated to post-treatment periods. This term serves a similar function to the bias-correction term in augmented synthetic control \citep{ben2021augmented}. 

To see why a forecasting step is needed, consider the decomposition implied by HSC. For each $(\rho,q)$, the pre-treatment residual
$\hat r_\pre(\rho,q):=Y_{\pre}-X_\pre\hat\omega(\rho,q)$ is split by
the smoother $S_{\rho,q}$ into two components:
\[
\hat r_\pre(\rho,q)
=
\underbrace{S_{\rho,q}\hat r_\pre(\rho,q)}_{\displaystyle\hat E_\pre(\rho,q)}
\;+\;
\underbrace{(I_{T_0}-S_{\rho,q})\hat r_\pre(\rho,q)}_{\displaystyle\hat u_\pre(\rho,q)}.
\]
The first term, $\hat E_\pre(\rho,q)$, is the treated-unit-specific smooth component: the smooth part of the residual that cannot be explained by the donors. The second term, $\hat u_\pre(\rho,q)$, is the rough remainder, containing the higher-roughness components of the residual. Note that any component of the residual in $\Null(K_q)$ (a constant for $q=1$, or an intercept-plus-linear-trend component for $q=2$) is contained in $\hat E_\pre(\rho,q)$.

The treated unit's pre-treatment outcome therefore admits the three-part decomposition. 
\begin{equation}
Y_{\pre}
=
\underbrace{X_\pre\hat\omega(\rho,q)}_{\text{donor-matched}}
\;+\;
\underbrace{\hat E_\pre(\rho,q)}_{\text{smooth}}
\;+\;
\underbrace{\hat u_\pre(\rho,q)}_{\text{rough remainder}}.
\label{eq:nested_decomp}
\end{equation}
In the post-treatment period, the donor-matched component extends to $X_\post\hat\omega(\rho,q)$ using observed donor outcomes. The rough remainder is the part of the pre-treatment residual that is suppressed by the smoother $S_{\rho,q}$; it is treated as noise and is not carried forward. The smooth component, however, represents a treated-unit-specific component that, if present in the pre-treatment period, is likely to persist. Discarding it would leave a predictable bias in the counterfactual. This motivates the introduction of a forecast operator to extrapolate $\hat E_\pre(\rho,q)$ into post-treatment periods.

Let $G_q\in\R^{T_\post\times T_0}$ be a deterministic linear forecast
operator. The HSC counterfactual is defined as
\begin{equation}
\hat Y_\post(0;\rho,q)
:=
X_\post\hat\omega(\rho,q)+G_q\hat E_\pre(\rho,q).
\label{eq:hsc_counterfactual}
\end{equation}

The counterfactual thus combines two branches: a donor-matching component, $X_\post\hat\omega(\rho,q)$, that uses cross-sectional information to predict the treated unit's counterfactual, and a time series forecasting component, $G_q\hat E_\pre(\rho,q)$, that uses temporal information in the pre-treatment residual to extrapolate the smooth component. Writing
$\Pi_{\rho,q}:=G_qS_{\rho,q}\in\R^{T_\post\times T_0}$ for the
composed forecast-smoother, the counterfactual can also be expressed as
\begin{equation}
\hat Y_\post(0;\rho,q)
=
X_\post\hat\omega(\rho,q)+\Pi_{\rho,q}\hat r_\pre(\rho,q).
\label{eq:hsc_counterfactual_alt}
\end{equation}

The forecast operator $G_q$ is treated as a fixed linear map throughout
the analysis of donor weights and counterfactual construction. In practice, $G_q$ may be estimated
from the pre-treatment data.
We impose one structural requirement on $G_q$.

\begin{definition}[Admissible forecast operator]
\label{def:Gq}
A linear operator $G_q:\R^{T_0}\to\R^{T_\post}$ is an
\emph{admissible forecast operator of order $q$} if, for every
polynomial $p$ of degree less than $q$,
\begin{equation}
G_q\,\bigl(p(t)\bigr)_{t=1}^{T_0}
\;=\;
\bigl(p(t)\bigr)_{t=T_0+1}^{T_0+T_\post}.
\label{eq:Gq_continuation}
\end{equation}
\end{definition}

For $q=1$, the requirement reduces to
$G_1\,\mathbf{1}_{T_0} = \mathbf{1}_{T_\post}$: constants are
continued as constants. For $q=2$, it adds
$G_2\,(1,2,\ldots,T_0)' = (T_0+1,\ldots,T_0+T_\post)'$: linear
trends are continued as linear trends. These are the only cases
we use in what follows.

Definition~\ref{def:Gq} ensures that the forecast operator
extrapolates the components which the roughness penalty leaves
unpenalized, namely those in $\Null(K_q)$, to the post-treatment
period without distortion. The components that pass mechanically
into $\hat E$ (constants for $q=1$; constants and linear trends
for $q=2$) are thereby continued unchanged. Since $S_{\rho,q}$
also leaves $\Null(K_q)$ untouched for every $\rho\in[0,1]$, the
composed operator $\Pi_{\rho,q} = G_q S_{\rho,q}$ inherits the
same continuation property whenever $G_q$ is admissible.

Definition~\ref{def:Gq} is a design requirement on the forecast
operator, not a property that generic forecasters automatically
satisfy. The simplest admissible operators are constant
extrapolation (for $q=1$) and linear extrapolation (for $q=2$).
Constant extrapolation defines
$(G_1^{\mathrm{const}}x)_h := x_{T_0}$ for $h=1,\dots,T_\post$:
each post-treatment period receives the last pre-treatment
value. Then
$G_1^{\mathrm{const}}\mathbf{1}_{T_0}=\mathbf{1}_{T_\post}$.
Linear extrapolation defines
$(G_2^{\mathrm{lin}}x)_h := x_{T_0}+h\cdot(x_{T_0}-x_{T_0-1})$
for $h=1,\dots,T_\post$. Then
$G_2^{\mathrm{lin}}\mathbf{1}_{T_0}=\mathbf{1}_{T_\post}$ and
$G_2^{\mathrm{lin}}t_\pre=t_\post$.

In practice, one may wish to use a richer forecasting model, for
example an autoregressive or ARIMA-type specification, to
extrapolate $\hat E_\pre(\rho,q)$. Because such procedures
involve parameter estimation, the resulting forecast operator
$\hat G_q$ is data-dependent and need not be admissible on its
own. A convenient way to enforce the requirement is to separate
the null-space and non-null-space components of the input.

Recall that $P_{0,q}$ denotes the orthogonal projector onto
$\Null(K_q)$, and $P_{\perp,q} := I_{T_0} - P_{0,q}$.
Definition~\ref{def:Gq} constrains a forecast operator only
through its action on $\Null(K_q)$. Let $G_q^{\mathrm{null}}$
denote the unique admissible action on the null space: it has
the property that, for $v\in\Null(K_q)$ equal to
$\bigl(p(t)\bigr)_{t=1}^{T_0}$ for the unique polynomial $p$ of
degree less than $q$,
$G_q^{\mathrm{null}} v = \bigl(p(t)\bigr)_{t=T_0+1}^{T_0+T_\post}$. Construct the forecast operator
\[
\widetilde G_q
:=
G_q^{\mathrm{null}} P_{0,q}
+
\hat G_q P_{\perp,q}.
\]
For every $z \in \Null(K_q)$,
$\widetilde G_q z = G_q^{\mathrm{null}} z$, so
Definition~\ref{def:Gq} holds by construction. For every
$z \in \Null(K_q)^\perp$, $\widetilde G_q z = \hat G_q z$, so the
data-driven forecaster retains full flexibility on the
$\Null(K_q)^{\perp}$ component. For components in $\Null(K_q)$, $G_q^{\mathrm{null}}$ simply
lets the constant ($q=1$) or the linear trend ($q=2$) persist
into the post-treatment periods. These are the carry-forward
$G_1^{\mathrm{const}}$ and the two-point linear extension
$G_2^{\mathrm{lin}}$ defined above, respectively.

The tuning parameter $\rho$ determines the content of the smooth
component $\hat E_\pre(\rho,q)$ and therefore what the forecast
operator $G_q$ is asked to extrapolate. At $\rho=1$, the smoother
$S_{1,q} = P_{0,q}$ extracts only the null-space component of
the residual, and $G_q$ continues these unpenalized components
into the post-treatment period. For interior values
$\rho \in (0,1)$, the smooth component absorbs additional
variation beyond the null space, and $G_q$ extrapolates these
additional components. At $\rho=0$, the smoother
$S_{0,q} = I_{T_0}$ assigns the entire residual to
$\hat E_\pre$, so $G_q$ carries the full pre-treatment
discrepancy forward. The choice of $\rho$ thus calibrates the
burden on $G_q$: minimal at $\rho=1$, where extrapolation
amounts to leaving null-space components unchanged, and maximal
at $\rho=0$, where the entire residual must be forecast.

\section{Spectral Interpretation and Tuning}
\label{sec:spectral}

The previous section defined the HSC estimator and derived its profiled representation. We now develop the complementary perspective on this allocation. The metric $W_{\rho,q}$ acts as a frequency-dependent gain function, so $\rho$ determines which frequency components of $Y_\pre$ enter the donor-matching criterion and which are diverted to the time series forecaster. We use cross-validation to select $\rho$ by minimizing out-of-sample prediction error. We then illustrate the resulting adaptation using the simulated data introduced in Section~\ref{sec:tradeoff}.

\subsection{Spectral decomposition of the HSC metric}\label{subsec:spectral}

The previous section showed that $\rho$ controls a smooth
allocation of $Y_\pre$ between donor matching and the
time series forecaster. We now examine the complementary
question: how does $\rho$ shape the donor-matching criterion?
The spectral decomposition of $W_{\rho,q}$ provides a precise
answer and reveals that the key difference between HSC and
existing synthetic control methods lies in how each method
weights different frequency components in the optimization
criterion.

Since the penalty matrix $K_q = D_q^{\top} D_q$ is symmetric
positive semidefinite, $\mathbb{R}^{T_0}$ admits an orthonormal
basis of eigenvectors of $K_q$, denoted
$v_{1,q}, \ldots, v_{T_0,q}$, with corresponding eigenvalues
$0 \leq \mu_{1,q} \leq \cdots \leq \mu_{T_0,q}$. Any vector of
length $T_0$ can be written as a linear combination of these
eigenvectors. Each eigenvalue $\mu_{j,q}$ measures the
roughness of the corresponding basis function as assessed by
the $q$th-difference penalty:
$\mu_{j,q} = v_{j,q}' K_q v_{j,q} = \|D_q v_{j,q}\|_2^2$. The
first $q$ eigenvectors span the null space $\Null(K_q)$,
corresponding to constants when $q=1$ and constants and linear
trends when $q=2$, with $\mu_{j,q}=0$. As the index $j$
increases beyond $q$, the eigenvectors oscillate progressively
more rapidly and $\mu_{j,q}$ grows.\footnote{For $q=1$, the
eigenvalues admit the closed form
$\mu_{j,1} = 4\sin^2\!\bigl((j-1)\pi/(2T_0)\bigr)$ for
$j=1,\dots,T_0$. For $q=2$, the eigenvalue spectrum grows more
steeply, as illustrated in Figure~\ref{fig:spectral}(a).}
Small $\mu_{j,q}$ corresponds to low-frequency variation, such
as slow-moving trends and long cycles, while large $\mu_{j,q}$
corresponds to high-frequency variation, including rapid,
short-run oscillations.

For a residual vector $r_{\pre}(\omega)$, define its spectral
coordinates by $\tilde r_j(\omega) := v_{j,q}' r_{\pre}(\omega)$,
the coefficient on the $j$-th eigenvector. One can thus write
$r_{\pre}(\omega) = \sum_{j=1}^{T_0} \tilde r_j(\omega) v_{j,q}$.
In these coordinates, the smoother $S_{\rho,q}$ and the profiled
metric $W_{\rho,q}$ each act componentwise through scalar
functions of the eigenvalue $\mu$:
\begin{equation}
S_{\rho,q} r_{\pre}(\omega)
=
\sum_{j=1}^{T_0}
\underbrace{\frac{1-\rho}{(1-\rho)+\rho\mu_{j,q}}}_{=:~s_q(\mu_{j,q};\,\rho)}
\tilde r_j(\omega)\,v_{j,q},
\label{eq:hsc_spectral_s}
\end{equation}
\begin{equation}
r_{\pre}(\omega)' W_{\rho,q} r_{\pre}(\omega)
=
\sum_{j=1}^{T_0}
\underbrace{\frac{\mu_{j,q}}{(1-\rho)+\rho\mu_{j,q}}}_{=:~w_q(\mu_{j,q};\,\rho)}
\tilde r_j(\omega)^2.
\label{eq:hsc_spectral_w}
\end{equation}

The shrinkage function $s_q(\mu;\rho)$ is decreasing in $\mu$:
smoother components (small $\mu$) in $r_{\pre}(\omega)$ are less
shrunk toward zero than rougher components (larger $\mu$).
Components in $\Null(K_q)$ ($\mu=0$) survive intact. The weight
function $w_q(\mu;\rho)$ is increasing in $\mu$: rougher
components in $r_{\pre}(\omega)$ receive more emphasis in the
donor-matching problem. Components in $\Null(K_q)$ ($\mu=0$)
are completely excluded from the donor-matching problem. From
this spectral perspective, HSC therefore routes low-frequency
discrepancy primarily to the time series branch and
high-frequency discrepancy primarily to the donor-matching
branch.

Figure~\ref{fig:spectral} displays these two functions and their
consequences for the simulated pre-treatment series introduced in Section~\ref{sec:tradeoff}, with $q=1$ on
the left and $q=2$ on the right. We highlight three
features of the spectral decomposition.

First, Panel~(a) plots the shrinkage function $s_q(\mu;\rho)$
for five values of $\rho$. This function determines how much of
each spectral coordinate of the pre-treatment residual is
retained in the smooth component $\hat E_\pre(\rho,q)$. For
$\rho \in (0,1)$, each curve equals one at $\mu=0$ (null-space
components are fully retained) and decays smoothly toward zero
as $\mu$ increases, with smaller $\rho$ producing slower decay.
At $\rho=0$, the shrinkage function equals one everywhere: the
entire residual is assigned to $\hat E_\pre$ and extrapolated by
$G_q$. At $\rho=1$, the function collapses to the null-space
indicator, $s_q(\mu;1) = \1\{\mu=0\}$. Intermediate $\rho$ values interpolate smoothly
between these extremes.

Second, Panel~(b) plots the complementary weight function
$w_q(\mu;\rho)$, which determines how much of each spectral
coordinate enters the donor-matching
criterion~\eqref{eq:hsc_profiled}. Every curve passes through
$(\mu,w) = (1,1)$: components with $\mu > 1$ are amplified
relative to uniform weighting, components with $\mu < 1$ are
shrunk. At $\rho=0$, $w_q(\mu;0) = \mu$ (the 45-degree line), so the profiled metric reduces to $K_q$ and
donors are matched entirely in $q$th differences. At $\rho=1$,
$w_q(\mu;1) = \1\{\mu>0\}$ (the horizontal line), so all
components outside $\Null(K_q)$ are retained intact in the donor-matching criterion, while the
null-space directions are absorbed by $\hat E_\pre$. The $\rho=1$ endpoint thus coincides
with SC with intercept at $q=1$ and
with SC with intercept-plus-linear-trend at $q=2$.

Third, the eigenvalue ranges of $K_1$ and $K_2$ differ
substantially. For $q=1$, the eigenvalues of $K_1$ lie in
$[0,4]$ when $T_0=80$, whereas for $q=2$ the eigenvalues of
$K_2$ span a much wider range (up to roughly $16$). Under the
same $\rho$, the wider spectrum therefore yields more extreme
amplification of rough components in the sense of $K_2$ than in
the sense of $K_1$ in the donor-matching criterion.

\begin{figure}[!h]
\centering
\includegraphics[width=0.9\textwidth]{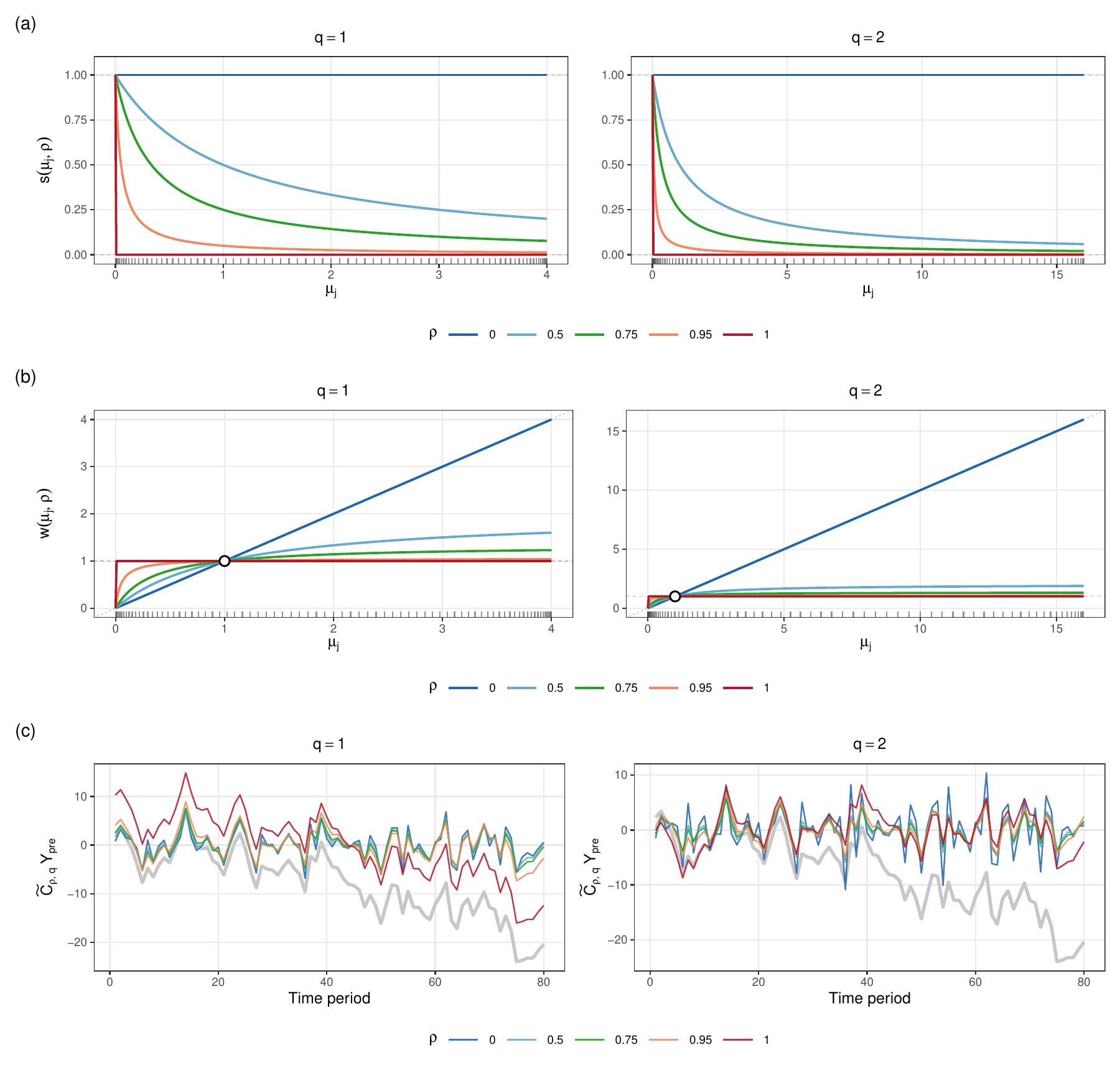}
\caption{Spectral interpretation of the HSC operators $S_{\rho,q}$ and
$W_{\rho,q}$.}
\label{fig:spectral}
\vspace{0.5em}
\begin{minipage}[]{\textwidth}
{\footnotesize\textbf{Note:} This figure illustrates the spectral
reweighting mechanism of HSC with $T_0=80$. Panel~(a) plots the
shrinkage function
$s_q(\mu;\rho)=(1-\rho)/\bigl((1-\rho)+\rho\mu\bigr)$, which
determines how much of each spectral coordinate is retained in the
smooth component $\hat E_\pre$. Panel~(b) plots the weight
function $w_q(\mu;\rho)=\mu/\bigl((1-\rho)+\rho\mu\bigr)$, which
determines how much of each spectral coordinate enters the
donor-matching criterion. Both panels display five values
$\rho\in\{0,\,0.5,\,0.75,\,0.95,\,1\}$, with $q=1$ (left) and $q=2$
(right). Rug marks along the horizontal axis show the positive
eigenvalues of $K_q$. In panel~(b), the open circle marks the fixed
point $(\mu,w)=(1,1)$ through which all curves pass; the dashed line at
$w=1$ corresponds to null-space-projected weighting ($\rho=1$); the dotted
45-degree line corresponds to pure differencing ($\rho=0$). Panel~(c)
displays the spectrally transformed pre-treatment series
$\widetilde C_{\rho,q}\,Y_\pre$ for the same five values of $\rho$,
with $q=1$ (left) and $q=2$ (right). The grey line shows the raw
pre-treatment series $Y_\pre$ of the treated unit from the
shared + idiosyncratic stochastic trend regime ($\kappa=2$) of Figure~\ref{fig:tradeoff}. At $\rho=1$, the
transformation removes the null-space component (mean for $q=1$; mean
and linear trend for $q=2$) and weights all remaining directions
equally. At $\rho=0$, the transformation applies $K_q^{1/2}$,
amplifying high-roughness directions; this effect is particularly
visible for $q=2$, where the wider eigenvalue range of $K_2$ produces a
noticeably rougher transformed series. Intermediate values of $\rho$
blend these effects, progressively attenuating trends while preserving
moderate-scale variation.\par}
\end{minipage}
\end{figure}

Because $W_{\rho,q}$ is symmetric positive semidefinite, it
admits a symmetric square root
$\widetilde C_{\rho,q} := W_{\rho,q}^{1/2}
= V_q \diag\!\bigl(\sqrt{w_q(\mu_{j,q};\rho)}\bigr) V_q'$,
where $V_q := (v_{1,q}, \ldots, v_{T_0,q})$ is the orthogonal
matrix of eigenvectors introduced above. The donor-matching problem
can equivalently be written as
\[
\hat\omega(\rho,q)
\in
\arg\min_{\omega \in \Delta_{N_0}}
\left\{
\|\widetilde C_{\rho,q}\,r_{\pre}(\omega)\|_2^2
+ \zeta^2 T_0 \|\omega\|_2^2
\right\}.
\]
The transformed series $\widetilde C_{\rho,q}\,Y_\pre$ makes the
effect of $\rho$ on the data visible. Panel~(c) of
Figure~\ref{fig:spectral} plots this transformation for the
treated unit from the shared + idiosyncratic stochastic trend regime ($\kappa=2$) of
Figure~\ref{fig:tradeoff}, at five values of $\rho$.

In the $q=1$ panel, at $\rho=1$, the transformation projects out the constant direction and retains all other components with
equal weight; the result resembles a demeaned version of
$Y_\pre$. At $\rho=0$, the transformation reduces to $K_1^{1/2}$,
which has the same quadratic form as the first-difference operator
($\|K_1^{1/2} r\|_2^2 = \|D_1 r\|_2^2$): slow trends
are strongly attenuated and period-to-period oscillations are
amplified. At $\rho \in \{0.5, 0.75, 0.95\}$, low-frequency
components are progressively attenuated while high-frequency
components are increasingly amplified. The $q=2$ panel exhibits the same pattern with sharper contrast.
At $\rho=1$, the transformation projects out both the constant
and linear-trend directions. At
$\rho=0$, it applies $K_2^{1/2}$, which has the same quadratic
form as the second-difference operator: the transformed series is visibly
rougher than the raw data because rapid oscillations that
contribute little to the original amplitude are greatly
amplified.

The name ``harmonic'' reflects the spectral structure of HSC.
The weight $w_q(\mu;\rho) = \mu/\bigl((1-\rho) + \rho\mu\bigr)$
admits the decomposition
\[
\frac{1}{w_q(\mu;\rho)}
\;=\;
(1-\rho)\cdot\frac{1}{\mu} \;+\; \rho\cdot 1,
\]
which identifies $w_q(\mu;\rho)$ as the weighted harmonic mean
of $\mu$ and $1$ with weights $(1-\rho)$ and $\rho$. The two
endpoints recover the arguments themselves: $w_q(\mu;0) = \mu$
and $w_q(\mu;1) = \1\{\mu>0\}$. This structure
mirrors the coefficients $1/\rho$ and $1/(1-\rho)$ on the
smoothness and residual terms in the primal
objective~\eqref{eq:hsc_primal}. Because the harmonic mean is
dominated by the smaller of its arguments, $w_q(\mu;\rho)$ is
small whenever either $\mu$ is small (a low-frequency,
near-null-space component) or $\rho$ is small (the user has
shifted emphasis toward $\hat E_\pre$), giving the metric its
characteristic soft cutoff.

\subsection{Selection of $\rho$}\label{subsec:rho_selection}

The central tuning decision in HSC is how to allocate
low-frequency components of the pre-treatment discrepancy
between donor matching and the smooth component $\hat E_\pre$.
We select $\rho$ by rolling-origin cross-validation using
\emph{only pre-treatment data of the treated unit}:\footnote{Each unit may carry its own idiosyncratic stochastic
trend, so the $\rho$ that minimizes prediction error on a
control unit need not be appropriate for the treated unit. We
therefore restrict cross-validation to the treated unit's own
pre-treatment history.} at each fold, HSC is fitted on a
shortened training window and used to predict held-out
pre-treatment outcomes, and the $\rho$ minimizing average
prediction error is selected. This mimics the forecasting task
the estimator faces at the treatment date.

The cross-validation procedure requires the researcher to
specify four inputs: the smoothness order $q \in \{1,2\}$, the
forecast operator $G_q$, a forecast horizon $h \ge 1$, and a
number of folds $L \ge 1$. The rolling forecast origins are
then $k_\ell = T_0 - h - L + \ell$ for $\ell = 1, \ldots, L$;
by construction, $k_L + h = T_0$, so all validation windows lie
within the pre-treatment period and no post-treatment data are
used. At each origin $k_\ell$, the training window is
$\{1,\dots,k_\ell\}$ and the validation window is
$\{k_\ell+1,\dots,k_\ell+h\}$. For each candidate $\rho$ and
each origin $k_\ell$, the procedure (i)~re-estimates the HSC
weights $\hat\omega(\rho,q)$ and the smooth component
$\hat E_\pre(\rho,q)$ on the training window, (ii)~refits any
data-dependent parameters of $G_q$ on the training window (for
instance, the coefficients of an ARIMA specification),
(iii)~constructs the counterfactual forecast
$\hat Y_{k_\ell+s}(0;\rho,q)$ for the validation window by
combining the donor-matched component with the
$G_q$-extrapolated smooth component, and (iv)~computes the
squared prediction errors against the actual outcomes
$Y_{k_\ell+1},\dots,Y_{k_\ell+h}$. Because all quantities are
re-estimated within each fold, the validation error is an
out-of-sample measure of the joint performance of donor
matching and time series extrapolation. The cross-validation
criterion is
\begin{equation}
\mathrm{CV}(\rho)
\;=\;
\frac{1}{Lh}\sum_{\ell=1}^{L}\sum_{s=1}^{h}
\bigl(Y_{k_\ell+s} - \hat Y_{k_\ell+s}(0;\rho,q)\bigr)^2,
\label{eq:cv}
\end{equation}
where $\hat Y_{k_\ell+s}(0;\rho,q)$ denotes the HSC
counterfactual prediction for period $k_\ell+s$ constructed
from the training window $\{1,\dots,k_\ell\}$.

One practical consideration in constructing the candidate grid
is worth noting. The equivalent smoothing parameter
$\lambda_\rho = \rho/(1-\rho)$ is a convex function of $\rho$
that increases slowly near $\rho=0$ but diverges as
$\rho \uparrow 1$. Consequently, a uniformly spaced grid in
$\rho$ induces a nonuniform grid in $\lambda_\rho$, with
increasingly coarse coverage on the $\lambda_\rho$ scale near
$\rho=1$. When the cross-validation curve exhibits dramatic change near $\rho=1$, a convenient alternative is to construct the
grid on a logarithmic scale in $\lambda_\rho$ and map back via
$\rho = \lambda_\rho/(1+\lambda_\rho)$, while retaining the
boundary $\rho=0$ and $\rho=1$ explicitly.

Beyond the choice of $\rho$, the same procedure can also be
used to select $(\rho, q)$ jointly by minimizing
$\mathrm{CV}(\rho, q)$ over the pair. If multiple forecast
operators are under consideration, the grid extends further to
$(\rho, q, G_q)$ triples, comparing, for example, constant
extrapolation against an ARIMA specification for each
$(\rho, q)$ combination. Joint selection remains computationally
inexpensive because the quadratic program for each
$(\rho, q, k_\ell)$ combination is fast to
solve.\footnote{Our implementation uses Gurobi.}

\subsection{An illustrative example}
\label{subsec:hsc_example}

We illustrate HSC on the two simulated regimes from
Section~\ref{sec:tradeoff}. Two configurations vary both the
smoothness order $q$ and the forecast operator $G_q$ to show
the roles of all three tuning choices ($\rho$, $q$, and $G_q$).
In both configurations, $\rho$ is selected by the
rolling-origin cross-validation procedure of
Section~\ref{subsec:rho_selection} with one-step-ahead horizon
($h=1$), 10 folds ($L=10$), and a grid of 21 equally spaced
values in $[0,1]$.

\paragraph*{Configuration 1: $q=1$, constant extrapolation.}
The first configuration sets $q=1$ and uses the
constant-extrapolation operator $G_1^{\mathrm{const}}$, the
most conservative admissible forecaster: each post-treatment
period receives the last pre-treatment value of $\hat E_\pre$.

Figure~\ref{fig:hsc_example_d1} displays the results. Each row
corresponds to one regime: shared stochastic trend ($\kappa=0$,
top) and shared + idiosyncratic stochastic trend ($\kappa=2$,
bottom). The three columns show the cross-validation curve,
the counterfactual fit, and the decomposition of the fitted
counterfactual into its donor-matched and time series
components.

In the regime with only shared stochastic trend,
cross-validation selects $\hat\rho=1$, the
endpoint corresponding to SC with intercept. This is consistent with the
spectral interpretation of Section~\ref{subsec:spectral}: when
the stochastic trend factor is shared across all units, the
donor pool can reproduce the treated unit's low-frequency
dynamics, so HSC lets the donor-matching branch carry most of
the load in modeling the treated series. The decomposition panel
confirms this allocation: the donor component
$X_\post\hat\omega$ (green) almost coincides with the
counterfactual $\hat Y_\post(0)$ (blue), while the extrapolated
smooth component $G_1^{\mathrm{const}}\hat E_\pre$ (red) is a
flat horizontal line.

In the regime with both shared and idiosyncratic stochastic
trends, cross-validation selects $\hat\rho=0$. The donor
pool can no longer reproduce the treated unit's idiosyncratic
stochastic trend, so cross-validation finds that matching
entirely in first differences, with the time series branch
carrying more of the prediction, yields the best out-of-sample
predictions. The decomposition panel reflects this: the
extrapolated smooth component $G_1^{\mathrm{const}}\hat E_\pre$
now carries a large level correction that accounts for the
drift accumulated up to $T_0$. In this regime, HSC
reduces to ridge-regularized synthetic control in first
differences, anchored at $T_0$, as shown in Section~\ref{sec:tradeoff}.

\begin{figure}[!h]
\centering
\includegraphics[width=\textwidth]{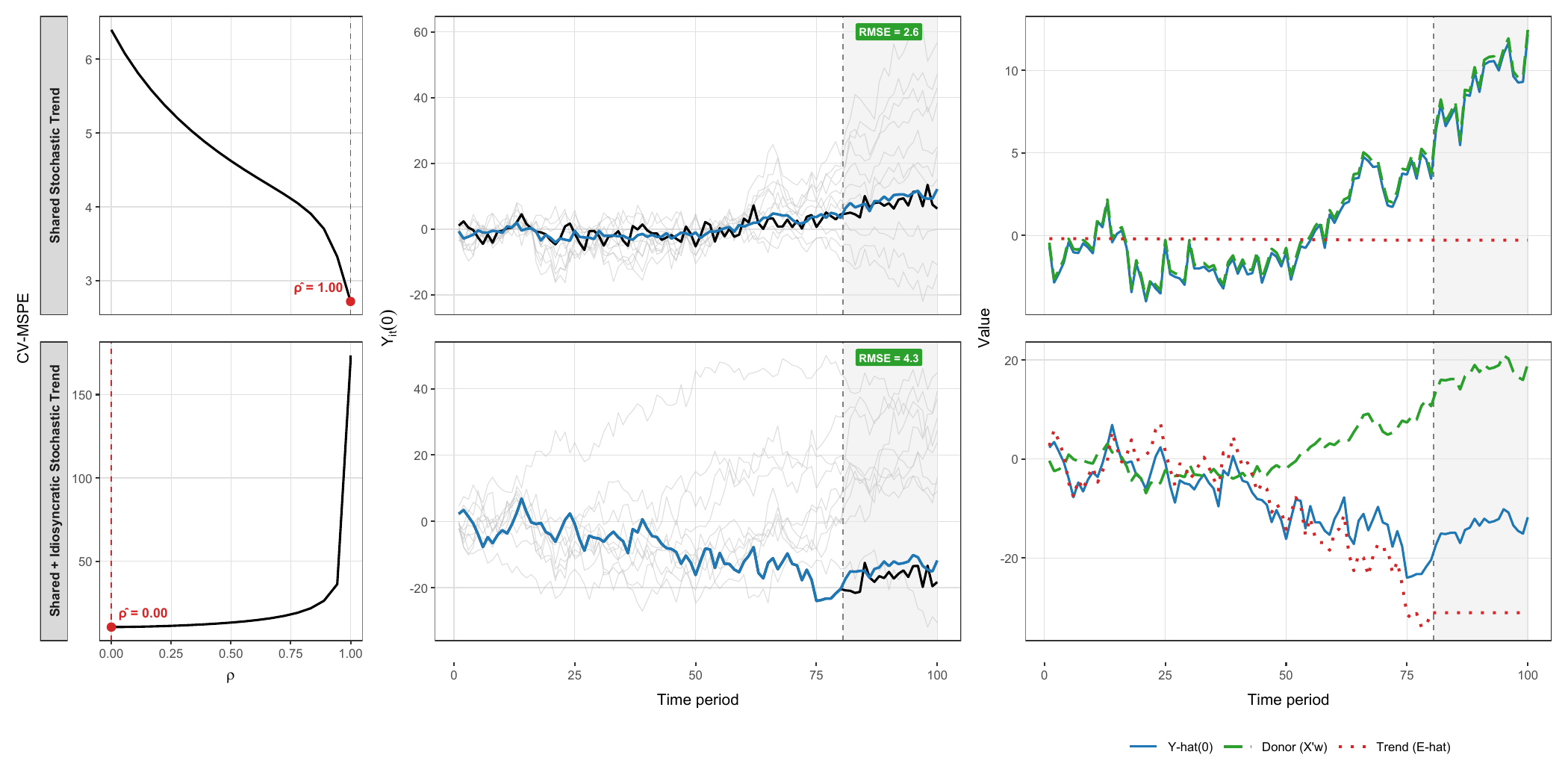}
\caption{HSC with $q=1$ and constant extrapolation, applied to the two
simulated regimes of Section~\ref{sec:tradeoff}.}
\label{fig:hsc_example_d1}
\vspace{0.5em}
\begin{minipage}[]{\textwidth}
{\footnotesize\textbf{Note:} Each row corresponds to one DGP regime from
Section~\ref{sec:tradeoff}: shared stochastic trend ($\kappa=0$,
top) and shared + idiosyncratic stochastic trend ($\kappa=2$, bottom). The DGP parameters and seed are identical to those in Figure~\ref{fig:tradeoff}.
HSC is estimated with smoothness order $q=1$ and constant extrapolation
forecast operator $G_1^{\mathrm{const}}$; the tuning parameter $\rho$ is
selected by rolling-origin cross-validation
(Section~\ref{subsec:rho_selection}) with one-step-ahead horizon $h=1$.
\emph{Left column}: Cross-validation MSPE as a function of $\rho$. The
red point and dashed vertical line mark the selected $\hat\rho$.
\emph{Middle column}: Counterfactual fit. The black line is the treated
unit's outcome $Y_{1t}(0)$; grey lines are donor units; the blue line is
the HSC counterfactual $\hat Y_t(0;\hat\rho,1)$. The shaded region marks
the post-treatment window ($t>T_0$), and the green label reports the
post-treatment RMSE.
\emph{Right column}: Decomposition of the HSC counterfactual for the
treated unit. The blue line is the counterfactual
$\hat Y_t(0)=X_t'\hat\omega + \hat E_t$ (identical to the middle
column); the green dashed line is the donor-matched component
$X_t'\hat\omega$; the red dotted line is the extrapolated smooth
component $\hat E_t$.\par}
\end{minipage}
\end{figure}

The contrast between $\hat\rho=1$ and $\hat\rho=0$ provides a sharp demonstration of the adaptive mechanism of HSC: the same
estimator with the same tuning procedure selects the two extreme
endpoints of the $\rho$ continuum, recovering pure level matching when
the stochastic trend is shared and pure differencing when the drift is partly
idiosyncratic, precisely the regime-dependent behavior that
Section~\ref{sec:tradeoff} argues is needed.

\paragraph*{Configuration 2: $q=2$, ARIMA forecast.}

The second configuration sets $q=2$ and uses an ARIMA(1,1,0)
model as the data-driven forecaster $\hat G_q$ within the
null-space separation construction of
Section~\ref{subsec:counterfactual}. Definition~\ref{def:Gq} is satisfied by design: the null-space
component of $\hat E_\pre$ (which now includes both an
intercept and a linear trend) is
extrapolated by the canonical linear forecaster
$G_2^{\mathrm{lin}}$, while the non-null-space component is
extrapolated by an ARIMA model. 

Figure~\ref{fig:hsc_example_d2} displays the results in the
same format as Figure~\ref{fig:hsc_example_d1}. In the shared
stochastic trend regime, cross-validation again selects
$\hat\rho=1$. As in Configuration~1, the donor pool suffices
to match the treated unit's dynamics, and the time series forecaster plays a minimal role. The post-treatment RMSE (2.7) is
comparable to Configuration~1 (2.6).

In the regime with both shared and idiosyncratic stochastic trends, cross-validation selects $\hat\rho=0.33$. The decomposition
panel reveals that the smooth component $\hat E_\pre$ exhibits
a clear downward trend in the pre-treatment period, and the
ARIMA forecaster extrapolates this trend into the
post-treatment window. This trending extrapolation is visible
in the red dotted line, which continues to decline after $T_0$
rather than remaining flat as in
Figure~\ref{fig:hsc_example_d1}. As a result, less of the
idiosyncratic drift needs to be absorbed by the level shift
alone, and the post-treatment RMSE improves from 4.3 to
3.4.\footnote{The small RMSE difference (4.3 vs.\ 4.4) between
Configuration~1 at $\hat\rho=0$ and the unregularized
first-differenced SC of Figure~\ref{fig:tradeoff} reflects the
ridge term $\zeta^2 T_0 \|\omega\|_2^2$.}

\begin{figure}[!h]
\centering
\includegraphics[width=\textwidth]{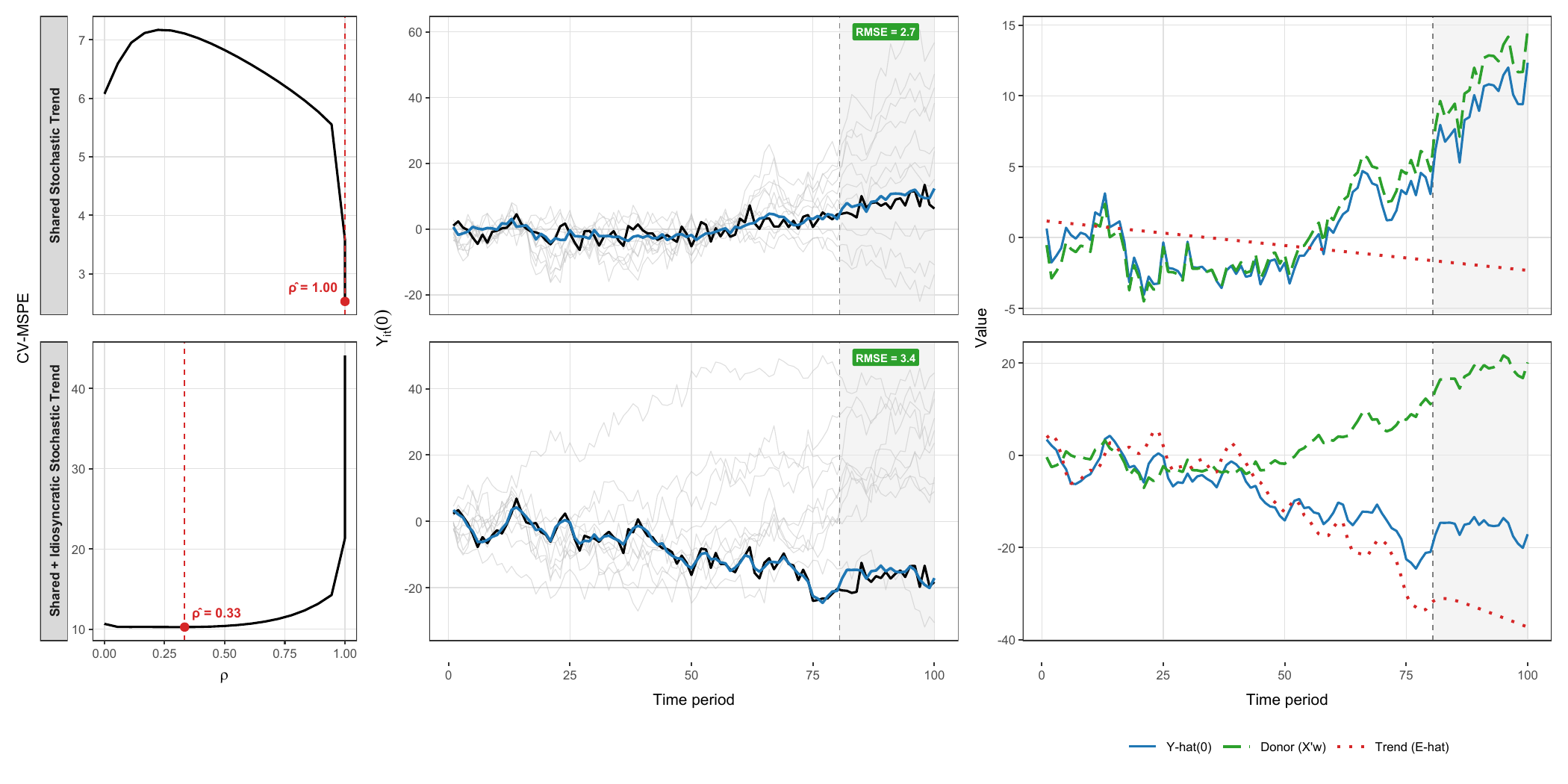}
\caption{HSC with $q=2$ and ARIMA(1,1,0) forecast, applied to the two
simulated regimes of Section~\ref{sec:tradeoff}.}
\label{fig:hsc_example_d2}
\vspace{0.5em}
\begin{minipage}[]{\textwidth}
{\footnotesize\textbf{Note:} Same DGP and format as
Figure~\ref{fig:hsc_example_d1}, but with smoothness order $q=2$ and an
ARIMA(1,1,0) forecast. The second-order penalty penalizes
changes in slope, so the null space of $K_2$ includes both intercept and
linear trend components. In the shared stochastic trend regime (top row),
cross-validation again selects $\hat\rho=1$, and the results are
similar to Figure~\ref{fig:hsc_example_d1}. In the shared + idiosyncratic stochastic trend regime (bottom row), cross-validation selects the interior value
$\hat\rho=0.33$. The decomposition panel shows that the smooth component
$\hat E_t$ (red dotted line) exhibits a trending pattern that the ARIMA
forecaster extrapolates beyond $T_0$, unlike the flat carry-forward in
Figure~\ref{fig:hsc_example_d1}. This trending extrapolation improves
the post-treatment RMSE from 4.3 (Figure~\ref{fig:hsc_example_d1}) to
3.4.\par}
\end{minipage}
\end{figure}

This comparison highlights the complementary roles of the three tuning
choices. The tuning parameter $\rho$ controls how much of the
pre-treatment discrepancy is allocated to each branch. The forecast operator $G_q$ determines how the
smooth component is extrapolated. The smoothness order $q$ determines what counts as ``smooth''. Across both configurations, the
cross-validation procedure adapts $\rho$ to the data-generating regime.

\section{Prediction-Error Decomposition}
\label{sec:predict_decompose}

Having established the HSC estimator and its spectral interpretation,
we now develop a formal decomposition of the counterfactual prediction
error.  This decomposition separates the prediction error into a
\emph{weight-estimation} component, which captures the discrepancy
introduced by estimating donor weights from observed outcomes rather
than from the underlying shared structure, and a \emph{forecasting}
component, which captures the prediction error that would remain even
with oracle weights.  Each component depends on $\rho$, and their
interplay provides a formal decomposition that helps interpret the
adaptive $\rho$-selection illustrated in
Section~\ref{subsec:hsc_example}.

Throughout this section we fix $q\in\{1,2\}$ and suppress the
$q$-subscript when no ambiguity arises, writing $K=K_q$,
$S_\rho=S_{\rho,q}$, $W_\rho=W_{\rho,q}$, $P_0=P_{0,q}$, and
$P_\perp=I_{T_0}-P_{0,q}$.  We use the null-space-separated forecast
operator $\widetilde G_q=G_q^{\mathrm{null}}P_0+\hat G_q P_\perp$
constructed in Section~\ref{subsec:counterfactual} and write
$\Pi_\rho:=\widetilde G_q S_\rho$ for the composed operator.  Since $S_\rho$ leaves $\Null(K)$ components intact and
$\widetilde G_q$ continues $\Null(K)$ components to post-treatment periods
(Definition~\ref{def:Gq}), the composed operator inherits the
null-space continuation property:
\begin{equation}
\Pi_\rho v = G_q^{\mathrm{null}} v
\qquad\text{for every }v\in\Null(K).
\label{eq:Pi_null}
\end{equation}

\subsection{Outcome model and oracle benchmark}
\label{subsec:LR_oracle}

We connect the prediction error to the data-generating structure
discussed in Section~\ref{sec:tradeoff}.  Recall from
Section~\ref{subsec:allocation} the decomposition
$Y_{it}(0)=L_{it}+R_{it}+\varepsilon_{it}$, where $L_{it}$ denotes
the shared component that provides the signal for donor matching, $R_{it}$ is an idiosyncratic stochastic
trend, and $\varepsilon_{it}$ is idiosyncratic short-run noise.  For the
prediction-error analysis, we work with the two-component grouping
\begin{equation}
Y_{it}(0) \;=\; L_{it} \;+\; \cR_{it},
\label{eq:DGP_2comp}
\end{equation}
where $\cR_{it}:=R_{it}+\varepsilon_{it}$ collects all components
not explained by the shared component.\footnote{The grouping
$\cR_{it}=R_{it}+\varepsilon_{it}$ is adopted for analytical
convenience: the prediction-error decomposition distinguishes between
the shared component $L$ and everything else, regardless of
whether that remainder is short-run noise or a stochastic trend.  The persistence
structure of $\cR_{it}$ matters for the spurious matching
channel of Term~A
(Section~\ref{subsec:TermA_channels}).}  

Define donor and treated stacks
for $L$ and $\cR$ analogously to $X_\pre$ and $Y_\pre$:
\[
X_\pre = L_{0,\pre}+\cR_{0,\pre},
\quad
Y_\pre = L_{1,\pre}+\cR_{1,\pre},
\qquad
X_\post = L_{0,\post}+\cR_{0,\post},
\quad
Y_\post(0) = L_{1,\post}+\cR_{1,\post}.
\]

We compare the HSC estimator against an oracle that observes the
signal $L$ directly and solves the $\rho=1$ HSC program on $L$
(with the null space projected out):\footnote{The oracle uses the
same ridge parameter $\zeta$ as the HSC estimator.  This ensures
that the oracle benchmark is not trivially superior due to
different regularization, so the decomposition isolates the
effects of observing $Y$ rather than $L$ and of using $W_\rho$
rather than $P_\perp$.}
\begin{equation}
\omega^\oracle
\;\in\;
\argmin_{\omega\in\Delta_{N_0}}
\left\{
\bigl\|P_\perp(L_{1,\pre}-L_{0,\pre}\omega)\bigr\|_2^2
+\zeta^2T_0\|\omega\|_2^2
\right\},
\label{eq:oracle_def}
\end{equation}
where $P_\perp$ projects onto $\Null(K)^\perp$, removing an
intercept ($q=1$) or an intercept and linear trend ($q=2$).
This is the same quadratic program as the $\rho=1$ endpoint of
HSC (Definition~\ref{def:hsc} with $W_{1,q}=P_\perp$), but on
the shared component $L$ rather than the observed outcome
$Y=L+\cR$.  The oracle is treated as a fixed benchmark
throughout the analysis.

\subsection{The prediction-error decomposition}
\label{subsec:AB_decomp}

Recall from~\eqref{eq:hsc_counterfactual_alt} that the HSC
counterfactual at any $\rho$ takes the form
$\hat Y_\post(0;\rho,q) = X_\post\hat\omega(\rho,q)
+ \Pi_\rho\,r_\pre(\hat\omega(\rho,q))$,
where $r_\pre(\omega):=Y_\pre-X_\pre\omega$ is the pre-treatment
residual.  We decompose the prediction error by introducing an oracle
predictor that replaces the estimated weights with $\omega^\oracle$
but retains the same $\rho$-dependent smoother and forecaster:
\begin{equation}
\widetilde Y_\post^\oracle(0;\rho)
\;:=\;
X_\post\omega^\oracle
\;+\;
\Pi_\rho\,r_\pre^\oracle,
\label{eq:oracle_pred}
\end{equation}
where $r_\pre^\oracle:=Y_\pre-X_\pre\omega^\oracle$ is the oracle
pre-treatment residual.  Adding and subtracting
$\widetilde Y_\post^\oracle(0;\rho)$ in the prediction error yields:

\begin{proposition}[Prediction-error decomposition]
\label{prop:AB}
For every $\rho\in[0,1]$,
\begin{equation}
\underbrace{Y_\post(0)-\hat Y_\post(0;\rho,q)}_{\text{prediction
error}}
\;=\;
\underbrace{\bigl(\widetilde Y_\post^\oracle(0;\rho)
  -\hat Y_\post(0;\rho,q)\bigr)}_{\mathrm{Term~A}(\rho)
  \ \text{(weight estimation)}}
\;+\;
\underbrace{\bigl(Y_\post(0)
  -\widetilde Y_\post^\oracle(0;\rho)\bigr)}_{\mathrm{Term~B}(\rho)
  \ \text{(forecasting)}}.
\label{eq:AB}
\end{equation}
The two terms admit the closed forms
\begin{align}
\mathrm{Term~A}(\rho)
&=
\bigl(X_\post-\Pi_\rho X_\pre\bigr)
\bigl(\omega^\oracle-\hat\omega(\rho,q)\bigr),
\label{eq:TA_closed}\\
\mathrm{Term~B}(\rho)
&=
\bigl(Y_\post(0)-X_\post\omega^\oracle\bigr)
\;-\;
\Pi_\rho\bigl(Y_\pre-X_\pre\omega^\oracle\bigr).
\label{eq:TB_closed}
\end{align}
\end{proposition}

\noindent
The derivation given in Appendix~\ref{app:proofs_sec5} uses only
the linearity of $r_\pre(\omega)$ in $\omega$ and the linearity
of~$\Pi_\rho$.

Term~A isolates the cost of using the estimated
weights $\hat\omega(\rho,q)$ rather than the oracle weights
$\omega^\oracle$ in post-treatment periods.  The weight discrepancy
$\omega^\oracle-\hat\omega(\rho,q)$ is transferred into prediction
error through the donor-forecast residual matrix
$C_\rho:=X_\post-\Pi_\rho X_\pre$, whose
column $j$ is the forecast residual when the composed operator
$\Pi_\rho$ extrapolates donor $j$ from the pre-treatment to the
post-treatment period.  The subtracted term $\Pi_\rho X_\pre$ appears
because the donor weights enter the counterfactual twice, directly
through the post-treatment donor block $X_\post\omega$ and indirectly
through the pre-treatment residual $r_\pre(\omega)=Y_\pre-X_\pre\omega$
that $\Pi_\rho$ extrapolates forward.  Term~A therefore depends on $\rho$
through both the estimated weights and $C_\rho$.

Term~B is the prediction error that would remain even if the oracle
weights were available.  It measures how accurately the composed
operator $\Pi_\rho$ extrapolates, from the pre-treatment to the
post-treatment period, the discrepancy between the treated unit and
its oracle synthetic counterfactual, where the oracle weights are
those defined by the shared component $L$.  Term~B depends on $\rho$
through the smoother $S_\rho$ embedded in $\Pi_\rho$, which controls
how much of the pre-treatment oracle residual is passed to the
forecaster and how much is discarded.

\subsection{Weight-estimation error: three channels of distortion}
\label{subsec:TermA_channels}

We now characterize the Term~A cost by identifying three
distinct channels through which the weight discrepancy
$\omega^\oracle-\hat\omega(\rho,q)$ arises.

The estimated weights $\hat\omega(\rho,q)$ and the oracle weights
$\omega^\oracle$ minimize different objectives over the same
constraint set $\Delta_{N_0}$.  The profiled HSC objective
(Definition~\ref{def:hsc}) is
\begin{equation}
	F_\rho(\omega)
	\;:=\;
	r_\pre(\omega)'W_\rho\,r_\pre(\omega)
	+\zeta^2T_0\|\omega\|_2^2,
	\label{eq:F_def}
\end{equation}
while the oracle objective~\eqref{eq:oracle_def} is
\begin{equation}
	H^\oracle(\omega)
	\;:=\;
	\bigl\|P_\perp(L_{1,\pre}-L_{0,\pre}\omega)\bigr\|_2^2
	+\zeta^2T_0\|\omega\|_2^2.
	\label{eq:H_def}
\end{equation}
These two objectives share the same ridge penalty but differ in two
ways.  First, the \emph{metric}: HSC evaluates the pre-treatment
residual under $W_\rho$, whereas the oracle uses $P_\perp$.  The two
coincide only at $\rho=1$ (where $W_\rho=P_\perp$).  Second, the
\emph{data}: HSC observes $X_\pre=L_{0,\pre}+\cR_{0,\pre}$ and
$Y_\pre=L_{1,\pre}+\cR_{1,\pre}$, whereas the oracle operates on the
signal $L$ alone.  The weight discrepancy
$\omega^\oracle-\hat\omega(\rho,q)$ reflects the combined effect of
these two differences.

To separate the contributions of $L$ and $\cR$ to the weight
discrepancy, we decompose the oracle residual by component.  For each
$Z\in\{L,\cR\}$, define
$e^Z_\pre:=Z_{1,\pre}-Z_{0,\pre}\omega^\oracle\in\R^{T_0}$, so that
$r_\pre^\oracle=e^L_\pre+e^\cR_\pre$ by
additivity~\eqref{eq:DGP_2comp}.  The signal residual $e^L_\pre$
contains a null-space component $P_0\,e^L_\pre\in\Null(K)$ that the
oracle objective does not use, and a complement
$P_\perp\,e^L_\pre\in\Null(K)^\perp$ that the oracle directly
minimizes.  Since both
$W_\rho$ and $P_\perp$ annihilate $\Null(K)$ components, we have
$(W_\rho-P_\perp)\,P_\perp\,e^L_\pre=(W_\rho-P_\perp)\,e^L_\pre$
and $W_\rho\,P_\perp\,e^L_\pre=W_\rho\,e^L_\pre$.

The following proposition bounds
$\|\mathrm{Term~A}(\rho)\|_2$ by three interpretable components that
correspond to the metric difference, the data difference, and their
interaction.  Define the bound components
\begin{align}
	A_1(\rho)
	&:=\frac{1}{T_0}\bigl\|L_{0,\pre}'(W_\rho-P_\perp)\,
	e^L_\pre\bigr\|_{Q_\rho^{-1}},
	\label{eq:A1_def}\\
	A_2(\rho)
	&:=\frac{1}{T_0}\bigl\|\cR_{0,\pre}'W_\rho\,
	e^L_\pre\bigr\|_{Q_\rho^{-1}},
	\label{eq:A2_def}\\
	A_3(\rho)
	&:=\frac{1}{\sqrt{T_0}}\,
	\bigl\|e^\cR_\pre\bigr\|_{W_\rho},
	\label{eq:A3_def}
\end{align}
where $\|v\|_{Q_\rho^{-1}}:=\sqrt{v'Q_\rho^{-1}v}$,\;
$\|v\|_{W_\rho}:=\sqrt{v'W_\rho v}$, and
\begin{equation}
	Q_\rho
	\;:=\;
	\frac{1}{T_0}X_\pre'W_\rho X_\pre+\zeta^2 I_{N_0}
	\label{eq:Q_def}
\end{equation}
is the Hessian of $F_\rho$ (up to the constant $2T_0$).

\begin{proposition}[Weight-estimation error envelope]
	\label{prop:TermA_env}
	For every $\rho\in[0,1]$,
	\begin{equation}
		\bigl\|\mathrm{Term~A}(\rho)\bigr\|_2
		\;\le\;
		\mathcal{P}_\rho
		\bigl[
		A_1(\rho)+A_2(\rho)+A_3(\rho)
		\bigr],
		\label{eq:TermA_env}
	\end{equation}
	where
	$\mathcal{P}_\rho
	:=\|(X_\post-\Pi_\rho X_\pre)Q_\rho^{-1/2}\|_{\mathrm{op}}$
	is the transfer multiplier.
\end{proposition}

Proposition~\ref{prop:TermA_env} bounds the weight-estimation
error with four ingredients, each with a distinct role.  The
terms $A_1(\rho)$, $A_2(\rho)$, and $A_3(\rho)$ are three sources of
the weight gap $\omega^\oracle-\hat\omega(\rho,q)$, each tied to a
specific distortion mechanism.  The matrix $Q_\rho$, the
ridge-regularized Gram matrix of the filtered donor series
$W_\rho^{1/2}X_\pre$, summarizes how sharply the filtered donor pool
distinguishes different weight-reallocation patterns; the dual norm
$\|\cdot\|_{Q_\rho^{-1}}$ through which $A_1$ and $A_2$ are measured
reflects this geometry.  The multiplier $\mathcal{P}_\rho$ converts
the weight gap into a prediction error and is the only place in the
Term~A envelope where the forecast operator $\widetilde G_q$ enters.
We unpack each ingredient in turn, then synthesize them into two
opposing forces that shape Term~A in~$\rho$.

The \emph{metric distortion} $A_1(\rho)$ measures the contribution
to the weight gap $\omega^\oracle-\hat\omega$ that arises because
HSC and the oracle use different \emph{metrics} to evaluate the
pre-treatment residual: HSC uses $W_\rho$, whereas the oracle uses
$P_\perp$.  Thus $A_1$ is the structural price HSC pays for
evaluating the pre-treatment residual under $W_\rho$ rather than the
oracle's $P_\perp$, which leaves the weight criterion no longer
aligned with the oracle's identification target.  Two limiting
cases eliminate this price entirely.  First, at $\rho=1$,
$W_\rho=P_\perp$ and $A_1$ vanishes because the two metrics
coincide.  Second, $A_1$ vanishes for
all $\rho$ whenever $P_\perp\,e^L_\pre=0$, that is, whenever the
treated unit's signal $L_{1,\pre}$ can be written as a convex
combination of donor signals $L_{0,\pre}\omega^\oracle$ plus a
$\Null(K)$ component.\footnote{Concretely, $P_\perp\,e^L_\pre=0$
	means that the oracle achieves zero residual-fit term
	in~\eqref{eq:oracle_def}: the mismatch in $L$ lies entirely in
	$\Null(K)$ (a level shift for $q=1$, an affine trend for $q=2$).
	Both $W_\rho$ and $P_\perp$ assign zero weight to $\Null(K)$, so
	the choice of metric is inconsequential.}  Beyond these two cases, the metric discrepancy $W_\rho-P_\perp$
widens as $\rho$ decreases from $1$, which pushes $A_1$ upward.
However, because $A_1$ is measured in the $Q_\rho^{-1}$ norm, which
also depends on $\rho$, the net behavior of $A_1(\rho)$ need not be
monotonic.  Notably, $A_1$ can be nonzero even when $\cR=0$, which
is the sense in which it is a purely structural channel.

The \emph{interaction} $A_2(\rho)$ captures how idiosyncratic
components in the donor units $\cR_{0,\pre}$ can perturb weight
estimation \emph{when the signal $L$ is not perfectly matched}.
The oracle benchmark is defined using the shared component $L$
only, and the HSC objective is evaluated on the observed donor
outcomes $X_{\pre}=L_{0,\pre}+\cR_{0,\pre}$.  As a result,
$\cR_{0,\pre}$ can accidentally align with $W_\rho\,e^L_{\pre}$ and
appear to help reduce the remaining signal mismatch in sample,
thereby shifting $\hat\omega(\rho)$ away from $\omega^\oracle$.
This term is large when $\cR_{0,\pre}$ has a substantial projection
onto $W_\rho\,e^L_{\pre}$, and it vanishes whenever
$P_\perp\,e^L_\pre=0$.

The \emph{spurious matching} $A_3(\rho)$ measures how tempted the
HSC weight criterion is to chase idiosyncratic components in the
donors.  Unlike $A_1$ and $A_2$, this channel does not require any
mismatch in $L$.  The oracle weights $\omega^\oracle$ are chosen to
fit $L$ alone. The $\cR$-mismatch
$e^\cR_\pre=\cR_{1,\pre}-\cR_{0,\pre}\omega^\oracle$ is
untouched by the oracle and can be large. The HSC weight
criterion thus has an incentive to deviate from $\omega^\oracle$ to
absorb it, pulling weights toward donors whose idiosyncratic
components happen to co-move with $\cR_{1,\pre}$.

The sensitivity of $A_3$ to $\rho$ depends critically on whether
$\cR$ is short-run noise or a stochastic trend.  At $\rho=1$,
$W_\rho=P_\perp$, and the weight criterion applies no spectral
down-weighting beyond removing the null-space component.  If $\cR$ contains a
random-walk component, $\|e^\cR_\pre\|_{P_\perp}$ grows at rate
$O_p(T_0)$, so $A_3(1)=O_p\!\bigl(T_0^{1/2}\bigr)\to\infty$, which
reflects the spurious regression phenomenon discussed in
Section~\ref{subsec:spurious}.  At $\rho=0$, the metric reduces to
$W_0=K=D_q'D_q$, so the weight criterion operates on the $q$th
differences of $e^\cR_\pre$; differencing renders the random-walk
component stationary and yields $A_3(0)=O_p(1)$, thereby controlling
the spurious channel.  For interior values $\rho\in(0,1)$, the
spectral weight $w(\mu;\rho)=\mu/\bigl((1-\rho)+\rho\mu\bigr)$
interpolates smoothly between the two extremes: larger $\rho$ retains
more low-frequency energy in the weight criterion and is therefore
more vulnerable to spurious matching when $\cR$ contains a stochastic
trend, whereas smaller $\rho$ down-weights low-frequency components
more aggressively, at the cost of the metric distortion already
discussed for $A_1(\rho)$.  When $e^\cR_\pre$ contains only short-run
noise, either because $\cR$ itself is stationary or because the
treated unit and the control donors form a cointegrated system under
the oracle weights so that the stochastic-trend components cancel,
$A_3$ is $O_p(1)$ at both endpoints.

The operative quantity is $\lambda_{\max}(Q_\rho^{-1})$, the largest
eigenvalue of $Q_\rho^{-1}$.  It measures how weakly the HSC
criterion identifies the donor weights: a large value means the
filtered donors are nearly collinear in some direction, so small
score discrepancies are amplified through the dual norm
$\|\cdot\|_{Q_\rho^{-1}}$, inflating $A_1$ and $A_2$.  How
$\lambda_{\max}(Q_\rho^{-1})$ changes with $\rho$ can be
\emph{non-monotonic}.  At $\rho=0$, the metric is $K_q$, so
low-frequency directions are already most strongly down-weighted
while high-frequency directions are amplified.  Moving $\rho$ away
from zero gradually restores weight on low-frequency components and
reduces the amplification of high-frequency components.  Depending on
the frequency composition of the control donors, these two effects
can make $\lambda_{\max}(Q_\rho^{-1})$ peak at an interior
$\rho\in(0,1)$.  The
ridge floor $\zeta^2$ in $Q_\rho$ guarantees that
$\lambda_{\max}(Q_\rho^{-1})\le 1/\zeta^2$ even when the filtered
donors are nearly collinear; the no-ridge comparison in
Appendix~\ref{app:mc_decomp} shows that without this floor
$\lambda_{\max}(Q_\rho^{-1})$ can grow dramatically.

The transfer multiplier
$\mathcal{P}_\rho=\|(X_\post-\Pi_\rho X_\pre)Q_\rho^{-1/2}\|_\mathrm{op}$
depends on $\rho$ through two distinct mechanisms.  The first is the
inverse curvature $Q_\rho^{-1/2}$, which reflects the same
identification geometry that shapes $A_1$ and $A_2$.  The second is
$C_\rho$; its dependence on $\rho$ runs through the smoother
$S_\rho$, which determines how much of each donor's pre-treatment
path is passed to the forecaster.

Taken together, the envelope $\mathcal{P}_\rho[A_1+A_2+A_3]$ is
governed by two opposing forces in $\rho$.  At large $\rho$ the
dominant cost is \emph{spurious matching}: the $A_3$ channel grows
when $\cR$ carries stochastic trends, formalizing the spurious
donor matching risk of Section~\ref{subsec:spurious}.  At small
$\rho$ the dominant cost is \emph{identification loss}: the metric
gap $W_\rho-P_\perp$ widens, possibly inflating $A_1$, and the filtered donor
design $W_\rho^{1/2}X_\pre$ sheds low-frequency variation, inflating
$\lambda_{\max}(Q_\rho^{-1})$.  This is the over-filtering cost of
Section~\ref{subsec:filtering}, and it is most severe when the donor
series are dominated by low-frequency variation, as is typical for
macroeconomic data.  The net shape of the envelope in $\rho$ is
therefore non-monotonic in general, and depends on the size and
structure of the signal mismatch $e^L_\pre$, the persistence of
$\cR$, the frequency composition of the donor design $X_\pre$, and the
forecast operator $\widetilde G_q$.

\subsection{Forecasting error}
\label{subsec:TermB}

Term~B is the prediction error that would remain even if the HSC
weights coincided with the oracle weights.  Its size depends on how
accurately the composed operator $\Pi_\rho=\widetilde G_q S_\rho$
extrapolates the oracle pre-treatment residual into the
post-treatment window.  Because $S_\rho$ varies with $\rho$ while
$\widetilde G_q$ does not, the $\rho$-dependence of Term~B is
governed entirely by the smoother and by what it forwards to the
forecaster.

Recall the oracle pre-treatment residual $r_\pre^\oracle$ from
Section~\ref{subsec:AB_decomp}, and define its post-treatment
counterpart:
\begin{equation}
r_\pre^\oracle
\;:=\;
Y_\pre - X_\pre\omega^\oracle
\;\in\;\R^{T_0},
\qquad
r_\post^\oracle
\;:=\;
Y_\post(0) - X_\post\omega^\oracle
\;\in\;\R^{T_\post}.
\label{eq:r_oracle}
\end{equation}
The null-space component is the exception under $\Pi_\rho$.  By
the null-space continuation property~\eqref{eq:Pi_null},
$\Pi_\rho v = G_q^{\mathrm{null}}v$ for every $v\in\Null(K)$
regardless of $\rho$: the smoother leaves such a component intact and
the forecaster then continues it by $G_q^{\mathrm{null}}$, in both
steps independently of $\rho$.  The null-space content of the oracle
residual therefore contributes a fixed offset to the post-treatment
prediction at every $\rho$ and plays no role in the $\rho$-dependent
tradeoff.  Accordingly, define
\begin{equation}
\eta_\pre^\oracle
\;:=\;
r_\pre^\oracle - P_0\,r_\pre^\oracle
\;=\;
P_\perp\,r_\pre^\oracle,
\qquad
\eta_\post^\oracle
\;:=\;
r_\post^\oracle
\;-\;
G_q^{\mathrm{null}}\,P_0\,r_\pre^\oracle.
\label{eq:eta_def}
\end{equation}
Both vectors are the oracle residual with the same null-space content
$P_0\,r_\pre^\oracle$ removed: directly in the pre-period, and
through its canonical continuation $G_q^{\mathrm{null}}\,P_0\,
r_\pre^\oracle$ in the post-period.  With this common adjustment,
$\mathrm{Term~B}(\rho)$ takes the form
\begin{equation}
\eta_\post^\oracle
\;-\;
\widetilde G_q\,S_\rho\,\eta_\pre^\oracle ,
\label{eq:TB_direct}
\end{equation}
with the derivation given in Appendix~\ref{app:proofs_sec5}.  Both
$\eta_\post^\oracle$ and $\eta_\pre^\oracle$ are $\rho$-independent;
all $\rho$-dependence in $\mathrm{Term~B}(\rho)$ enters through the
smoother~$S_\rho$.

At the endpoint $\rho=0$, $S_0=I$ and
$\mathrm{Term~B}(0)=\eta_\post^\oracle-\widetilde G_q\,
\eta_\pre^\oracle$.  At the other endpoint $\rho=1$, $S_1=P_0$, and
since $\eta_\pre^\oracle\in\Null(K)^\perp$ by construction,
$S_1\,\eta_\pre^\oracle=0$ and therefore
$\mathrm{Term~B}(1)=\eta_\post^\oracle$.  At $\rho=1$ the forecaster
makes no contribution to Term~B beyond the canonical null-space
continuation already absorbed into $\eta_\post^\oracle$, so the
choice of $\widetilde G_q$ is irrelevant at this endpoint.  Between
the endpoints, as $\rho$ increases from $0$ to $1$, the smoothed
input $S_\rho\,\eta_\pre^\oracle$ decreases monotonically toward zero
in the spectral sense of Section~\ref{sec:spectral}, so $\rho$ is the
dial that controls how much of $\eta_\pre^\oracle$ reaches the
forecaster.

How $\rho$ affects $\|\mathrm{Term~B}\|$ then depends on how well the
raw forecast $\widetilde G_q\,\eta_\pre^\oracle$ tracks
$\eta_\post^\oracle$.  When $\widetilde G_q\,\eta_\pre^\oracle$
already predicts $\eta_\post^\oracle$ well, the raw forecast is
useful and $\rho=0$ is preferred.  When
$\widetilde G_q\,\eta_\pre^\oracle$ over-extrapolates the noisier
part of the residual, smoothing its input first improves the
forecast and an interior $\rho>0$ is preferred.  When
$\widetilde G_q\,\eta_\pre^\oracle$ is far from $\eta_\post^\oracle$,
the forecaster is harmful, and $\rho=1$, which discards its input
entirely and leaves $\mathrm{Term~B}$ equal to $\eta_\post^\oracle$,
is preferred.  This logic suggests that a longer post-treatment
window favors a larger $\rho$: the time series forecaster becomes
less reliable at distant horizons, which pushes the preferred regime
toward stronger regularization or full suppression of the time series forecaster.

\subsection{Implications for tuning}
\label{subsec:synthesis}

Sections~\ref{subsec:TermA_channels} and~\ref{subsec:TermB}
characterized the two errors that HSC trades off as $\rho$
varies.  We now collect their implications for the choices a
practitioner makes.  HSC exposes four such choices, and the
decomposition shows they play distinct roles.  The tuning
parameter $\rho$ allocates the pre-treatment residual between
donor matching and the time series forecaster.  The forecast
operator $\widetilde G_q$ and the smoothness order $q$ are
structural: $\widetilde G_q$ fixes how the non-null residual is
extrapolated, and $q$ fixes what counts as smooth and how the
null space is continued.  The cross-validation horizon $h$ does
not change the estimator; it determines which of the two errors
the cross-validation criterion weighs most heavily.  We take $\rho$
first, then $\widetilde G_q$, $q$, and $h$.

$\rho$ is the allocation lever, and the tradeoff it controls is
two-sided.  At large $\rho$ the weight criterion retains the
low-frequency content of the pre-treatment residual, so when
$\cR$ carries a stochastic trend the spurious matching channel
$A_3$ grows and Term~A rises.  At small $\rho$ the criterion is
restricted toward high-frequency content: the metric distortion
$A_1$ widens and the filtered donor design sheds the
low-frequency variation that identifies the weights, so Term~A
rises again through identification loss.  The $\rho$-shape of
Term~B is governed instead by how accurately the composed
operator extrapolates the oracle residual: when the raw forecast
tracks $\eta_\post^\oracle$ well a small $\rho$ is preferred, and
when it does not a large $\rho$, which suppresses the
forecaster's input, is preferred.  Neither error is monotone in
$\rho$, and their sum has no general optimum; the best $\rho$
depends on the data-generating regime.

The forecast operator $\widetilde G_q$ is the lever common to
both errors.  It enters Term~A only through $C_\rho$ inside the
transfer multiplier $\mathcal{P}_\rho$, and it drives Term~B
directly through $\widetilde G_q\,S_\rho\,\eta_\pre^\oracle$.  A
forecaster that extrapolates the oracle residual well lowers the
Term~B floor.  Its effect on the multiplier $\mathcal{P}_\rho$ is
separate: $\mathcal{P}_\rho$ depends on the forecaster only through the
donor forecast residuals $C_\rho=X_\post-\Pi_\rho X_\pre$, and the same
operator can behave differently on $\eta_\pre^\oracle$ than on the
donor paths, so the two effects need not move together.  In the Monte
Carlo study, the constant carry-forward and the ARIMA$(1,1,0)$
forecasters both perform well.\footnote{In both cases the rule is the data-driven component
$\hat G_q$ of the construction in Definition~\ref{def:Gq}:
it is applied to the non-null part $P_{\perp,q}r$ of the
pre-treatment residual $r$, while the null-space part
$P_{0,q}r$ is continued by the canonical $G_q^{\mathrm{null}}$.
When $\hat G_q$ is the constant carry-forward, the two parts
recombine in closed form.  For $q=1$, the continued mean plus
the carried-forward demeaned residual equals the last entry
$r_{T_0}$ held constant, so $\widetilde G_1$ coincides with the
constant carry-forward applied directly to the raw
pre-treatment residual.  For $q=2$, the composed forecast at
horizon $h$ is $r_{T_0}+\hat\beta\,h$, where $\hat\beta$ is the
slope of the line fitted to the pre-treatment residual over
$\Null(K_2)$; equivalently, the fitted linear trend is
extrapolated with its level re-anchored to the last entry
$r_{T_0}$.  Under the ARIMA$(1,1,0)$ rule, the null-space part is
continued in the same way, while the non-null part is forecast
by the ARIMA model.}

The smoothness order $q$ is a structural choice: it fixes the
penalty $K_q$, the null space $\Null(K_q)$, and the canonical
continuation $G_q^{\mathrm{null}}$, and the decomposition makes
its role precise.  At $\rho=0$ the metric is $K_q=D_q'D_q$, so
$A_3(0)\propto\|D_q\,e^\cR_\pre\|_2$ is controlled only if $q$ is
large enough that $D_q$ stationarizes the idiosyncratic
component of $\cR$.  An $I(1)$ idiosyncratic stochastic trend is
stationarized by $q=1$, whereas an $I(2)$ idiosyncratic
stochastic trend is not and requires $q=2$.  Raising $q$ to $2$
also enlarges the null space and allows a more flexible
specification: since $\Null(K_1)\subset\Null(K_2)$, any
approximately affine gap between the treated unit's signal and
its donor combination is then absorbed at no metric-distortion
cost.  These benefits come with two costs.  First, the same
enlargement forces $G_q^{\mathrm{null}}$ to continue that affine
direction: for $q=2$ it extends the line fitted to
$P_{0,2}\,r_\pre^\oracle$, and by the null-space continuation
property~\eqref{eq:Pi_null} this continuation sits inside
$\mathrm{Term~B}$ at every $\rho$ and can carry an extrapolation
bias that grows with the post-treatment window $T_\post$,
whereas the $q=1$ continuation carries a level forward and
leaves a floor that is flat in the horizon.  Second, the
spectrum of $K_2$ is much wider than that of $K_1$, which makes the amplification to the high-frequency components much more significant. The
metric gap $W_{\rho,q}-P_{\perp,q}$, and hence $A_1$, rises more
steeply as $\rho$ falls from $1$.

The cross-validation horizon $h$ does not alter the estimator
but selects which error the criterion minimizes.  At a short
horizon the forecaster's extrapolation bias is typically small,
so the criterion is dominated by Term~A.  At a long horizon the extrapolation bias
accumulates and Term~B can dominate; the criterion then rewards
a large $\hat\rho$ mechanically, which suppresses the time series forecaster.
It is worth noting that a larger $h$ leaves less pre-treatment
data for cross-validation.  In practice, when the pre-treatment
window is short, the researcher must balance the post-treatment
horizon of interest against the amount of pre-treatment data
retained for cross-validation.

These choices are not independent.  The cross-validation
criterion of Section~\ref{subsec:rho_selection} selects $\rho$,
and optionally $q$ and $\widetilde G_q$ jointly.  In practice
the structural choices can be guided by what is known about the
application.  Set $q$ to the smallest order that plausibly
stationarizes the raw data.  Choose $\widetilde G_q$
conservatively unless the pre-treatment data give clear evidence
that a richer forecaster predicts better.  Then let
cross-validation at the policy-relevant horizon $h$ select
$\rho$.

\section{Monte Carlo Evidence}
\label{sec:simulation}

Sections~\ref{sec:tradeoff}--\ref{sec:predict_decompose} motivate
HSC as a soft allocation mechanism, develop its spectral
interpretation, and derive a prediction-error decomposition that
separates donor matching from residual extrapolation.  This
section reports a Monte Carlo study that evaluates HSC's
finite-sample performance against standard synthetic control
estimators and documents how the cross-validated tuning parameter
$\hat\rho$ adapts to the underlying data-generating regime. Full details of the data-generating process appear in
Appendix~\ref{app:mc_estimators_design}.

\subsection{Design}
\label{subsec:mc_design}

The data-generating process retains the additive structure used
throughout the paper.  Untreated potential outcomes follow
\begin{equation}\label{eq:dgp_mc}
Y_{j,t}(0)=L_{j,t}+\kappa\,\mathcal{E}_{j,t}+\varepsilon_{j,t}+\alpha_j+\delta_t,
\end{equation}
where $L_{j,t}=\sum_{k=1}^{3}\Lambda_{j,k}F_{k,t}$ is a low-rank
component built from three factors $F_{k,t}$ (one random walk, one
ARIMA$(1,1,0)$, one stationary AR(1)), $\kappa \mathcal{E}_{j,t}$ is a
unit-specific ARIMA$(1,1,0)$ component whose innovations interpolate
between a common shock and an idiosyncratic shock by
$\sqrt{\rho_u}u_t^{\mathrm{c}}+\sqrt{1-\rho_u}u_{j,t}^{\mathrm{i}}$,
$\varepsilon_{j,t}\sim\mathcal{N}(0,1)$ is stationary noise, $\alpha_j$
is a unit fixed effect, and $\delta_t$ is a time fixed effect.  The
factor paths $F_{k,t}$, the idiosyncratic component $\mathcal{E}_{j,t}$, the
noise $\varepsilon_{j,t}$, and the time fixed effects $\delta_t$ are
redrawn in every replication; the factor loadings $\Lambda_{j,k}$
and the unit fixed effects $\alpha_j$ are drawn once and held fixed
across replications.  The treated unit's loadings are constructed as
a sparse Dirichlet convex combination of donor loadings, placing the
treated unit inside the donor convex hull.  Two scalar parameters
govern the persistence structure: $\kappa\in\{0,0.5,1,2\}$ controls
the amplitude of the unit-specific stochastic trend, and
$\rho_u\in\{0,0.5,1\}$ controls how much of that persistence is
shared across units.  We replicate every $(\kappa,\rho_u)$ cell
$R=500$ times with $T_0=200$ pre-treatment periods, $T_\post=20$
post-treatment periods, and $N_0=50$ donors.  The treated unit
receives no treatment effect, so the post-period RMSE between the
estimated counterfactual and the untreated potential outcome
measures predictive accuracy.

We compare five baseline synthetic control estimators against HSC.
The baselines are plain SC \citep{abadie2010synthetic}, SC with an
intercept (SC-INT, \citealt{doudchenko2016balancing}), synthetic
difference-in-differences \citep[SDID,][]{arkhangelsky2021synthetic},
and two variants of the synthetic business-cycle estimator of
\citet{shi2025synthetic} that differ in the pre-treatment filter used
to extract the cyclical component (SBCA-ARIMA and SBCA-Hamilton).  Plain SC matches in levels and is therefore
vulnerable to spurious matching.  SC-INT removes a unit-specific level shift with an intercept, and
SDID constructs its unit weights from a level-matching problem with an
intercept and a ridge penalty; both still match the residual variation
in levels.  The SBCA family
applies a pre-treatment filter to separate trend from cycle, matches
donors on the cycle, and extrapolates the treated trend
independently.

For HSC we evaluate four time series forecasters that all
satisfy Definition~\ref{def:Gq}: in each case the forecaster is
applied only to the non-null-space component, while the
null-space component is continued by the canonical
$G_q^{\mathrm{null}}$.  The \texttt{last\_constant} forecaster
carries the last fitted value of the non-null-space component
forward as a constant; under $q=1$ this, combined with the
canonical constant continuation of the null space, recovers
carrying the last fitted value of the residual forward, and
under $q=2$ it adds a constant offset to the null-space linear
extension.  The \texttt{arima110} forecaster fits an
ARIMA$(1,1,0)$ to the non-null-space component, the correctly
specified model for the DGP's idiosyncratic ARIMA$(1,1,0)$
stochastic trend.  The \texttt{ar} forecaster fits a stationary
AR(4) to the non-null-space component, whose forecasts
mean-revert.  The \texttt{hamilton} forecaster forecasts the non-null-space
component with the $h$-step-ahead regression of
\citet{hamilton2018you}.  Each forecaster is evaluated at both
smoothness orders $q\in\{1,2\}$.  The cross-validation horizon is fixed at $h=1$ in
this section; Appendix~\ref{app:mc_cv_h} examines the effect of
choosing $h=20$ instead.

\subsection{HSC ties or improves on baselines across regimes}
\label{subsec:mc_rmse}

Figure~\ref{fig:mc_rmse} reports the post-period RMSE pooled across
the 20 post-treatment periods for each of the 12 $(\kappa,\rho_u)$
cells.  Within each panel we show the five baseline estimators as
single bars, and we show each HSC forecaster as two bars side by
side: the lighter bar reports the $q=1$ result and the darker bar
reports $q=2$.

\begin{figure}[!ht]
\caption{Post-period pooled RMSE by method across the
$(\kappa,\rho_u)$ grid}\label{fig:mc_rmse}
\centering
\begin{minipage}{1\linewidth}{
\centering
\includegraphics[width=\textwidth]{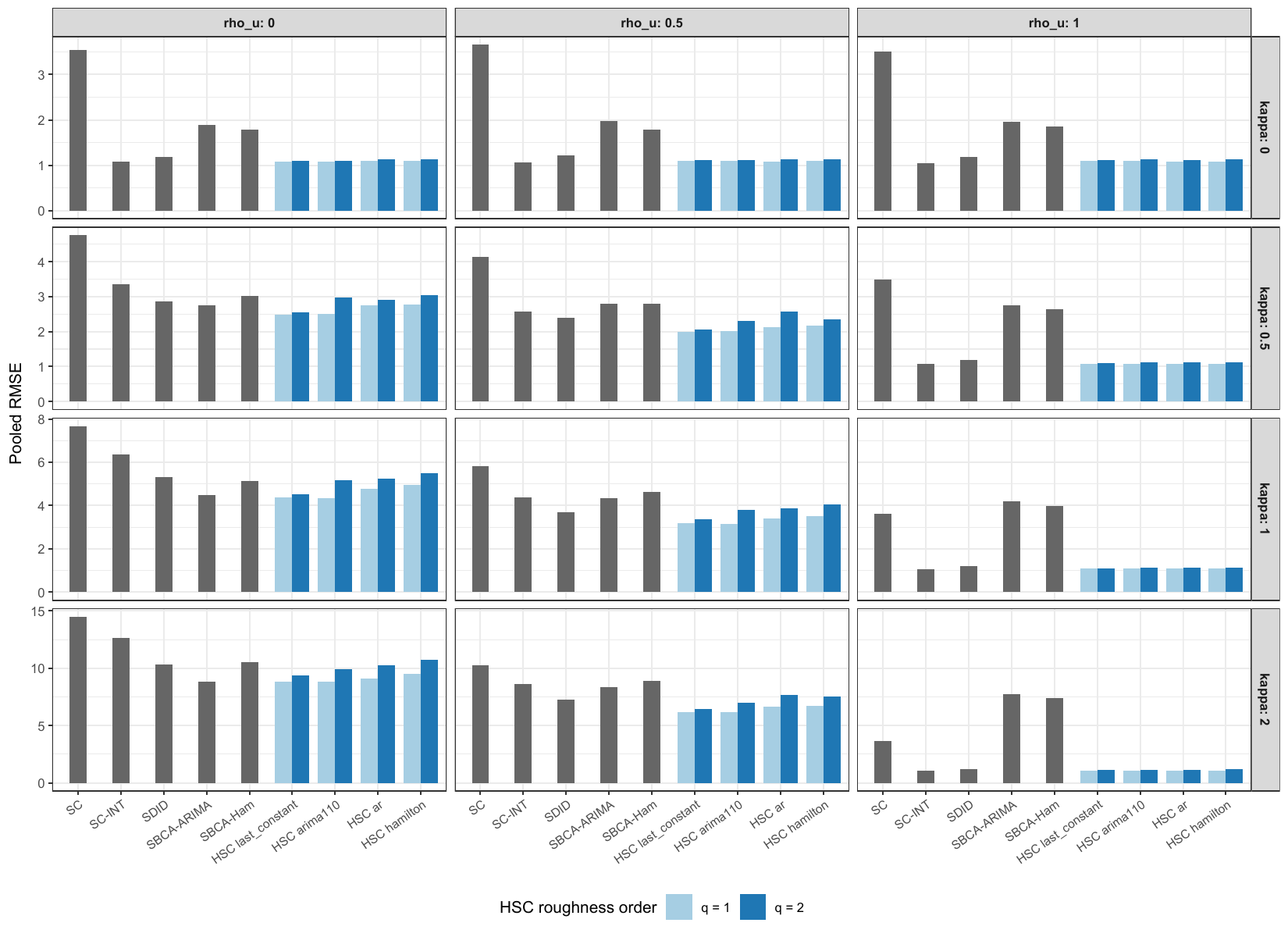}
}it
\footnotesize\textit{Notes:} Bars report pooled RMSE
$\sqrt{R^{-1}T_\post^{-1}\sum_{r,h}(\hat Y_{1,T_0+h}^{(r)}-Y_{1,T_0+h}^{(0,r)})^2}$
across $R=500$ replications and $T_\post=20$ post-treatment periods,
for each $(\kappa,\rho_u)$ cell.  The five baseline estimators
appear as single grey bars; the four HSC forecasters appear as
paired bars with the lighter shade denoting $q=1$ and the darker
shade denoting $q=2$.  $T_0=200$, $N_0=50$, $h=1$
cross-validation; the treated unit's loadings lie inside the donor
convex hull.
\end{minipage}
\end{figure}

Three patterns are visible. First, when the component $\mathcal{E}_{j,t}$ is absent or shared across units, all four HSC forecasters tie SC-INT and SDID and substantially improve on plain SC and the SBCA family. The top row of the figure ($\kappa=0$) and the right column ($\rho_u=1$) display this behavior: HSC's pooled RMSE is within a few percent of SC-INT and SDID, plain SC sits well above the others because it cannot remove the heterogeneous unit intercepts, and the SBCA filters strip away common variation that the other methods can match. Second, when the $\mathcal{E}_{j,t}$ component is present and partially or fully idiosyncratic ($\kappa\ge 0.5$, $\rho_u\le 0.5$), HSC with $q=1$ delivers the lowest pooled RMSE in most cells. The margin between HSC and SC-INT or SDID grows with $\kappa$ and is largest in the corner with the most idiosyncratic drift. The SBCA family performs well when $\mathcal{E}_{j,t}$ is completely idiosyncratic but becomes worse when $\mathcal{E}_{j,t}$ is partially shared. Third, the difference in performance between time series forecasters is visible here; the \texttt{last\_constant} and \texttt{arima110} forecasters perform better than the \texttt{ar} and \texttt{hamilton} forecasters. The $q=1$ HSC also performs better than the $q=2$ configuration. These differences in configurations reflect the design of the DGP, as the idiosyncratic stochastic trend is an ARIMA(1,1,0) model and $q=1$ suffices to control the spurious matching. Appendix~\ref{app:mc_perperiod} shows that this pooled ranking holds horizon by horizon for the strongest configurations, the constant carry-forward and the $q=1$ ARIMA$(1,1,0)$ forecaster, and Appendix~\ref{app:mc_bias_var} attributes the pooled advantage to a variance reduction that more than offsets a small bias penalty.

\subsection{Cross-validation adapts to the regime}
\label{subsec:mc_rho}

The argument in Section~\ref{sec:predict_decompose} predicts that
the optimal $\rho$ depends on the data: when the stochastic trend is shared across units, level matching identifies donor
weights well and $\rho$ near one is optimal; when the stochastic trend is idiosyncratic, the donor pool cannot reproduce it and
filtering it out by pushing $\rho$ toward zero is preferred.
Figure~\ref{fig:mc_rho} reports the distribution of the
cross-validated $\hat\rho$ across the same $(\kappa,\rho_u)$ grid
for the four HSC forecasters at both smoothness orders.

\begin{figure}[!ht]
\caption{Distribution of cross-validated $\hat\rho$ across the
$(\kappa,\rho_u)$ grid}\label{fig:mc_rho}
\centering
\begin{minipage}{1\linewidth}{
\centering
\includegraphics[width=\textwidth]{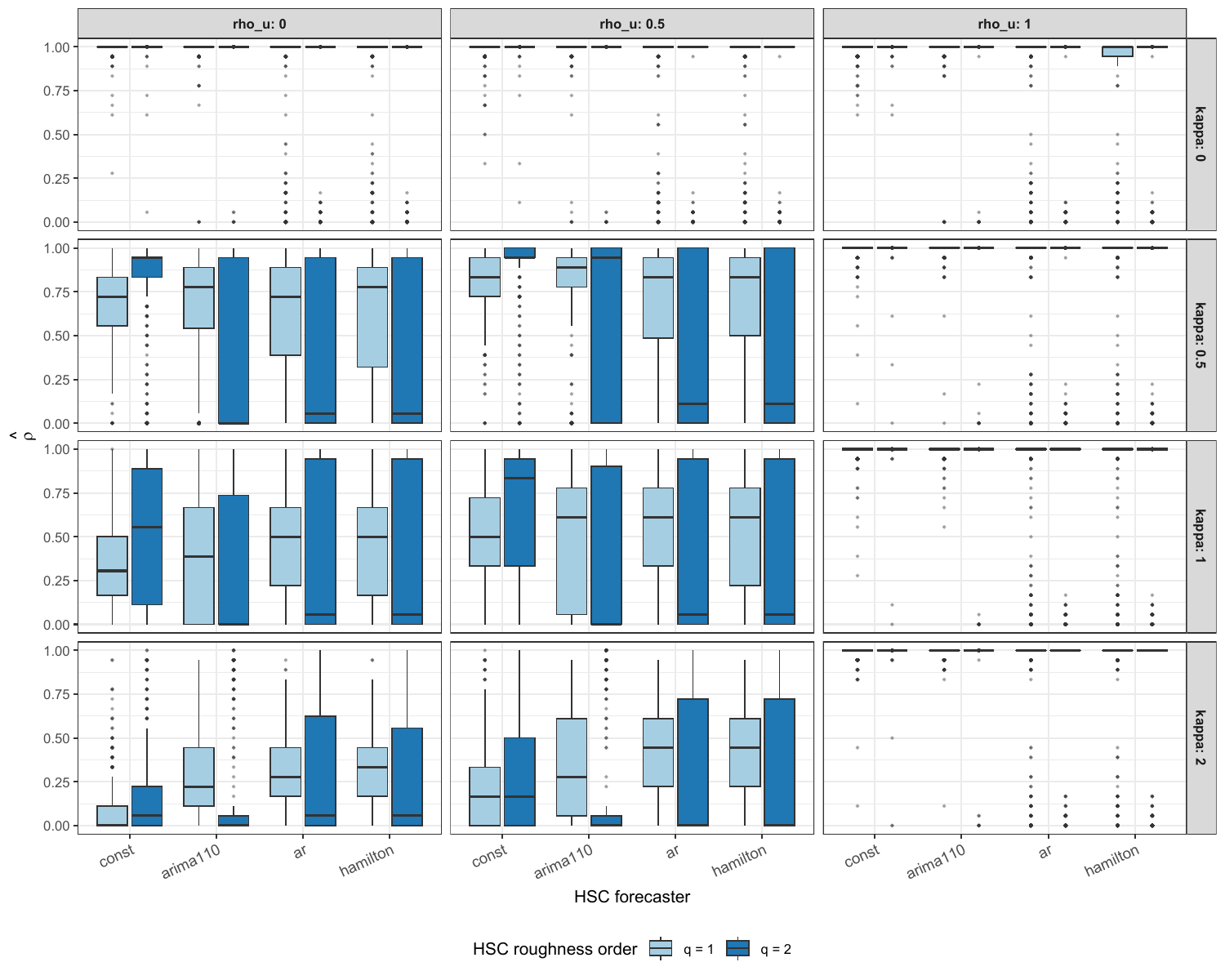}
}
\footnotesize\textit{Notes:} Boxplots report the distribution of
the cross-validated $\hat\rho$ across $R=500$ replications for each
$(\kappa,\rho_u)$ cell, for the four HSC time series forecasters at
smoothness orders $q=1$ (lighter shade) and $q=2$ (darker shade).
$T_0=200$, $N_0=50$, cross-validation horizon $h=1$.
\end{minipage}
\end{figure}

The cross-validated selection behaves as the theory predicts.
When the idiosyncratic stochastic trend is absent ($\kappa=0$, the top row of the
figure or $\rho_u=1$, the right column of the figure) the distribution of $\hat\rho$ almost piles up at one for every
forecaster and both smoothness orders: with no 
idiosyncratic stochastic trend component to filter out, the CV objective rewards
matching in levels.  When the stochastic trend is
present and at least partially unit-specific ($\kappa\ge 1$,
$\rho_u\le 0.5$, bottom-left region of the figure), $\hat\rho$
shifts downward: medians fall to around $0.5$ at $\kappa=1$ and
near zero at $\kappa=2$ with $\rho_u=0$, with substantial dispersion
across replications.  The shift happens for all four forecasters
and at both smoothness orders, confirming that the rolling-origin CV
identifies the correct allocation between donor matching and time series forecaster directly from the data.  Two further
properties of the cross-validated allocation are deferred to the
appendix: Appendix~\ref{app:mc_cv_h} reports how $\hat\rho$ and
post-period RMSE shift when the CV horizon $h$ is extended from
one to twenty.

\section{Empirical Application: The 1997 Handover of Hong Kong}
\label{sec:empirical}

We illustrate HSC on the study of the per-capita GDP after the 1997 return of Hong Kong to Chinese
sovereignty.  The example was introduced by \citet{hsiao2012panel} and
revisited by \citet{shi2025synthetic}.
\citet{hsiao2012panel} difference the data to a stationary growth-rate
outcome and select a small set of geographically and economically
proximate donors, including mainland China and Hong Kong's Asian trading
partners, by an information criterion, fitting an unrestricted
regression of the treated series on the selected donors.
\citet{shi2025synthetic} instead work with the nonstationary level of
annual per-capita GDP, restrict the donor pool to developed economies
with comparable long-run growth, and impose a hard separation between a
treated-unit trend forecast from Hong Kong's own history and a
donor-matched cyclical component.   Throughout we use the sign convention
$\hat\tau_t = Y_{1t}-\hat Y_{1t}(0)$, so a negative value means
observed Hong Kong GDP lies below the estimated no-handover
counterfactual.  The main text uses the annual data of
\citet{shi2025synthetic} so that the comparison with the most closely
related estimator is exact; Appendix~\ref{app:hk_robust} reports
robustness to the cross-validation horizon and to the
\citet{hsiao2012panel} geographic-neighbour donor pool.

\subsection{Data and cross-validated configuration}
\label{subsec:hk_setup}

The panel is the one assembled by \citet{shi2025synthetic}: annual real
per-capita GDP for Hong Kong and eleven developed donor
economies, including Australia, Austria, Canada, Denmark, France, Germany,
Italy, Korea, the Netherlands, New Zealand, and the United
States over $1961$--$2003$.  The treatment year is $1997$, giving
$T_0=36$ pre-treatment years ($1961$--$1996$) and $T_\post=7$
post-treatment years ($1997$--$2003$); the United Kingdom and mainland
China are excluded as parties directly involved in the handover, and
economies exposed to the $1997$--$98$ Asian financial crisis or with
heterogeneous welfare-state structures are excluded, following
\citet{shi2025synthetic}.  The panel-data implementation of
\citet{hsiao2012panel}, which instead selects geographically proximate
Asian donors by an information criterion on differenced data, is
examined as a robustness check in Appendix~\ref{app:hk_robust}.

We evaluate HSC under four configurations: the \texttt{last\_constant}
and ARIMA$(1,1,0)$ forecasters, each at roughness orders
$q\in\{1,2\}$, which are the configurations that performed well in the Monte
Carlo study of Section~\ref{sec:simulation}.  The tuning parameter
$\hat\rho$ is selected by rolling-origin cross-validation with
one-step-ahead horizon ($h=1$), a $21$-point $\rho$-grid, and the
SDID-style ridge $\zeta=T_\post^{1/4}\hat\sigma_{\Delta X}$.
Figure~\ref{fig:hk_cv} reports the cross-validated mean squared
prediction error along the $\rho$-grid for the four configurations.
The cross-validation selects an interior optimum: the best
configuration is ARIMA$(1,1,0)$ at $q=1$ with $\hat\rho=0.11$.

\begin{figure}[!ht]
\caption{Hong Kong: cross-validated MSPE for the four HSC
configurations}\label{fig:hk_cv}
\centering
\begin{minipage}{1\linewidth}{
\centering
\includegraphics[width=\textwidth]{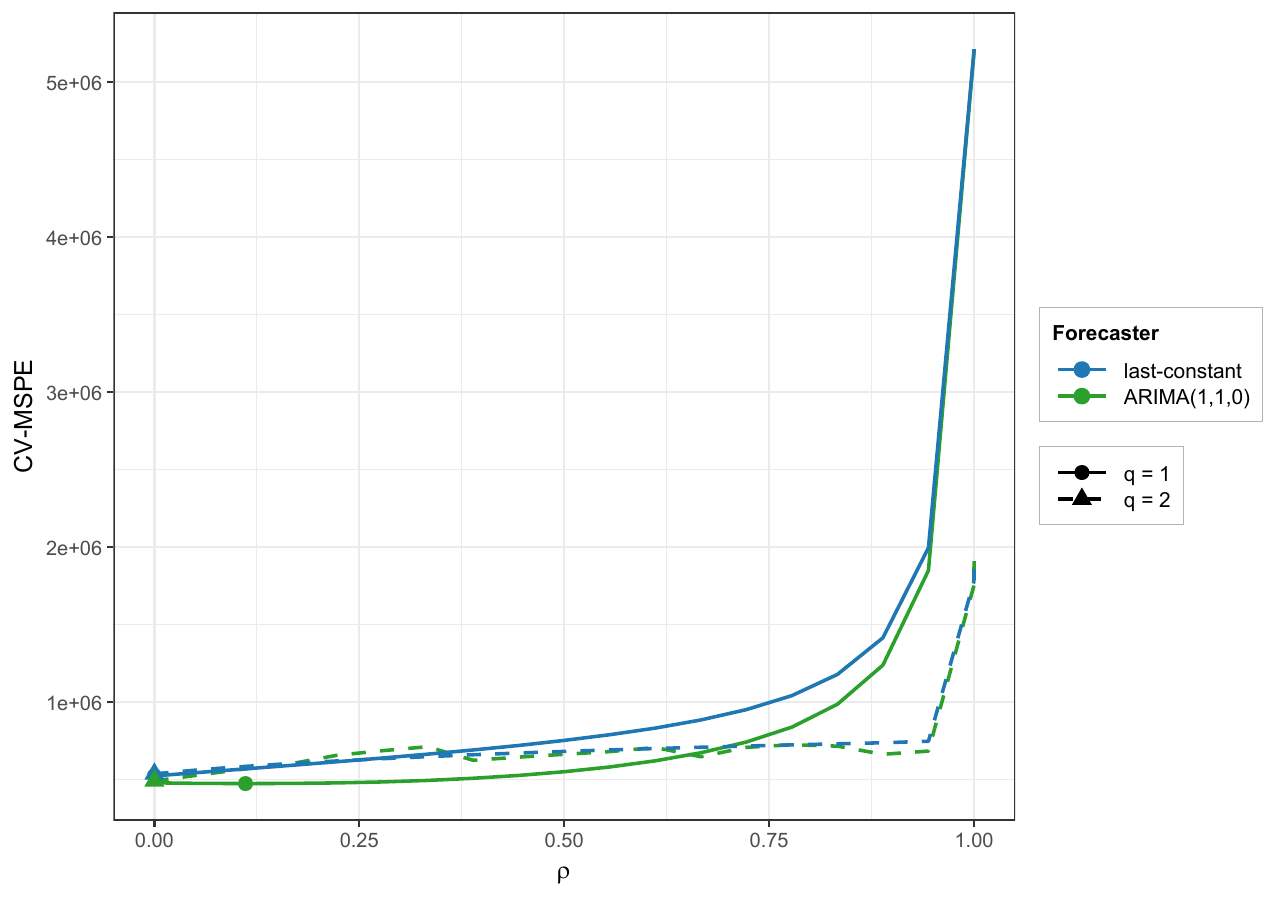}
}
\footnotesize\textbf{Note:} Cross-validated mean squared prediction
error (rolling-origin, one-step-ahead, $\texttt{cv\_last}=15$, $21$
folds) as a function of the spectral allocation parameter $\rho$, for
the four HSC configurations (\texttt{last\_constant} and
ARIMA$(1,1,0)$ forecasters at roughness orders $q=1$ and $q=2$).  A
marker on each curve denotes that configuration's cross-validated
$\hat\rho$.  The selected configuration is ARIMA$(1,1,0)$, $q=1$,
$\hat\rho=0.11$.  Sample: Hong Kong plus eleven developed donor
economies, annual per-capita GDP, $1961$--$1996$ pre-treatment
fitting window.
\end{minipage}
\end{figure}

\subsection{Counterfactual comparison}
\label{subsec:hk_cf}

Figure~\ref{fig:hk_cf} overlays the cross-validation-selected HSC
counterfactual with those of plain synthetic control
\citep[SC,][]{abadie2010synthetic}, synthetic control with an
intercept \citep[SC-INT,][]{doudchenko2016balancing}, synthetic
difference-in-differences \citep[SDID,][]{arkhangelsky2021synthetic},
and the synthetic business-cycle estimator of \citet{shi2025synthetic}
(SBCA-Hamilton), together with observed Hong Kong GDP.  The estimators
fall into three groups.  The HSC counterfactual tracks observed Hong
Kong closely throughout the post-treatment window, reaching about
\$30{,}000 by $2003$ against an observed \$28{,}100, an implied
$\hat\tau_{2003}\approx-\$1{,}900$.  SC, SC-INT, and SDID drift
moderately above the observed series.  SBCA-Hamilton diverges sharply:
because it forecasts Hong Kong's post-$1997$ trend by a recursive
linear projection of its own pre-$1997$ history, and Hong Kong's
pre-handover growth was unusually steep, that projection rises to
roughly \$36{,}100 by $2003$, an implausibly large effect.  

\begin{figure}[!ht]
\caption{Hong Kong: counterfactual per-capita GDP by
estimator}\label{fig:hk_cf}
\centering
\includegraphics[width=0.8\textwidth]{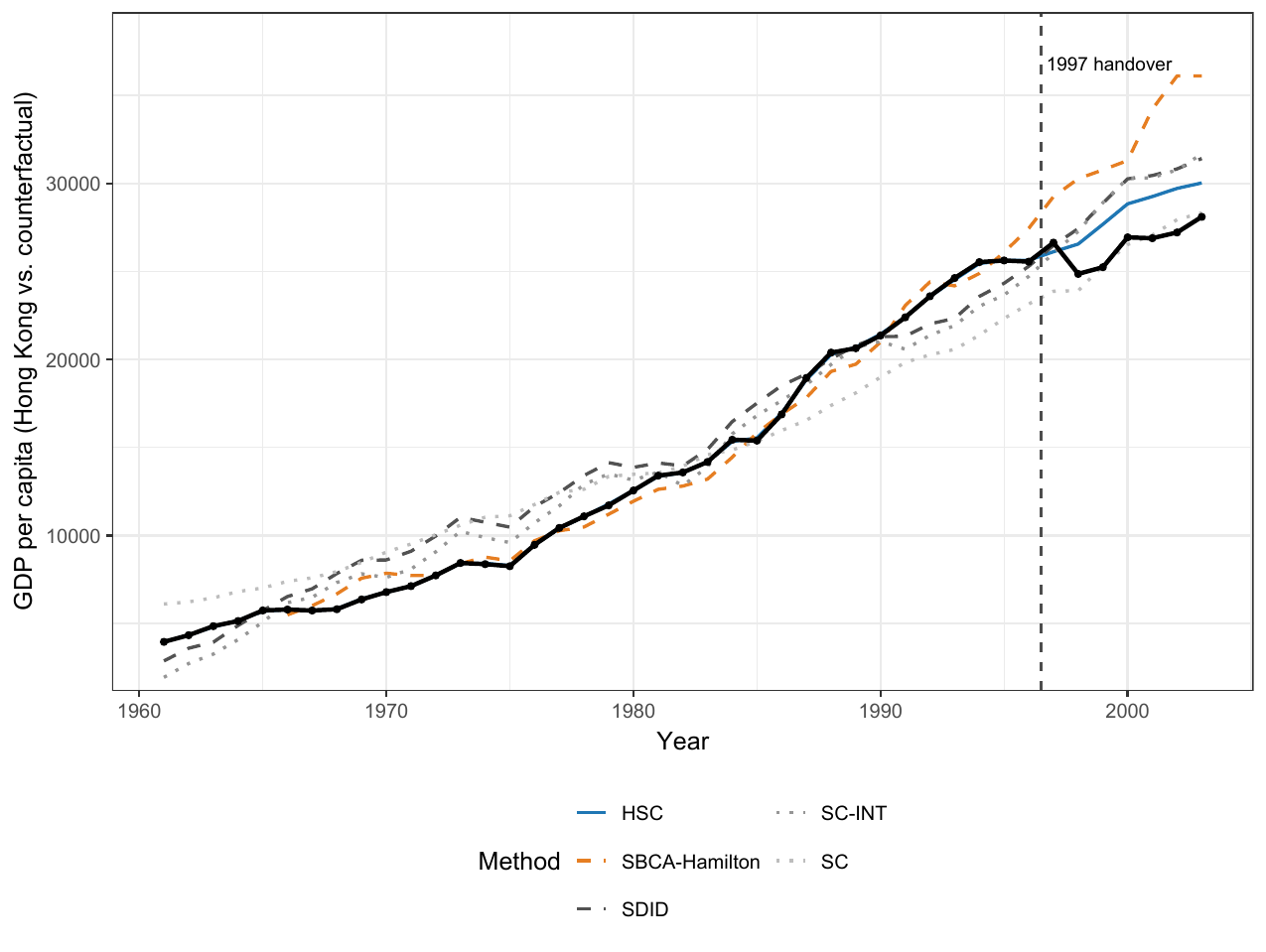}
\begin{minipage}{1\linewidth}
\footnotesize\textbf{Note:} Observed Hong Kong per-capita GDP (solid
black, with markers) and estimated no-handover counterfactuals from
HSC (the cross-validation-selected ARIMA$(1,1,0)$, $q=1$,
$\hat\rho=0.11$), SBCA-Hamilton \citep{shi2025synthetic}, SDID
\citep{arkhangelsky2021synthetic}, SC-INT
\citep{doudchenko2016balancing}, and plain SC
\citep{abadie2010synthetic}.  The vertical dashed line marks the
$1997$ handover.  SBCA-Hamilton's post-$1997$ trend is a recursive
linear projection of Hong Kong's pre-$1997$ series and rises well
above the observed path.  Sample: eleven developed donor economies,
annual per-capita GDP, $1961$--$2003$ ($T_0=36$, $T_\post=7$).
\end{minipage}
\end{figure}

\subsection{Donor-weight diversification}
\label{subsec:hk_weights}

Figure~\ref{fig:hk_weights} compares the donor weights that HSC, SDID,
SC-INT, and SBCA-Hamilton assign across the eleven donors.  The
contrast is stark.  HSC distributes weight broadly across all eleven
economies, with no single weight exceeding $0.19$ (the largest are
Korea $0.18$, Germany $0.14$, the United States $0.13$, and Italy
$0.11$).  SC-INT collapses onto a corner solution, placing $0.91$ on
the United States and $0.09$ on Korea.  SBCA-Hamilton concentrates on
four donors (Italy $0.43$, Germany $0.25$, Korea $0.18$, the United
States $0.09$).  SDID is intermediate: its ridge penalty
de-concentrates the weights relative to SC-INT. The United States
weight falls from $0.91$ to $0.56$ and mass spreads to Denmark
($0.23$), Korea, and Germany.  

\begin{figure}[!ht]
\caption{Hong Kong: donor weights by estimator}\label{fig:hk_weights}
\centering
\includegraphics[width=0.8\textwidth]{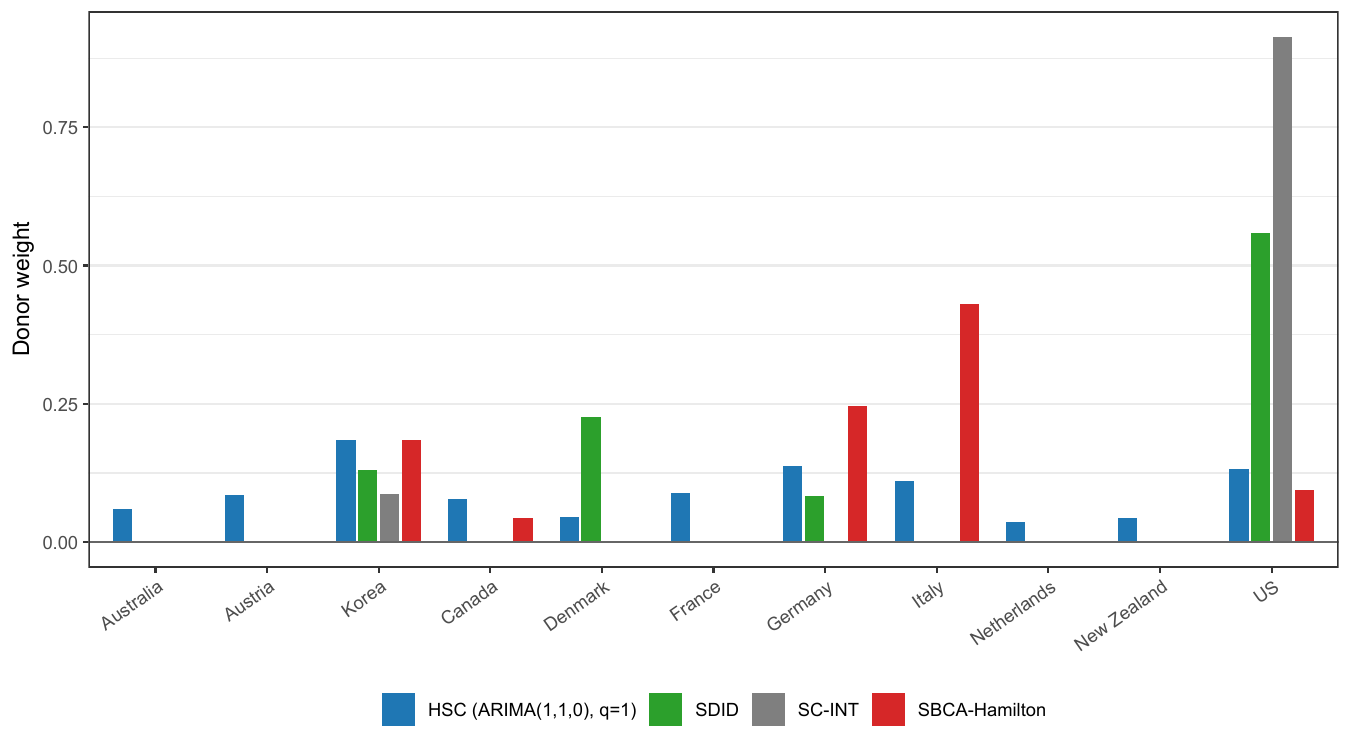}
\begin{minipage}[]{1\linewidth}
\footnotesize\textbf{Note:} Donor weights assigned to the eleven
developed donor economies by the cross-validation-selected HSC
configuration (ARIMA$(1,1,0)$, $q=1$, $\hat\rho=0.11$), SDID, SC-INT,
and SBCA-Hamilton.  All four estimators constrain the donor weights to
sum to one.  HSC spreads weight across all eleven donors; SC-INT and
SBCA-Hamilton concentrate on one and four donors respectively; SDID's
ridge penalty places it between these extremes.  Sample: annual
per-capita GDP, $1961$--$1996$ pre-treatment fitting window.
\end{minipage}
\end{figure}

\subsection{Out-of-sample accuracy}
\label{subsec:hk_oos}

Pre-treatment fit cannot discriminate among these estimators, because
each minimizes a different in-sample criterion.  We therefore evaluate
every method by the same rolling-origin, one-step-ahead
cross-validated MSPE used to select $\hat\rho$.  HSC is the most
accurate method by a wide margin: at $h=1$ the selected configuration
attains a CV-MSPE of $4.8\times10^{5}$, and all four HSC
configurations ($4.8$--$5.4\times10^{5}$) fall below every competing
estimator---SBCA-Hamilton ($1.2\times10^{6}$), SDID
($1.5\times10^{6}$), SC-INT ($3.6\times10^{6}$), and plain SC
($9.0\times10^{6}$).  HSC thus improves on the synthetic business-cycle
estimator by roughly a factor of two and a half, and on the
level-matching estimators by one to two orders of magnitude, on a
criterion that uses only pre-treatment data.  Appendix~\ref{app:hk_robust}
shows that this ranking is preserved when the cross-validation horizon
is lengthened to $h=4$ and when the donor pool is replaced by the
geographic-neighbour pool of \citet{hsiao2012panel}.

\clearpage

\section{Conclusion}
\label{sec:conclusion}

Harmonic synthetic control (HSC) addresses counterfactual estimation when untreated outcomes may contain both shared and idiosyncratic stochastic trends, a regime that the researcher cannot reliably distinguish ex ante. Instead of committing in advance to matching in the raw level or to differencing before matching, HSC introduces a treated-unit-specific smooth component and a single tuning parameter that rolling-origin cross-validation uses to allocate predictive responsibility between donor matching and time series forecaster. The spectral interpretation, the prediction-error decomposition, the Monte Carlo evidence, and the Hong Kong application all point to the same conclusion: a soft, data-driven allocation is more robust and can adapt to different regimes.

The present paper develops and evaluates the HSC point estimator; formal uncertainty quantification is the natural next step. A prediction interval for the HSC counterfactual should combine donor-weight estimation uncertainty with out-of-sample forecast-error calibration for the smooth component, extending the synthetic-control prediction-interval framework of \citet{cattaneo2021prediction} to the soft-allocation setting. Because that construction rests on additional assumptions beyond those required for the point estimator, we leave it to future works.

\clearpage
\onehalfspacing
\bibliographystyle{apsr}
\bibliography{scm}
\clearpage

\appendix
\onehalfspacing
\setcounter{page}{1}
\setcounter{table}{0}
\setcounter{figure}{0}
\setcounter{equation}{0}
\setcounter{footnote}{0}
\setcounter{lemma}{0}
\setcounter{proposition}{0}

\renewcommand{\theassumption}{A\arabic{assumption}}
\renewcommand\thetable{A\arabic{table}}
\renewcommand\thefigure{A\arabic{figure}}
\renewcommand{\thepage}{A-\arabic{page}}
\renewcommand{\theequation}{A\arabic{equation}}
\renewcommand{\thefootnote}{A\arabic{footnote}}
\renewcommand{\thelemma}{A\arabic{lemma}}

\begin{center}
    \Large\bf Online Supplementary Materials
\end{center}
\bigskip

\section{Proofs}
\label{app:proofs}

This appendix collects the proofs of all propositions, lemmas, and
corollaries stated in the main text.  Section~\ref{app:proofs_sec3}
covers the results of Section~\ref{sec:estimator} (the HSC estimator
and its endpoint properties), and
Section~\ref{app:proofs_sec5} covers the results of
Section~\ref{sec:predict_decompose} (the prediction-error
decomposition and the Term~A and Term~B envelopes).

\subsection{Proofs for Section~\ref{sec:estimator}}
\label{app:proofs_sec3}

\subsubsection{Proof of Proposition~\ref{prop:hsc_profiled}}

\begin{proof}
Fix $q\in\{1,2\}$, $\rho\in(0,1)$, and $\omega\in\Delta_{N_0}$, and write
$r:=r_\pre(\omega)=Y_{\pre}-X_\pre\omega$. The inner problem is
\begin{equation}
\min_{E\in\R^{T_0}} f(E),
\qquad
f(E):=\frac{1}{\rho}\|r-E\|_2^2+\frac{1}{1-\rho}E'K_qE.
\label{eq:app_profile_obj}
\end{equation}
The Hessian of $f$ is $\frac{2}{\rho}I_{T_0}+\frac{2}{1-\rho}K_q\succ 0$
(since $\rho\in(0,1)$), so $f$ is strictly convex and the minimizer is unique.

\paragraph*{Step 1: First-order condition.}
Setting $\nabla_E f=0$ gives
\[
\frac{1}{\rho}(\hat E-r)+\frac{1}{1-\rho}K_q\hat E=0.
\]
Multiplying by $\rho$ and writing $\lambda_\rho:=\rho/(1-\rho)$, this becomes
$(I_{T_0}+\lambda_\rho K_q)\hat E=r$, whence
\begin{equation}
\hat E = S_{\rho,q}\,r,
\qquad
S_{\rho,q}:=(I_{T_0}+\lambda_\rho K_q)^{-1}.
\label{eq:app_Ehat}
\end{equation}

\paragraph*{Step 2: Profiled objective value.}
The FOC can be rewritten as
\begin{equation}
\frac{1}{\rho}(r-\hat E)=\frac{1}{1-\rho}K_q\hat E.
\label{eq:app_foc_rewrite}
\end{equation}
Using~\eqref{eq:app_foc_rewrite} to substitute for the first term
evaluated at $\hat E$,
\begin{align}
f(\hat E)
&=
(r-\hat E)'\!\cdot\!\frac{1}{\rho}(r-\hat E)
\;+\;
\hat E'\!\cdot\!\frac{1}{1-\rho}K_q\hat E
\notag\\
&=
(r-\hat E)'\!\cdot\!\frac{1}{1-\rho}K_q\hat E
\;+\;
\hat E'\!\cdot\!\frac{1}{1-\rho}K_q\hat E
\notag\\
&=
r'\!\cdot\!\frac{1}{1-\rho}K_q\hat E
\notag\\
&=
\frac{1}{1-\rho}\,r'K_qS_{\rho,q}\,r.
\label{eq:app_fmin}
\end{align}
Now observe that
\[
\frac{1}{\rho}(I_{T_0}-S_{\rho,q})
=
\frac{1}{\rho}\,\lambda_\rho K_q S_{\rho,q}
=
\frac{1}{1-\rho}\,K_q S_{\rho,q},
\]
where the first equality uses
$(I_{T_0}+\lambda_\rho K_q)S_{\rho,q}=I_{T_0}$, which rearranges to
$I_{T_0}-S_{\rho,q}=\lambda_\rho K_q S_{\rho,q}$.\footnote{Because
$S_{\rho,q}=(I_{T_0}+\lambda_\rho K_q)^{-1}$ is a function of
$K_q$, the two operators commute and $K_qS_{\rho,q}=S_{\rho,q}K_q$,
so either ordering may be used throughout.}
Therefore,
defining $W_{\rho,q}:=\frac{1}{\rho}(I_{T_0}-S_{\rho,q})$,
\begin{equation}
f(\hat E)=r'W_{\rho,q}\,r.
\label{eq:app_fmin_W}
\end{equation}

\paragraph*{Step 3: Profiled weight problem.}
In the full HSC criterion~\eqref{eq:hsc_primal}, only $f(E)$ involves $E$.
Replacing $f(E)$ by its minimum~\eqref{eq:app_fmin_W} yields
\[
\hat\omega(\rho,q)
\;\in\;
\argmin_{\omega\in\Delta_{N_0}}
\left\{
r_\pre(\omega)'W_{\rho,q}\,r_\pre(\omega)
+\zeta^2T_0\|\omega\|_2^2
\right\},
\]
which is~\eqref{eq:hsc_profiled}. Substituting $\hat\omega(\rho,q)$
into~\eqref{eq:app_Ehat} gives~\eqref{eq:Ehat_profiled}.
\end{proof}

\subsubsection{Proof of Proposition~\ref{prop:endpoints}}

\begin{proof}
Fix $q\in\{1,2\}$ and write $K=K_q$, $P_0=P_{0,q}$. Let $0=\mu_1=\cdots=\mu_d<\mu_{d+1}\le\cdots\le\mu_{T_0}$ be the eigenvalues of $K$ (with $d=\dim\Null(K)=q$), and let $v_1,\dots,v_{T_0}$ be a corresponding orthonormal eigenbasis. For $\rho\in(0,1)$, the definitions $\lambda_\rho=\rho/(1-\rho)$ and $S_{\rho,q}=(I+\lambda_\rho K)^{-1}$ give
\[
S_{\rho,q}\,v_k = \frac{1}{1+\lambda_\rho\mu_k}\,v_k,
\qquad
W_{\rho,q}\,v_k = \frac{1}{\rho}\Bigl(1-\frac{1}{1+\lambda_\rho\mu_k}\Bigr)v_k
= \frac{\mu_k}{(1-\rho)+\rho\mu_k}\,v_k.
\]

\paragraph*{Part (i): limits.}
For each eigenvalue $\mu_k$, define $s_k(\rho):=1/(1+\lambda_\rho\mu_k)$ and $w_k(\rho):=\mu_k/\bigl((1-\rho)+\rho\mu_k\bigr)$.

\emph{As $\rho\downarrow 0$}: $\lambda_\rho\to 0$, so $s_k(\rho)\to 1$ for every $k$. Hence $S_{\rho,q}\to I_{T_0}$. Similarly, $w_k(\rho)\to\mu_k$ for every $k$, so $W_{\rho,q}\to K=K_q$.

\emph{As $\rho\uparrow 1$}: $\lambda_\rho\to\infty$. For the null-space eigenvectors ($\mu_k=0$, $k\le d$): $s_k(\rho)=1$ and $w_k(\rho)=0$ for all $\rho$. For the positive eigenvectors ($\mu_k>0$, $k>d$): $s_k(\rho)=1/(1+\lambda_\rho\mu_k)\to 0$ and $w_k(\rho)=\mu_k/((1-\rho)+\rho\mu_k)\to 1$. Therefore $S_{\rho,q}\to P_{0,q}$ and $W_{\rho,q}\to I_{T_0}-P_{0,q}$.

Since these limits agree with the boundary definitions~\eqref{eq:S_boundary} and~\eqref{eq:W_boundary}, both families admit unique continuous extensions to $[0,1]$.

\paragraph*{Part (ii): positive semidefiniteness.}
For $\rho\in(0,1)$, the eigenvalues $w_k(\rho)\ge 0$ for all $k$, so $W_{\rho,q}\succeq 0$. At the endpoints, $W_{0,q}=K=D_q'D_q\succeq 0$ and $W_{1,q}=I_{T_0}-P_{0,q}$ is an orthogonal projector, hence positive semidefinite.
\end{proof}

\subsection{Proofs for Section~\ref{sec:predict_decompose}}
\label{app:proofs_sec5}

\subsubsection{Proof of Proposition~\ref{prop:AB}}
\label{app:proof_AB}

\paragraph*{Term~A.}
From the oracle predictor~\eqref{eq:oracle_pred} and the HSC
counterfactual:
\begin{align*}
\mathrm{Term~A}(\rho)
&=
\widetilde Y_\post^\oracle(0;\rho) - \hat Y_\post(0;\rho,q)\\
&=
X_\post\omega^\oracle + \Pi_\rho\,r_\pre^\oracle
-X_\post\hat\omega - \Pi_\rho\,r_\pre(\hat\omega)\\
&=
X_\post(\omega^\oracle-\hat\omega)
+\Pi_\rho\bigl(r_\pre^\oracle - r_\pre(\hat\omega)\bigr).
\end{align*}
Since $r_\pre(\omega)=Y_\pre-X_\pre\omega$ is linear in $\omega$,
$r_\pre^\oracle - r_\pre(\hat\omega)
= -X_\pre(\omega^\oracle-\hat\omega)$, so
\[
\mathrm{Term~A}(\rho)
=
(X_\post - \Pi_\rho X_\pre)(\omega^\oracle-\hat\omega).
\]

\paragraph*{Term~B.}
From the definition~\eqref{eq:AB} and the oracle
predictor~\eqref{eq:oracle_pred}:
\begin{align*}
\mathrm{Term~B}(\rho)
&=
Y_\post(0) - \widetilde Y_\post^\oracle(0;\rho)\\
&=
Y_\post(0) - X_\post\omega^\oracle - \Pi_\rho\,r_\pre^\oracle\\
&=
(Y_\post(0) - X_\post\omega^\oracle)
- \Pi_\rho(Y_\pre - X_\pre\omega^\oracle).
\end{align*}
\hfill$\square$

\subsubsection{Gradient decomposition}
\label{app:grad_decomp}

The proof of Proposition~\ref{prop:TermA_env} relies on an exact
decomposition of the gradient discrepancy between the HSC and oracle
objectives at the oracle weights.

\begin{lemma}[Gradient decomposition at $\omega^\oracle$]
\label{lem:grad_decomp}
Let $F_\rho$ and $H^\oracle$ be defined
by~\eqref{eq:F_def}--\eqref{eq:H_def}.  Then
\begin{equation}
\frac{1}{2}\bigl(\nabla F_\rho(\omega^\oracle)
-\nabla H^\oracle(\omega^\oracle)\bigr)
\;=\;
-\underbrace{L_{0,\pre}'(W_\rho-P_\perp)\,
  e^L_\pre}_{g_1}
\;-\;
\underbrace{\cR_{0,\pre}'W_\rho\,
  e^L_\pre}_{g_2}
\;-\;
\underbrace{X_\pre'W_\rho\,
  e^\cR_\pre}_{g_3}.
\label{eq:grad_diff}
\end{equation}
\end{lemma}

\begin{proof}
The gradients are
\begin{align*}
\nabla F_\rho(\omega)
&=
-2X_\pre'W_\rho\,r_\pre(\omega)+2\zeta^2T_0\omega,\\
\nabla H^\oracle(\omega)
&=
-2L_{0,\pre}'P_\perp(L_{1,\pre}-L_{0,\pre}\omega)
+2\zeta^2T_0\omega.
\end{align*}
At $\omega=\omega^\oracle$, the ridge terms cancel.  For the oracle
objective,
$P_\perp(L_{1,\pre}-L_{0,\pre}\omega^\oracle)=P_\perp\,e^L_\pre$.
For the HSC objective,
$r_\pre(\omega^\oracle)=r_\pre^\oracle=e^L_\pre+e^\cR_\pre$.
Since $P_0\,e^L_\pre\in\Null(K)$ and $W_\rho$ annihilates
$\Null(K)$ (Proposition~\ref{prop:endpoints}),
$W_\rho\,r_\pre^\oracle=W_\rho\,e^L_\pre+W_\rho\,e^\cR_\pre$.
Expanding $X_\pre=L_{0,\pre}+\cR_{0,\pre}$:
\begin{align*}
\tfrac{1}{2}\bigl(\nabla F_\rho(\omega^\oracle)
-\nabla H^\oracle(\omega^\oracle)\bigr)
&=
-X_\pre'W_\rho(e^L_\pre+e^\cR_\pre)
+L_{0,\pre}'P_\perp\,e^L_\pre\\
&=
-L_{0,\pre}'W_\rho\,e^L_\pre
-\cR_{0,\pre}'W_\rho\,e^L_\pre
-X_\pre'W_\rho\,e^\cR_\pre
+L_{0,\pre}'P_\perp\,e^L_\pre\\
&=
-L_{0,\pre}'(W_\rho-P_\perp)\,e^L_\pre
-\cR_{0,\pre}'W_\rho\,e^L_\pre
-X_\pre'W_\rho\,e^\cR_\pre,
\end{align*}
where the last line groups
$-L_{0,\pre}'W_\rho\,e^L_\pre+L_{0,\pre}'P_\perp\,e^L_\pre
=-L_{0,\pre}'(W_\rho-P_\perp)\,e^L_\pre$.
\end{proof}

\subsubsection{Constrained minimizer gap}

\begin{lemma}\label{lem:Q_gap}
Let
$\hat\omega(\rho)\in\argmin_{\omega\in\Delta_{N_0}}F_\rho(\omega)$
and
$\omega^\oracle\in\argmin_{\omega\in\Delta_{N_0}}H^\oracle(\omega)$.
Then
\begin{equation}
\bigl\|\omega^\oracle-\hat\omega(\rho)\bigr\|_{Q_\rho}
\;\le\;
\frac{1}{2T_0}\,
\bigl\|\nabla F_\rho(\omega^\oracle)
-\nabla H^\oracle(\omega^\oracle)\bigr\|_{Q_\rho^{-1}}.
\label{eq:lem_Qgap}
\end{equation}
\end{lemma}

\begin{proof}
Let $d:=\omega^\oracle-\hat\omega(\rho)$.  The minimizers of the
convex differentiable functions $F_\rho$ and $H^\oracle$ over the
closed convex set $\Delta_{N_0}$ satisfy the variational inequalities
\[
\bigl\langle\nabla F_\rho(\hat\omega),\,
  \omega-\hat\omega\bigr\rangle\ge 0
\quad\forall\omega\in\Delta_{N_0},
\qquad
\bigl\langle\nabla H^\oracle(\omega^\oracle),\,
  \omega-\omega^\oracle\bigr\rangle\ge 0
\quad\forall\omega\in\Delta_{N_0}.
\]
Setting $\omega=\omega^\oracle$ in the first and
$\omega=\hat\omega$ in the second, then adding:
\[
\bigl\langle\nabla F_\rho(\hat\omega)
-\nabla H^\oracle(\omega^\oracle),\,d\bigr\rangle\ge 0.
\]
Inserting $\pm\nabla F_\rho(\omega^\oracle)$ and using
$\nabla F_\rho(\hat\omega)-\nabla F_\rho(\omega^\oracle)
=-2T_0Q_\rho d$ (since $F_\rho$ is quadratic with Hessian
$2T_0Q_\rho$):
\[
2T_0\|d\|_{Q_\rho}^2
\le
\bigl\langle\nabla F_\rho(\omega^\oracle)
-\nabla H^\oracle(\omega^\oracle),\,d\bigr\rangle
\le
\bigl\|\nabla F_\rho(\omega^\oracle)
-\nabla H^\oracle(\omega^\oracle)\bigr\|_{Q_\rho^{-1}}
\|d\|_{Q_\rho},
\]
where the last step is the generalized Cauchy--Schwarz inequality.
Dividing both sides by $2T_0\|d\|_{Q_\rho}$ (the claim is trivial
if $d=0$) yields~\eqref{eq:lem_Qgap}.
\end{proof}

\subsubsection{Dual-norm bound}

\begin{lemma}[Dual-norm bound]\label{lem:dual_norm}
For any $u\in\R^{T_0}$,
\begin{equation}
\frac{1}{T_0}\bigl\|X_\pre'W_\rho u\bigr\|_{Q_\rho^{-1}}
\;\le\;
\frac{1}{\sqrt{T_0}}\|u\|_{W_\rho}.
\label{eq:lem_dual_norm}
\end{equation}
\end{lemma}

\begin{proof}
Set $b:=W_\rho^{1/2}u$ and
$\widetilde X:=T_0^{-1/2}W_\rho^{1/2}X_\pre$, so that
$Q_\rho=\widetilde X'\widetilde X+\zeta^2I_{N_0}$ and
$X_\pre'W_\rho u=\sqrt{T_0}\,\widetilde X'b$.  Then
\[
\frac{1}{T_0^2}\bigl\|X_\pre'W_\rho u\bigr\|_{Q_\rho^{-1}}^2
=
\frac{1}{T_0}\,b'\widetilde X
(\widetilde X'\widetilde X+\zeta^2I)^{-1}\widetilde X'b.
\]
Using a compact SVD $\widetilde X=U\Sigma V'$, we have
$\widetilde X
(\widetilde X'\widetilde X+\zeta^2I)^{-1}\widetilde X'
=U\diag\!\bigl(\sigma_i^2/(\sigma_i^2+\zeta^2)\bigr)U'
\preceq I_{T_0}$, so the expression is bounded by
$T_0^{-1}\|b\|_2^2=T_0^{-1}\|u\|_{W_\rho}^2$.
\end{proof}

\subsubsection{Proof of Proposition~\ref{prop:TermA_env}}

\begin{proof}
The proof proceeds in three steps.

\emph{Step~1: gradient decomposition.}
By Lemma~\ref{lem:grad_decomp}, the gradient discrepancy
at $\omega^\oracle$ decomposes as
$\frac{1}{2}(\nabla F_\rho(\omega^\oracle)
-\nabla H^\oracle(\omega^\oracle))=-(g_1+g_2+g_3)$,
with $g_1$, $g_2$, $g_3$ as in~\eqref{eq:grad_diff}.

\emph{Step~2: weight gap bound.}
From Lemma~\ref{lem:Q_gap} and the triangle inequality:
\[
\bigl\|\omega^\oracle-\hat\omega\bigr\|_{Q_\rho}
\le
\frac{1}{T_0}\|g_1\|_{Q_\rho^{-1}}
+\frac{1}{T_0}\|g_2\|_{Q_\rho^{-1}}
+\frac{1}{T_0}\bigl\|X_\pre'W_\rho\,
e^\cR_\pre\bigr\|_{Q_\rho^{-1}}.
\]
The first two terms equal $A_1$ and $A_2$ by
definition~\eqref{eq:A1_def}--\eqref{eq:A2_def}.  Applying
Lemma~\ref{lem:dual_norm} with $u=e^\cR_\pre$ to the third term
gives $A_3$ as defined in~\eqref{eq:A3_def}.  Hence
$\|\omega^\oracle-\hat\omega\|_{Q_\rho}\le A_1+A_2+A_3$.

\emph{Step~3: transfer to prediction error.}
From~\eqref{eq:TA_closed}, insert
$Q_\rho^{-1/2}Q_\rho^{1/2}$:
\[
\bigl\|\mathrm{Term~A}(\rho)\bigr\|_2
=\bigl\|C_\rho(\omega^\oracle-\hat\omega)\bigr\|_2
\le\|C_\rho Q_\rho^{-1/2}\|_{\mathrm{op}}\,
\|\omega^\oracle-\hat\omega\|_{Q_\rho}
=\mathcal{P}_\rho\,
\|\omega^\oracle-\hat\omega\|_{Q_\rho}
\le\mathcal{P}_\rho(A_1+A_2+A_3).
\qedhere
\]
\end{proof}

\subsubsection{Proof of~\eqref{eq:TB_direct}}

\begin{proof}
From the closed form~\eqref{eq:TB_closed},
\[
\mathrm{Term~B}(\rho)
\;=\;
r_\post^\oracle\;-\;\Pi_\rho\,r_\pre^\oracle.
\]
Decompose the pre-period residual along
$\Null(K)$ and its orthogonal complement:
$r_\pre^\oracle = P_0\,r_\pre^\oracle + \eta_\pre^\oracle$
by~\eqref{eq:eta_def}.  Because $S_\rho$ fixes $\Null(K)$
(Proposition~\ref{prop:endpoints}),
$S_\rho\,(P_0\,r_\pre^\oracle) = P_0\,r_\pre^\oracle$, and because
$\widetilde G_q$ satisfies the null-space continuation
property $\widetilde G_q v = G_q^{\mathrm{null}}v$ for every
$v\in\Null(K)$ of Section~\ref{subsec:counterfactual},
\[
\Pi_\rho\,(P_0\,r_\pre^\oracle)
\;=\;
\widetilde G_q\,(P_0\,r_\pre^\oracle)
\;=\;
G_q^{\mathrm{null}}\,P_0\,r_\pre^\oracle.
\]
Therefore
\begin{align*}
\mathrm{Term~B}(\rho)
&=\; r_\post^\oracle
   \;-\;\Pi_\rho\,(P_0\,r_\pre^\oracle)
   \;-\;\Pi_\rho\,\eta_\pre^\oracle\\
&=\; r_\post^\oracle
   \;-\;G_q^{\mathrm{null}}\,P_0\,r_\pre^\oracle
   \;-\;\widetilde G_q\,S_\rho\,\eta_\pre^\oracle\\
&=\; \eta_\post^\oracle
   \;-\;\widetilde G_q\,S_\rho\,\eta_\pre^\oracle,
\end{align*}
using the definition of $\eta_\post^\oracle$ in~\eqref{eq:eta_def}.
\end{proof}


\clearpage
\setcounter{page}{1}
\setcounter{table}{0}
\setcounter{figure}{0}
\setcounter{equation}{0}
\setcounter{footnote}{0}
\setcounter{lemma}{0}
\setcounter{proposition}{0}

\renewcommand{\theassumption}{B\arabic{assumption}}
\renewcommand\thetable{B\arabic{table}}
\renewcommand\thefigure{B\arabic{figure}}
\renewcommand{\thepage}{B-\arabic{page}}
\renewcommand{\theequation}{B\arabic{equation}}
\renewcommand{\thefootnote}{B\arabic{footnote}}
\renewcommand{\thelemma}{B\arabic{lemma}}

\section{Monte Carlo evaluation of the prediction-error decomposition}
\label{app:mc_decomp}

This appendix reports a Monte Carlo study that evaluates the
prediction-error decomposition of Section~\ref{sec:predict_decompose}
along the $\rho$-grid.  The goals are: (i)~to display the three
Term~A channels of Section~\ref{subsec:TermA_channels} and the
identification geometry encoded in $Q_\rho^{-1}$ and
$\mathcal{P}_\rho$; (ii)~to exhibit how the $\rho$-profile of
Term~B varies with the forecaster $\widetilde G_q$
(Section~\ref{subsec:TermB}); and (iii)~to show how Term~A and
Term~B combine into the $\rho$-shape of the prediction error across
four $(q,\widetilde G_q)$ configurations under two DGP regimes.

\subsection{Setup}
\label{app:mc_setup}

\paragraph*{Data generating process.}
For each replication we simulate panel data
$Y_{it}=\Lambda_i F_t+\kappa\,R_{it}^{\mathrm{rw}}+\varepsilon_{it}$
on $i\in\{0,1,\dots,N_0\}$ (one treated unit and $N_0$ donors) and
$t\in\{1,\dots,T_0+T_\post\}$.  The single common factor
$F_t=\sum_{s\le t}u_s$ is a random walk with $u_s\sim N(0,1)$.  The
loadings $\Lambda_i\sim N(1,0.5^2)$ are drawn once per replication and
held fixed across $t$.  The unit-specific random walks
$R_{it}^{\mathrm{rw}}=\sum_{s\le t}v_{is}$ have $v_{is}\sim N(0,1)$.
The idiosyncratic short-run noise $\varepsilon_{it}\sim N(0,0.5^2)$ is
iid across $(i,t)$.  The residual is then
$\cR_{it}=\kappa R_{it}^{\mathrm{rw}}+\varepsilon_{it}$, and the
observed outcome decomposes as $Y_{it}=L_{it}+\cR_{it}$ with
$L_{it}=\Lambda_i F_t$.  We set $N_0=10$, $T_0=80$, $T_\post=5$, and
the ridge parameter $\zeta=T_\post^{1/4}\,\hat\sigma_{\Delta X}$ as in
\citet{arkhangelsky2021synthetic}, where $\hat\sigma_{\Delta X}$ is
the empirical standard deviation of the first-differenced donor
pre-period series.  

\paragraph*{Two experiments.}
We consider two values of $\kappa$, corresponding to the regimes of
Section~\ref{subsec:design_goal}:
\begin{itemize}
\item \emph{Shared stochastic trend regime} ($\kappa=0$): $\cR$
  reduces to idiosyncratic short-run noise.
\item \emph{Shared + idiosyncratic stochastic trend regime}
  ($\kappa=2$): each unit carries an independent random walk in
  addition to the idiosyncratic short-run noise.
\end{itemize}

\paragraph*{Four configurations.}
For each replication we estimate HSC under four configurations.  All
four use the admissible forecast operator
$\widetilde G_q=G_q^{\mathrm{null}}P_{0,q}+\hat G_q P_{\perp,q}$ of
Section~\ref{subsec:counterfactual}: the null-space part of the
pre-period residual, $P_{0,q}r$, is always continued by the canonical
operator $G_q^{\mathrm{null}}$ (a constant for $q=1$, an
intercept-plus-linear trend for $q=2$), while the data-driven
component $\hat G_q$ acts only on the non-null part $P_{\perp,q}r$.
The configurations differ in the smoothness order $q$ and in the
choice of $\hat G_q$:
\begin{itemize}
\item $(q=1,\ G_1^{\mathrm{const}})$: $\hat G_1$ is the constant
  carry-forward, holding the last value of the non-null part flat.
  Recombined with the continued null-space mean, the two parts
  collapse to the raw last residual $r_{T_0}$ held constant across the
  post-period window.
\item $(q=1,\ \mathrm{ARIMA}(1,1,0))$: $\hat G_1$ is an
  ARIMA$(1,1,0)$ model fitted to the non-null part of the pre-period
  residual, with the null-space mean continued exactly as above.
\item $(q=2,\ G_2^{\mathrm{const}})$: $\hat G_2$ is again the constant
  carry-forward on the non-null part.  Recombined with the canonical
  intercept-plus-linear continuation on $\Null(K_2)$, the composed
  forecast at horizon $h$ is $r_{T_0}+\hat\beta\,h$, where $\hat\beta$
  is the slope of the line fitted to the pre-period residual over
  $\Null(K_2)$; equivalently, the fitted linear trend is extrapolated
  with its level re-anchored to the last observed residual $r_{T_0}$.
\item $(q=2,\ \mathrm{ARIMA}(1,1,0))$: $\hat G_2$ is an
  ARIMA$(1,1,0)$ model on the non-null part, with the null-space
  intercept-plus-linear trend continued by $G_2^{\mathrm{null}}$.
\end{itemize}

\paragraph*{Grid in $\rho$.}
For every simulation round we sweep
$\rho$ over the 19-point grid

\noindent $\{0,0.05,0.1,0.2,0.3,\dots,0.8,0.85,0.9,0.93,0.95,0.97,0.98,0.99,
0.995,1\}$, denser near $\rho=1$ where the most rapid changes occur.

\paragraph*{Computed quantities.}
At every simulation round we compute the squared prediction error and
its RMSE counterpart $\sqrt{T_\post^{-1}\sum_t(\cdot)^2}$, the norms
$\|\mathrm{Term~A}(\rho)\|_2$ and $\|\mathrm{Term~B}(\rho)\|_2$ from
the closed forms~\eqref{eq:TA_closed}--\eqref{eq:TB_closed}, the three
channels $A_1,A_2,A_3$ of~\eqref{eq:A1_def}--\eqref{eq:A3_def},
$\lambda_{\max}(Q_\rho^{-1})$, and the transfer multiplier
$\mathcal{P}_\rho=\|C_\rho Q_\rho^{-1/2}\|_{\mathrm{op}}$, and average
them across $B=200$ replications.  The design is run twice: with the
baseline ridge $\zeta=T_\post^{1/4}\hat\sigma_{\Delta X}$ and with
$\zeta=0$, to isolate the stabilizing role of the ridge.

\subsection{Term~A: channels and identification geometry}
\label{app:mc_termA}

The channels $A_1,A_2,A_3$ and $\lambda_{\max}(Q_\rho^{-1})$ depend on
the data and the penalty $K_q$ but not on $\widetilde G_q$, so we
report them at one configuration per $q$ ($G_1^{\mathrm{const}}$ and
$G_2^{\mathrm{const}}$); the transfer multiplier
$\mathcal{P}_\rho=\|C_\rho Q_\rho^{-1/2}\|_{\mathrm{op}}$ depends on
$\widetilde G_q$ through $C_\rho=X_\post-\Pi_\rho X_\pre$ and is
reported for all four configurations.

\paragraph*{Three channels.}
Figure~\ref{fig:appx_A123_lmaxQinv} (top) plots the three channels of
Section~\ref{subsec:TermA_channels}.  The metric-distortion channel
$A_1$ is hump-shaped and vanishes at $\rho=1$ (where $W_\rho=P_\perp$).
The interaction channel $A_2$ is uniformly small and never the binding
channel.  The spurious matching channel
$A_3=T_0^{-1/2}\|e^\cR_\pre\|_{W_\rho}$ stays small and bounded in the
shared stochastic trend regime; in the shared + idiosyncratic regime
it is controlled at small $\rho$, where the $q$th differencing in
$W_0=K_q$ stationarizes the unit-specific random walks, and rises
sharply toward $\rho=1$ as $W_\rho$ retains their low-frequency
content.

\paragraph*{Identification geometry.}
Figure~\ref{fig:appx_A123_lmaxQinv} (bottom) plots
$\lambda_{\max}(Q_\rho^{-1})$ (Lemma~\ref{lem:Q_gap}).  Since
$Q_\rho\succeq\zeta^2 I_{N_0}$, the bound
$\lambda_{\max}(Q_\rho^{-1})\le 1/\zeta^2$ holds for every $\rho$.  The
ridged profile is nearly flat, around $0.25$ in the shared stochastic
trend regime and around $0.065$ in the shared + idiosyncratic regime.
With $\zeta=0$ (Figure~\ref{fig:appx_A123_lmaxQinv_noridge}) it is
markedly larger and far more variable, demonstrating the stabilizing
role of the ridge.

\paragraph*{Transfer multiplier.}
Figure~\ref{fig:appx_Prho} reports $\mathcal{P}_\rho$, the
$\rho$-dependent factor multiplying $A_1+A_2+A_3$ in the envelope of
Proposition~\ref{prop:TermA_env}.  In every configuration
$\mathcal{P}_\rho$ is larger in the shared + idiosyncratic regime,
peaking near $8$, than in the shared stochastic trend regime, where it
peaks near $3$ for the $q=1$ configurations and near $4$--$5$ for the
$q=2$ configurations (largest under the $q=2$ ARIMA$(1,1,0)$
forecaster).

\paragraph*{Tightness of the envelope.}
Figure~\ref{fig:appx_envelope} overlays
$\mathcal{P}_\rho\,(A_1+A_2+A_3)$ with the realized
$\|\mathrm{Term~A}(\rho)\|$.  In every ridged cell and at every grid
point the envelope lies weakly above the realized norm, confirming
Proposition~\ref{prop:TermA_env} numerically.  It is conservative,
exceeding $\|\mathrm{Term~A}\|$ by roughly five- to eighteenfold, the
gap being largest where $\|\mathrm{Term~A}\|$ is smallest, in the
shared stochastic trend regime.

\begin{figure}[!ht]
\centering
\includegraphics[width=0.9\textwidth]{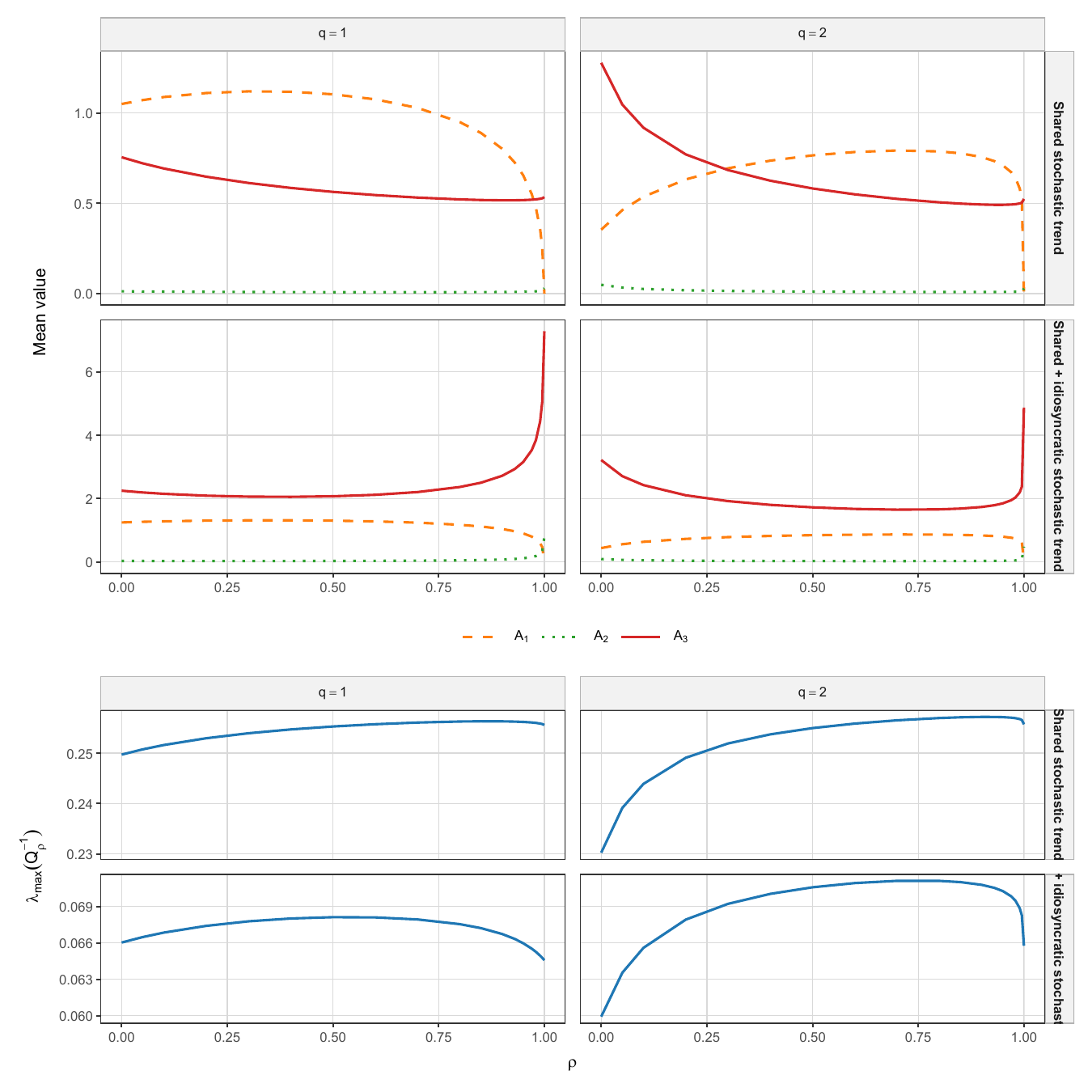}
\caption{Term~A channels and identification geometry.}
\label{fig:appx_A123_lmaxQinv}
\vspace{0.5em}
{\footnotesize\textbf{Note:} Top row: the three channels of the
Term~A envelope of Proposition~\ref{prop:TermA_env}, plotted as a
function of $\rho$.  Orange dashed: $A_1$, the metric-distortion
channel.  Green dotted: $A_2$, the interaction channel.  Red solid:
$A_3$, the spurious matching channel.  Bottom row:
$\lambda_{\max}(Q_\rho^{-1})$, the largest eigenvalue of the inverse
of the Hessian
$Q_\rho=T_0^{-1}X_\pre'W_\rho X_\pre+\zeta^2 I_{N_0}$.  Columns
correspond to the representative configurations $(q=1,
G_1^{\mathrm{const}})$ and $(q=2, G_2^{\mathrm{const}})$.  Top facet
within each column: shared stochastic trend regime ($\kappa=0$).
Bottom facet: shared + idiosyncratic stochastic trend regime
($\kappa=2$).  All curves are
means over $B=200$ replications of the DGP described in
Section~\ref{app:mc_setup}; ridge parameter
$\zeta=T_\post^{1/4}\hat\sigma_{\Delta X}$.\par}
\end{figure}

\begin{figure}[!ht]
\centering
\includegraphics[width=0.9\textwidth]{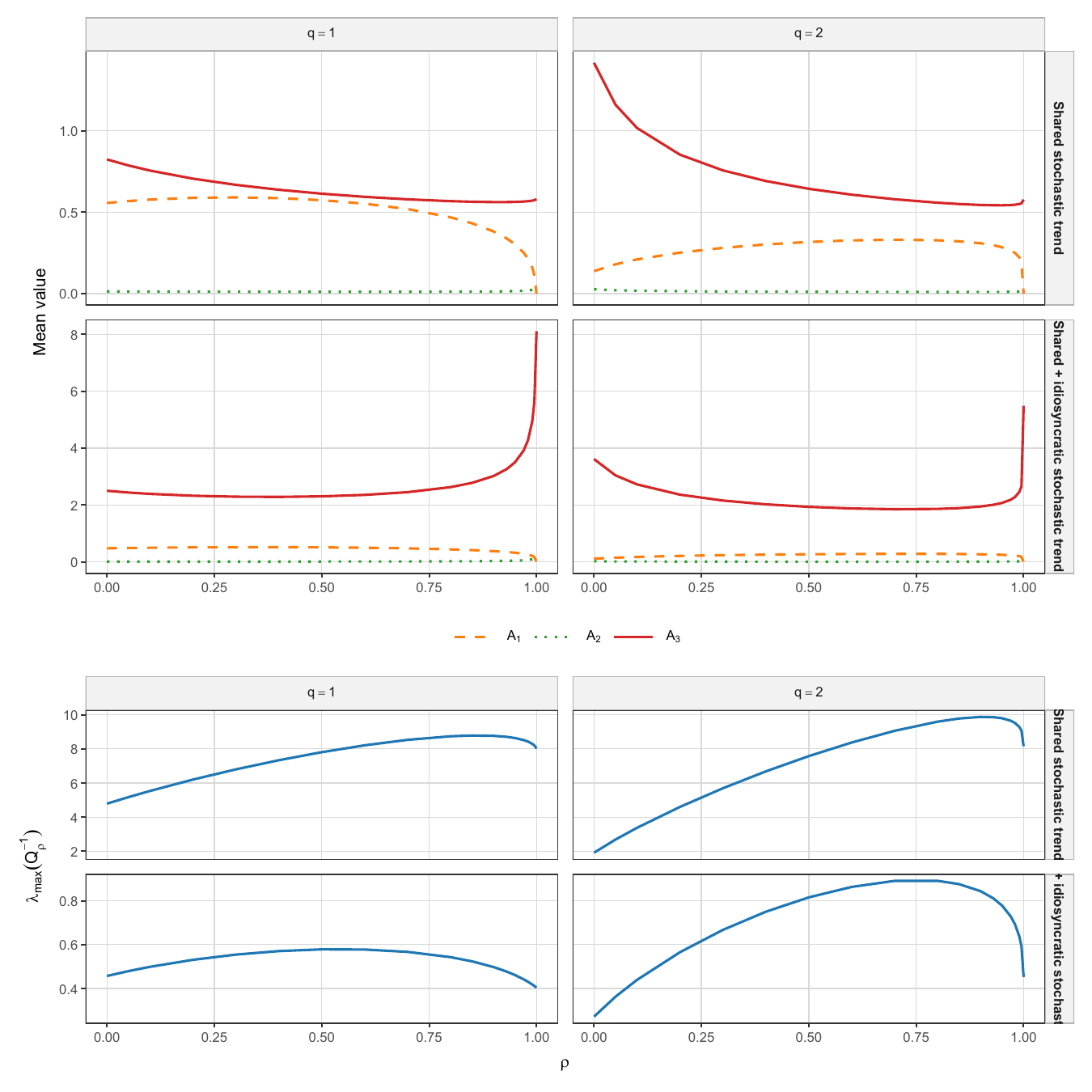}
\caption{Term~A channels and identification geometry without ridge
regularization ($\zeta=0$).}
\label{fig:appx_A123_lmaxQinv_noridge}
\vspace{0.5em}
{\footnotesize\textbf{Note:} Same layout as
Figure~\ref{fig:appx_A123_lmaxQinv} but with $\zeta=0$ (no ridge
regularization).  Without the ridge floor, $\lambda_{\max}(Q_\rho^{-1})$
can grow much larger, demonstrating the stabilizing role of the ridge
term.  The $A_1$, $A_2$, $A_3$ channels in the top row are also
computed at $\zeta=0$.\par}
\end{figure}

\begin{figure}[!ht]
\centering
\includegraphics[width=\textwidth]{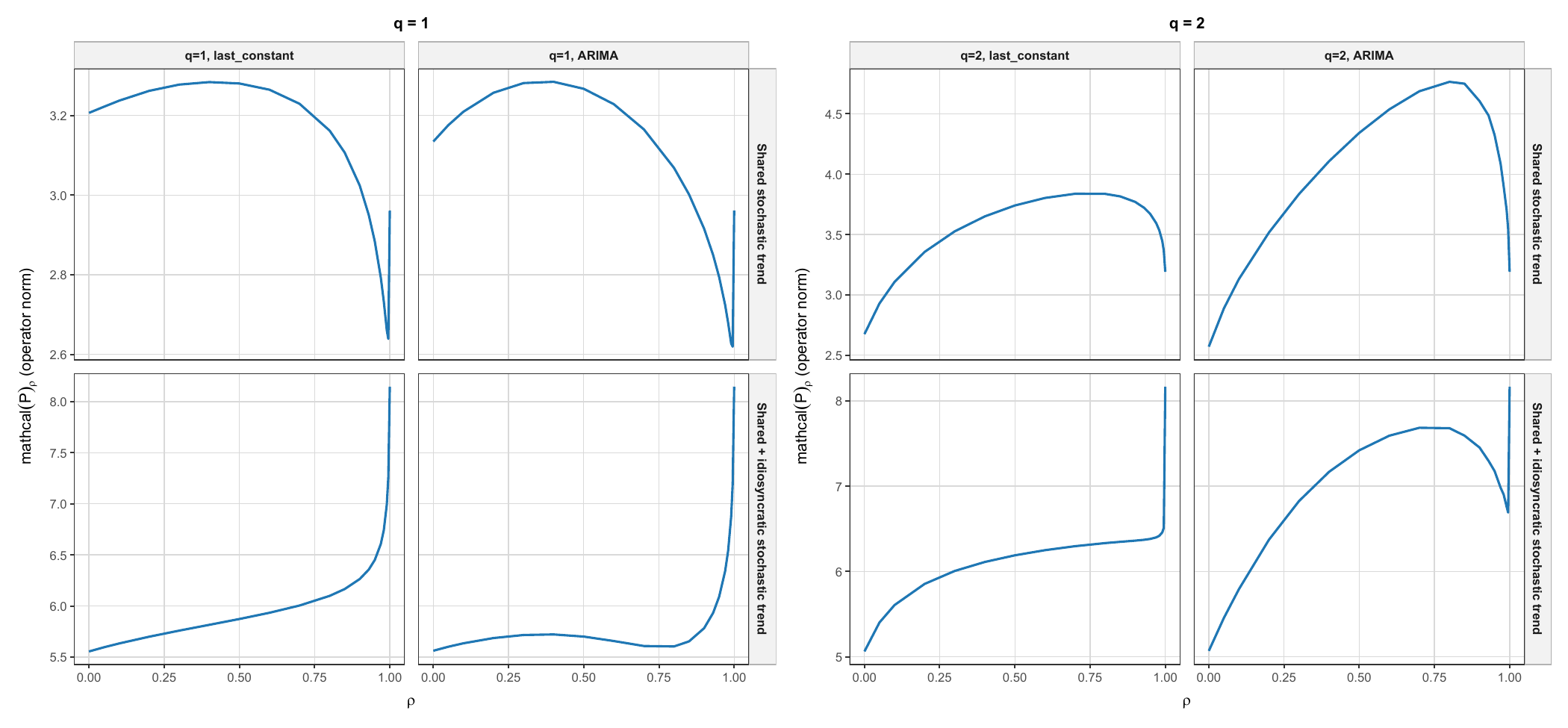}
\caption{Transfer multiplier $\mathcal{P}_\rho$.}
\label{fig:appx_Prho}
\vspace{0.5em}
{\footnotesize\textbf{Note:} Mean operator norm
$\mathcal{P}_\rho=\|C_\rho Q_\rho^{-1/2}\|_{\mathrm{op}}$ as a
function of $\rho$, where $C_\rho:=X_\post-\Pi_\rho X_\pre$.  Left
panel: $q=1$.  Right panel: $q=2$.  Within each panel, columns are
the two estimator configurations
($G_q^{\mathrm{const}}$ and ARIMA$(1,1,0)$) and rows are the two
DGP regimes (shared stochastic trend on top, shared +
idiosyncratic stochastic trend on bottom).  Means are over $B=200$
replications.\par}
\end{figure}

\begin{figure}[!ht]
\centering
\includegraphics[width=\textwidth]{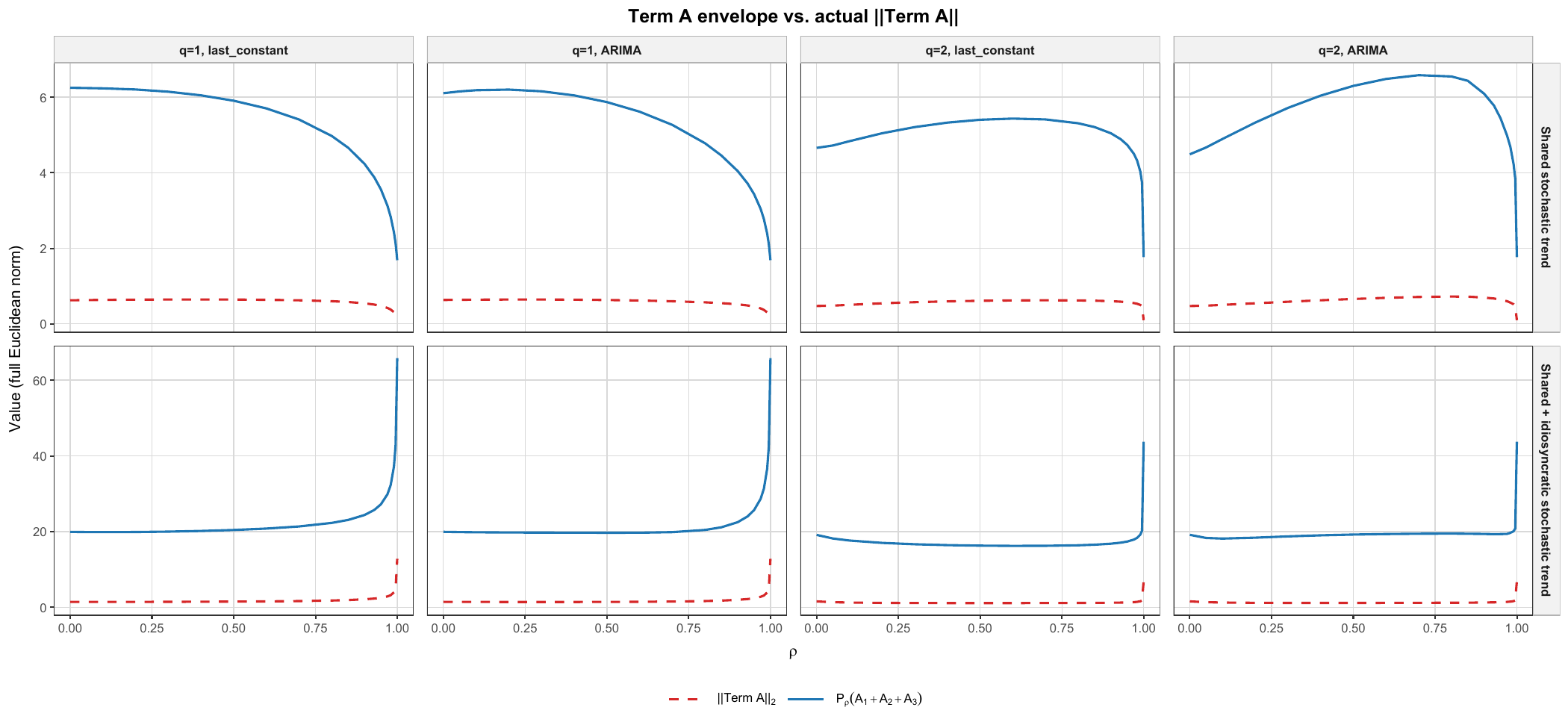}
\caption{Tightness of the Term~A envelope.}
\label{fig:appx_envelope}
\vspace{0.5em}
{\footnotesize\textbf{Note:} Mean Term~A envelope
$\mathcal{P}_\rho\,(A_1+A_2+A_3)$ of
Proposition~\ref{prop:TermA_env} (orange dashed) and the mean
realized norm $\|\mathrm{Term~A}(\rho)\|_2$ (black solid), as
functions of $\rho$.  Left panel: $q=1$.  Right panel: $q=2$.  Within
each panel, columns are the two estimator configurations
($G_q^{\mathrm{const}}$ and ARIMA$(1,1,0)$) and rows are the two DGP
regimes (shared stochastic trend on top, shared + idiosyncratic
stochastic trend on bottom).  The envelope lies weakly above the
realized norm at every grid point in every cell, confirming
Proposition~\ref{prop:TermA_env} numerically; the gap measures the
conservativeness of the uniform bound.  Means are over $B=200$
replications; ridge parameter
$\zeta=T_\post^{1/4}\hat\sigma_{\Delta X}$.\par}
\end{figure}

\subsection{Synthesis: $\mathrm{Term~A}$, $\mathrm{Term~B}$, and the
$\rho$-shape of the prediction error}
\label{app:mc_synthesis}

Figure~\ref{fig:appx_rmse_AB} overlays the per-period RMSE with the
per-period norms of $\mathrm{Term~A}$ and $\mathrm{Term~B}$.  Across
every panel $\|\mathrm{Term~B}\|$ exceeds $\|\mathrm{Term~A}\|$
throughout the grid, so forecasting is the dominant contributor to the
prediction error; and the $\rho$-shape of $\|\mathrm{Term~B}\|$ is
configuration-specific, neither monotone nor uniformly oriented across
the four $(q,\widetilde G_q)$ cells (Section~\ref{subsec:TermB}).

\paragraph*{Shared stochastic trend regime ($\kappa=0$).}
The RMSE is low, with mean RMSE $\approx0.70$--$0.77$ at the optimum
(Table~\ref{tab:appx_summary}).  $\|\mathrm{Term~B}\|$ is the larger
term throughout (per-period $\approx 0.6$ against $\approx 0.1$--$0.2$
for $\|\mathrm{Term~A}\|$), but both are small in absolute terms
($\|\mathrm{B}\|\approx 1.4$, $\|\mathrm{A}\|\approx 0.3$--$0.5$ at the
optimum).

\paragraph*{Shared + idiosyncratic stochastic trend regime ($\kappa=2$).}
Both terms grow and Term~B dominates: at the optimum
$\|\mathrm{Term~B}\|\approx 7.2$--$7.4$ against
$\|\mathrm{Term~A}\|\approx 1.1$--$1.6$, roughly a factor of five.  The
ex-post RMSE-minimizing $\rho^\star$ is small for every configuration
(between $0$ and $0.30$), the mean RMSE is $\approx 3.3$, and the RMSE
rises toward $\rho=1$ (Table~\ref{tab:appx_summary}).

\begin{figure}[!ht]
\centering
\includegraphics[width=\textwidth]{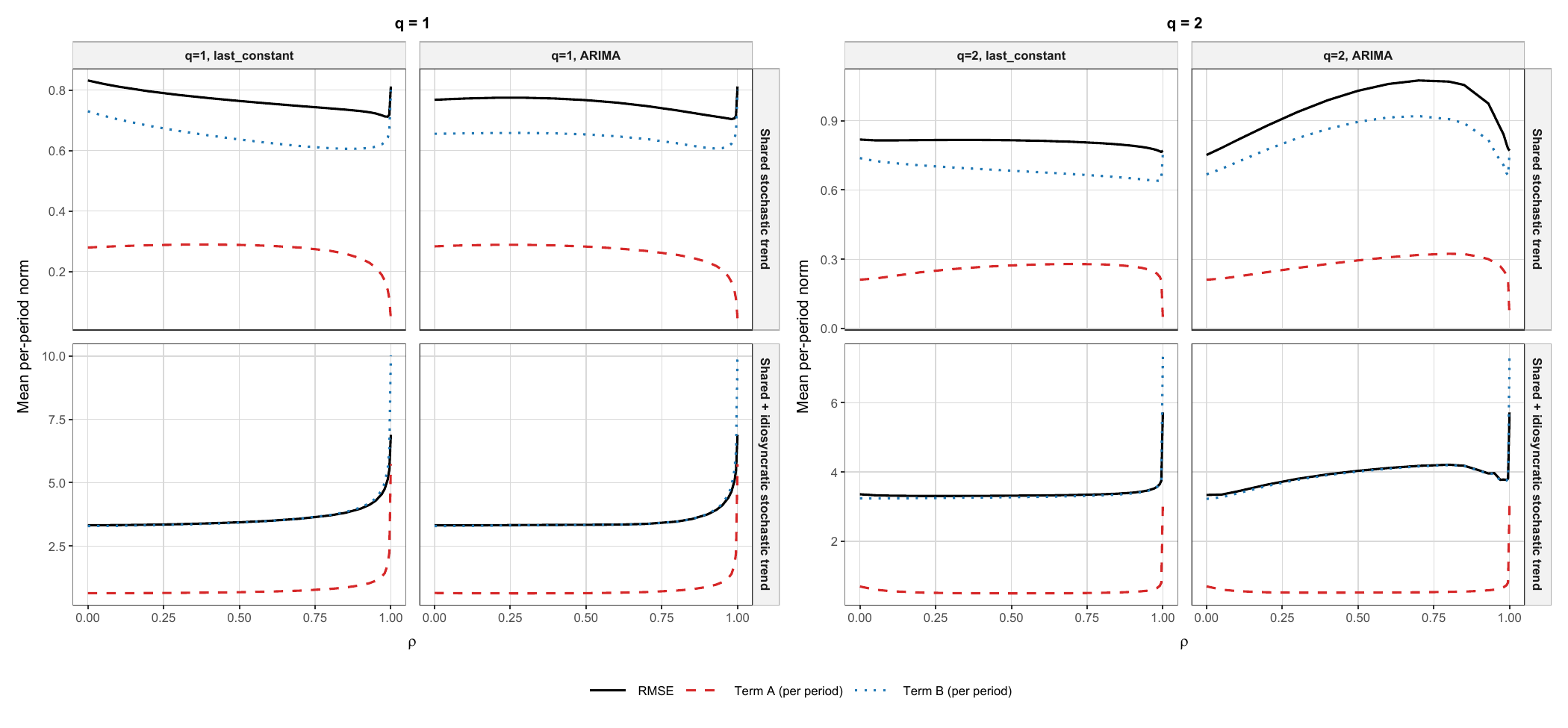}
\caption{Per-period RMSE, $\mathrm{Term~A}$, and $\mathrm{Term~B}$.}
\label{fig:appx_rmse_AB}
\vspace{0.5em}
{\footnotesize\textbf{Note:} Mean per-period RMSE (black solid),
mean per-period $\|\mathrm{Term~A}(\rho)\|_2/\sqrt{T_\post}$ (red
dashed), and mean per-period
$\|\mathrm{Term~B}(\rho)\|_2/\sqrt{T_\post}$ (blue dotted) as
functions of $\rho$.  Left panel: $q=1$.  Right panel: $q=2$.  Within
each panel, columns are the two estimator configurations
($G_q^{\mathrm{const}}$ and ARIMA$(1,1,0)$) and rows are the two
DGP regimes (shared stochastic trend on top, shared +
idiosyncratic stochastic trend on bottom).  Per-period normalization
places the three curves on the same scale as the RMSE, which is the
square root of the per-period mean squared error.  Means are over
$B=200$ replications.\par}
\end{figure}

\begin{table}[!ht]
\centering
\caption{Ex-post RMSE-minimizing $\rho^\star$, mean RMSE, and the
corresponding mean norms of $\mathrm{Term~A}$ and $\mathrm{Term~B}$
at the optimum.}
\label{tab:appx_summary}
\vspace{0.3em}
\resizebox{\textwidth}{!}{\small
\begin{tabular}{llcccccccc}
\hline\hline
& & \multicolumn{4}{c}{Shared stochastic trend ($\kappa=0$)}
& \multicolumn{4}{c}{Shared + idiosyncratic stochastic trend ($\kappa=2$)} \\
\cmidrule(lr){3-6}\cmidrule(lr){7-10}
& Configuration
& $\rho^\star$ & RMSE
& $\|\mathrm{A}\|$ & $\|\mathrm{B}\|$
& $\rho^\star$ & RMSE
& $\|\mathrm{A}\|$ & $\|\mathrm{B}\|$ \\
\hline
\multicolumn{10}{l}{\itshape Panel~A: Ridge regularization
$\zeta=T_\post^{1/4}\hat\sigma_{\Delta X}$ (default)} \\[2pt]
& $q=1$,\; $G_1^{\mathrm{const}}$
  & 0.990 & 0.71 & 0.31 & 1.44
  & 0.050 & 3.32 & 1.42 & 7.36 \\
& $q=1$,\; ARIMA$(1,1,0)$
  & 0.980 & 0.70 & 0.37 & 1.39
  & 0.050 & 3.31 & 1.42 & 7.35 \\
& $q=2$,\; $G_2^{\mathrm{const}}$
  & 0.995 & 0.76 & 0.47 & 1.44
  & 0.300 & 3.31 & 1.12 & 7.27 \\
& $q=2$,\; ARIMA$(1,1,0)$
  & 0.000 & 0.75 & 0.47 & 1.49
  & 0.000 & 3.34 & 1.56 & 7.20 \\[4pt]
\multicolumn{10}{l}{\itshape Panel~B: No ridge regularization
$(\zeta=0)$} \\[2pt]
& $q=1$,\; $G_1^{\mathrm{const}}$
  & 0.970 & 0.60 & 0.46 & 1.35
  & 0.100 & 3.41 & 4.05 & 8.06 \\
& $q=1$,\; ARIMA$(1,1,0)$
  & 0.970 & 0.60 & 0.46 & 1.35
  & 0.200 & 3.39 & 4.00 & 7.96 \\
& $q=2$,\; $G_2^{\mathrm{const}}$
  & 0.995 & 0.63 & 0.61 & 1.41
  & 0.700 & 3.39 & 4.15 & 8.15 \\
& $q=2$,\; ARIMA$(1,1,0)$
  & 1.000 & 0.63 & 0.45 & 1.40
  & 0.000 & 3.47 & 4.33 & 7.85 \\
\hline
\end{tabular}}
\vspace{4pt}
\begin{minipage}{0.95\textwidth}
\footnotesize
\textit{Notes:} Each entry reports the ex-post RMSE-minimizing grid
value $\rho^\star$ on the 19-point grid that minimizes mean per-period
RMSE over $B=200$ replications, the corresponding mean RMSE, and
the mean Euclidean norms $\|\mathrm{Term~A}(\rho^\star)\|_2$ and
$\|\mathrm{Term~B}(\rho^\star)\|_2$ at the same $\rho^\star$.  This
$\rho^\star$ is an ex-post RMSE minimizer and is distinct from the
cross-validation selector $\hat\rho$ of
Section~\ref{subsec:rho_selection}.  RMSE is the
per-period quantity $\sqrt{T_\post^{-1}\sum_t (\cdot)^2}$, while
$\|\mathrm{Term~A}\|$ and $\|\mathrm{Term~B}\|$ are full Euclidean
norms over the post-treatment window of length $T_\post=5$; dividing
the latter by $\sqrt{T_\post}\approx 2.24$ converts them to the same
per-period scale as the RMSE column.  The DGP is described in
Section~\ref{app:mc_setup}, with $N_0=10$, $T_0=80$, $T_\post=5$,
$\sigma_\Lambda=\sigma_\varepsilon=0.5$.  Panel~A uses the default
ridge parameter
$\zeta=T_\post^{1/4}\hat\sigma_{\Delta X}$; Panel~B sets
$\zeta=0$.
\end{minipage}
\end{table}

\paragraph*{Role of ridge regularization.}
In the shared stochastic trend regime the no-ridge RMSE is lower
(e.g.\ $0.60$ vs.\ $0.71$ for $q=1$, $G_1^{\mathrm{const}}$), since
without the ridge bias the estimator matches more precisely when
identification is strong.  In the shared + idiosyncratic regime
removing the ridge inflates $\|\mathrm{Term~A}\|$ by roughly
$2.8$--$3.7\times$ yet the RMSE barely moves (e.g.\ $3.41$ vs.\
$3.32$), because Term~B dominates and is largely unaffected by the
ridge.  The ridge thus stabilizes weight estimation but has little
effect on RMSE when forecasting error binds.

\subsection{Summary and limitations}
\label{app:mc_summary}

The Monte Carlo evidence matches Section~\ref{sec:predict_decompose}:
several channel profiles are non-monotone in $\rho$, and Term~B
consistently exceeds Term~A (roughly fivefold at the optimum in the
shared + idiosyncratic regime, threefold to fivefold in the shared
stochastic trend regime).  The ex-post RMSE-minimizing $\rho$ is small
throughout the shared + idiosyncratic regime; it depends jointly on
the DGP and the forecaster $\widetilde G_q$
(Section~\ref{subsec:synthesis}).

Three caveats apply.  The design fixes $N_0=10$, $T_0=80$,
$T_\post=5$; the channel magnitudes and the location of the RMSE
minimum depend on $T_0$ and on $\sigma_\Lambda/\sigma_\varepsilon$.
The persistent part of $\cR$ is a unit-specific random walk; richer time series processes would change the relative magnitudes of $A_3$
and may shift the optimal $\rho$.  And the figures are sample means
over $B=200$ replications.

Cross-validation (Section~\ref{subsec:rho_selection}) navigates this
landscape empirically, with $\rho$ allocating the pre-period
information between weight estimation (Term~A) and forecasting
(Term~B), and the channels of
Sections~\ref{subsec:TermA_channels}--\ref{subsec:TermB} accounting
for the shapes it traces.

\clearpage
\setcounter{page}{1}
\setcounter{table}{0}
\setcounter{figure}{0}
\setcounter{equation}{0}
\setcounter{footnote}{0}
\setcounter{lemma}{0}
\setcounter{proposition}{0}

\renewcommand{\theassumption}{C\arabic{assumption}}
\renewcommand\thetable{C\arabic{table}}
\renewcommand\thefigure{C\arabic{figure}}
\renewcommand{\thepage}{C-\arabic{page}}
\renewcommand{\theequation}{C\arabic{equation}}
\renewcommand{\thefootnote}{C\arabic{footnote}}
\renewcommand{\thelemma}{C\arabic{lemma}}

\section{Monte Carlo robustness for the estimator comparison}
\label{app:mc_estimators}

This appendix supplements the Monte Carlo evidence of
Section~\ref{sec:simulation}.
Section~\ref{app:mc_estimators_design} reports the simulation
details, including the parameter values for the data-generating
process, the seed scheme, and the frozen-versus-redrawn split
for each component of the panel.
Section~\ref{app:mc_perperiod} reports the per-period RMSE
trajectory that underlies the pooled-RMSE comparison of
Section~\ref{subsec:mc_rmse}.
Section~\ref{app:mc_bias_var} decomposes the pooled RMSE into
pooled bias and pooled variance and briefly attributes HSC's
RMSE advantage to a variance reduction.
Section~\ref{app:mc_cv_h} compares the cross-validated $\hat\rho$
and post-period RMSE between $h=1$ and $h=20$ cross-validation.

\subsection{Simulation details}
\label{app:mc_estimators_design}

The data-generating process used throughout
Section~\ref{sec:simulation} and this appendix is
\begin{equation}\label{eq:dgp_full}
Y_{j,t}(0)=L_{j,t}+\kappa\,\mathcal{E}_{j,t}+\varepsilon_{j,t}+\alpha_j+\delta_t,
\qquad j\in\{0,1,\dots,N_0\},\ t\in\{1,\dots,T_0+T_\post\},
\end{equation}
where the index $j=0$ denotes the treated unit and $j=1,\dots,N_0$
denote the donors (this appendix indexes the treated unit as $j=0$ for
notational convenience in the DGP).  The treated unit receives no
treatment effect ($\tau=0$), so the post-period error of any
counterfactual estimator $\hat Y_{0,T_0+h}$ is
$\hat Y_{0,T_0+h}-Y_{0,T_0+h}(0)$.

\paragraph*{Low-rank component.}
The component $L_{j,t}$ is built from $K=3$ latent factors:
\begin{equation}\label{eq:factor_mix}
L_{j,t}=\sum_{k=1}^{3}\Lambda_{j,k}F_{k,t},
\end{equation}
where $F_{1,t}$ is a random walk with innovation variance
$\sigma_{\mathrm{rw}}^{2}=4$, $F_{2,t}$ is an integrated
ARIMA$(1,1,0)$ process with autoregressive coefficient
$\phi_{F}=0.5$ and innovation variance $\sigma_{\mathrm{arima}}^{2}=4$,
and $F_{3,t}$ is a stationary AR(1) with autoregressive coefficient
$\rho_{s}=0.6$ and innovation variance $\sigma_{s}^{2}=1$.  The
factor paths $F_{k,t}$ are drawn fresh in every replication.  The
loadings $\Lambda_{j,k}$ for the donors are drawn iid from
$\mathcal{N}(0,0.5^{2})$ and truncated to $[-2,2]$, and are held
fixed across replications.

\paragraph*{Treated loading.}
The treated unit's loading vector is the convex combination
\begin{equation}\label{eq:lambda_treat}
\lambda_{0}=\sum_{j\in\mathcal{S}}\omega^{\star}_{j}\Lambda_{j,\cdot},
\qquad \omega^{\star}\in\Delta_{N_0},\ \abs{\mathcal{S}}=8,
\end{equation}
where $\mathcal{S}\subset\{1,\dots,N_0\}$ is a uniformly drawn
support of eight donors and $\omega^{\star}$ is drawn from a
sparse Dirichlet distribution on that support.  By construction
$\lambda_0$ lies inside the convex hull of the donor loadings.
The treated loading is held fixed across replications.

\paragraph*{Idiosyncratic stochastic trend.}
The component $\mathcal{E}_{j,t}$ is a unit-specific ARIMA$(1,1,0)$ process
with autoregressive coefficient $\phi_{e}=0.25$.  Its innovation
$U_{j,t}$ mixes a common shock and an idiosyncratic shock:
\begin{equation}\label{eq:E_innov}
U_{j,t}=\sqrt{\rho_u}\,u_{t}^{\mathrm{c}}+\sqrt{1-\rho_u}\,u_{j,t}^{\mathrm{i}},
\end{equation}
with $u_{t}^{\mathrm{c}},u_{j,t}^{\mathrm{i}}\sim\mathcal{N}(0,1-\phi_{e}^{2})$
iid.  The amplitude $\kappa\in\{0,0.5,1,2\}$ scales the entire
$\mathcal{E}_{j,t}$ component, and the cross-unit correlation parameter
$\rho_u\in\{0,0.5,1\}$ continuously interpolates between purely
idiosyncratic drift and purely shared drift.  Both shocks are
drawn fresh in every replication.

\paragraph*{Stationary noise and fixed effects.}
The stationary noise $\varepsilon_{j,t}$ is iid
$\mathcal{N}(0,1)$ and is drawn fresh in every replication.  The
unit fixed effect $\alpha_{j}$ is drawn iid from $U(5,15)$ for the
donors and fixed at $\alpha_{0}=0$ for the treated unit; the unit
fixed effects are held fixed across replications.  The time fixed
effect $\delta_{t}\sim\mathcal{N}(0,1)$ is iid and is drawn fresh
in every replication.

\paragraph*{Sample sizes and grid.}
We use $T_0=200$ pre-treatment periods, $T_\post=20$
post-treatment periods, $N_0=50$ donors, and $R=500$ replications
per cell of the $(\kappa,\rho_u)$ grid.  Table~\ref{tab:dgp_split}
summarizes which components are redrawn per replication and
which are held fixed.

\begin{table}[!ht]
    \centering
    \caption{Frozen versus redrawn components of the
    data-generating process.}\label{tab:dgp_split}
\resizebox{0.7\textwidth}{!}{\small
\begin{tabular}{lcc}
\hline\hline
Component & Redrawn per rep & Fixed across reps \\
\hline
Donor loadings $\Lambda_{j,k}$               &  & \checkmark \\
Treated loading $\lambda_{0}$                 &  & \checkmark \\
Unit fixed effects $\alpha_{j}$               &  & \checkmark \\
Factor paths $F_{k,t}$                        & \checkmark &  \\
Idiosyncratic component $\mathcal{E}_{j,t}$             & \checkmark &  \\
Stationary noise $\varepsilon_{j,t}$          & \checkmark &  \\
Time fixed effects $\delta_{t}$               & \checkmark &  \\
\hline
\end{tabular}}
\end{table}

The seed for each replication is derived deterministically from a
master seed and the cell index $(\kappa,\rho_u,\text{rep})$, so the
entire study is reproducible.  Bit-for-bit reproducibility is
verified by re-running the harness with the same master seed and
checking $R$-level identity of all numerical columns of the raw
output.

\subsection{Per-period RMSE}
\label{app:mc_perperiod}

Section~\ref{subsec:mc_rmse} reports RMSE pooled across the 20
post-treatment periods.  Figure~\ref{fig:c_perperiod} shows the period-by-period picture
behind the pooled comparison.  In the cells with $\kappa\ge 0.5$ and
$\rho_u\le 0.5$, the constant carry-forward configurations
(\texttt{last\_constant} at $q=1$ and $q=2$) and the $q=1$
\texttt{arima110} configuration lie at or below SC-INT and SDID at
every horizon.  The $q=2$ \texttt{ar}, \texttt{arima110}, and
\texttt{hamilton} trajectories instead rise above SDID and, at the
longest horizons, above SC-INT, with the unfavorable gap widening as
$h$ grows; the $q=1$ \texttt{hamilton} configuration (and, in the
$\kappa{=}0.5$, $\rho_u{=}0$ cell, $q=1$ \texttt{ar}) likewise exceeds
SDID at the longest horizons while staying below SC-INT.  For example,
at $(\kappa,\rho_u)=(2,0.5)$ and $h=20$ the $q=2$ \texttt{ar}
trajectory reaches about $11.6$ against SDID $9.5$ and SC-INT $10.6$.  In the shared stochastic trend
cells (top row $\kappa=0$ or right column $\rho_u=1$) the per-period
RMSEs are small (about $1.0$--$1.3$) and the HSC trajectories stay
near SC-INT and SDID, except the $q=2$ forecasters, which run up to
roughly $10$--$18\%$ above the better of the two baselines at the
longest horizons.

\begin{figure}[!ht]
\caption{Per-period RMSE by method across the
$(\kappa,\rho_u)$ grid}\label{fig:c_perperiod}
\centering
\begin{minipage}{1\linewidth}{
\centering
\includegraphics[width=\textwidth]{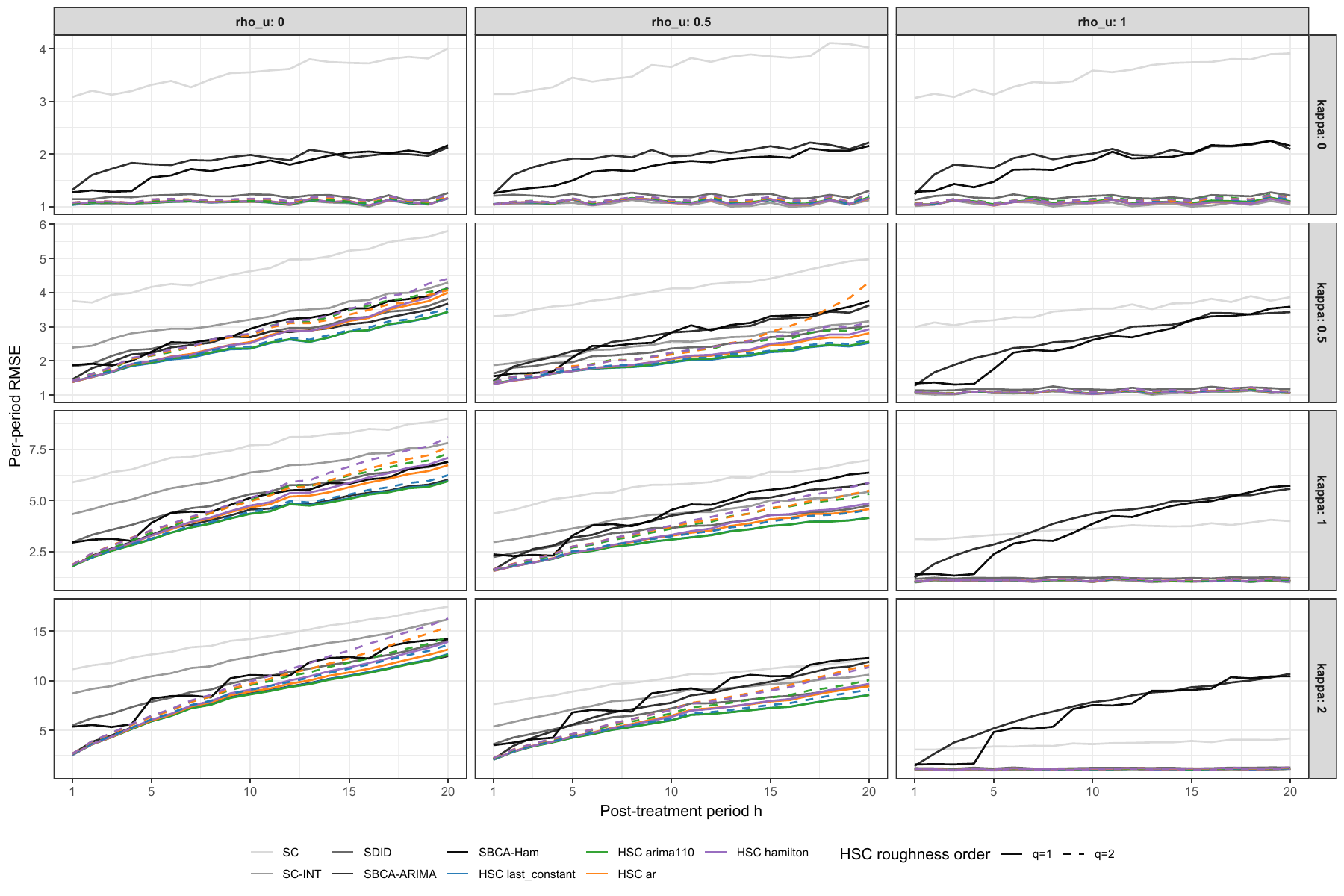}
}
\footnotesize\textit{Notes:} Lines report per-period RMSE
$\sqrt{R^{-1}\sum_{r}(\hat Y_{1,T_0+h}^{(r)}-Y_{1,T_0+h}^{(0,r)})^2}$
across $R=500$ replications, by post-treatment period
$h\in\{1,\dots,20\}$.  Solid lines correspond to $q=1$ and dashed
lines to $q=2$ for each HSC forecaster; baselines do not depend on
$q$.  $T_0=200$, $N_0=50$, $h=1$ cross-validation.
\end{minipage}
\end{figure}

\subsection{Pooled bias and pooled variance}
\label{app:mc_bias_var}

The pooled RMSE in Figure~\ref{fig:mc_rmse} can be decomposed into
a bias contribution and a variance contribution.
Figures~\ref{fig:c_pooled_bias} and~\ref{fig:c_pooled_variance}
report the pooled bias and pooled variance using the same layout
as Figure~\ref{fig:mc_rmse}.

\begin{figure}[!ht]
\caption{Pooled bias by method across the
$(\kappa,\rho_u)$ grid}\label{fig:c_pooled_bias}
\centering
\begin{minipage}{1\linewidth}{
\centering
\includegraphics[width=\textwidth]{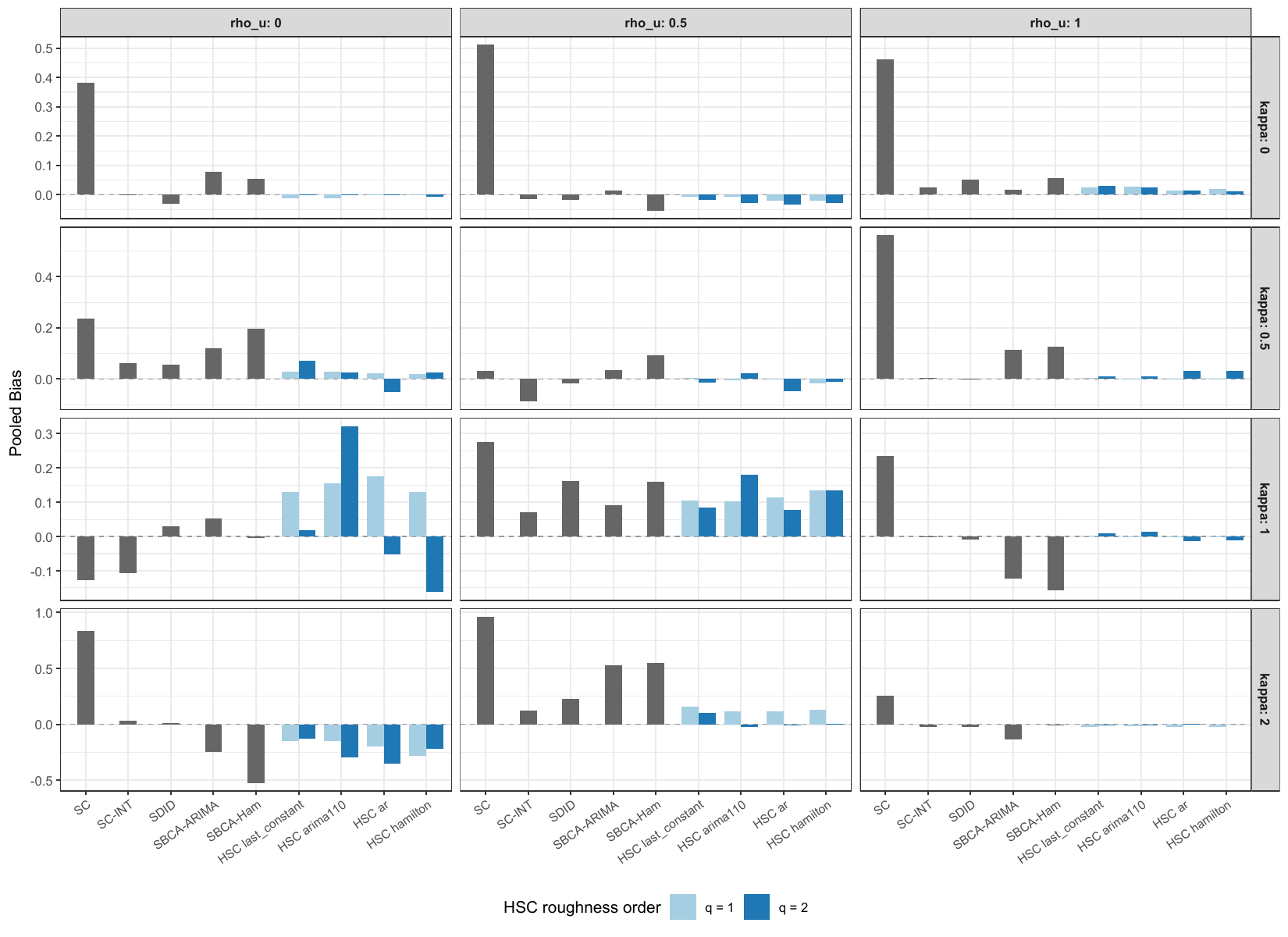}
}
\footnotesize\textit{Notes:} Bars report pooled bias
$R^{-1}T_\post^{-1}\sum_{r,h}(\hat Y_{1,T_0+h}^{(r)}-Y_{1,T_0+h}^{(0,r)})$
across $R=500$ replications and $T_\post=20$ post-treatment periods,
for each $(\kappa,\rho_u)$ cell.  Layout and color encoding match
Figure~\ref{fig:mc_rmse}.
\end{minipage}
\end{figure}

\begin{figure}[!ht]
\caption{Pooled variance by method across the
$(\kappa,\rho_u)$ grid}\label{fig:c_pooled_variance}
\centering
\begin{minipage}{1\linewidth}{
\centering
\includegraphics[width=\textwidth]{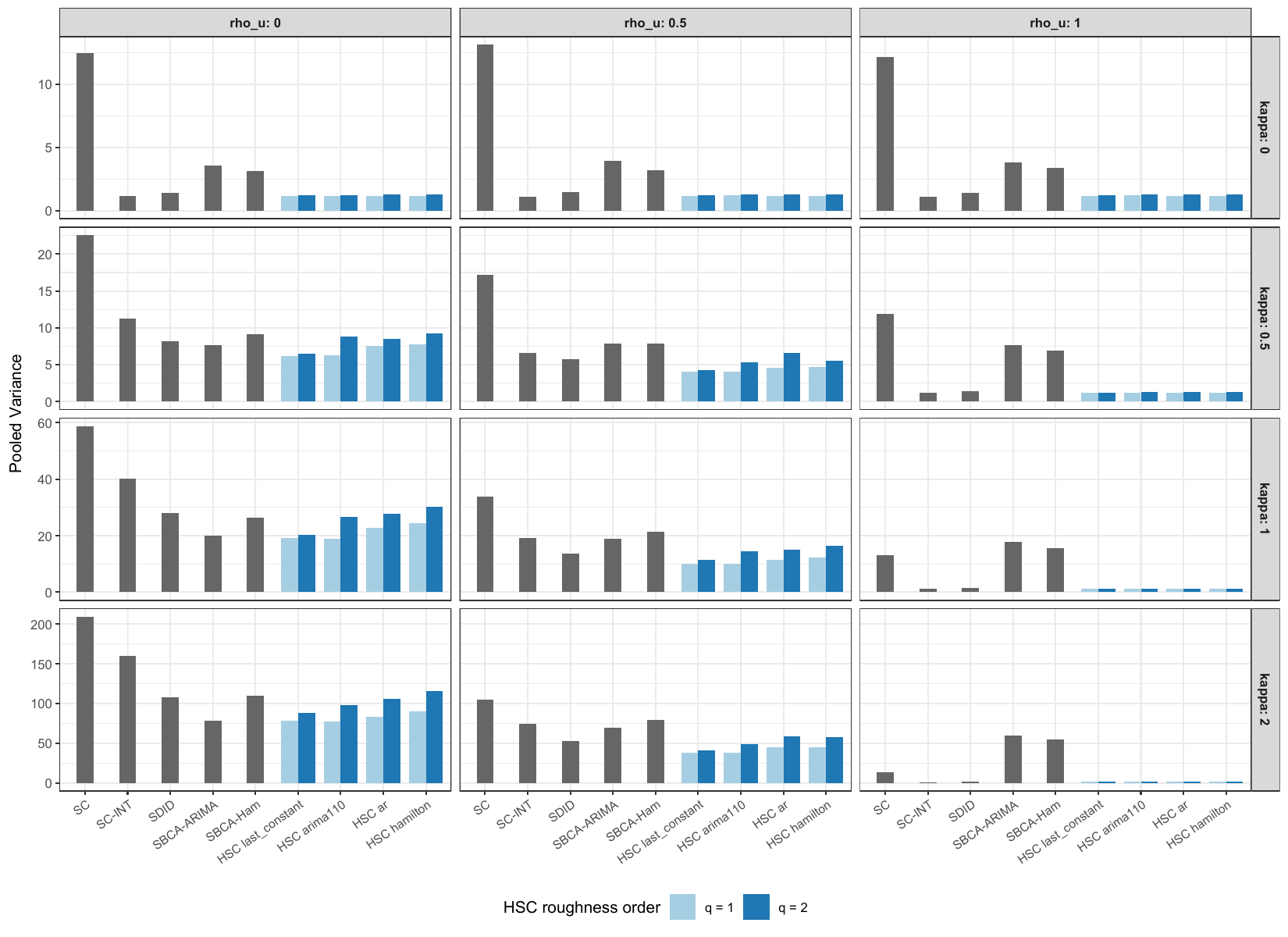}
}
\footnotesize\textit{Notes:} Bars report pooled variance
$\mathrm{Var}_{r,h}(\hat Y_{1,T_0+h}^{(r)}-Y_{1,T_0+h}^{(0,r)})$
across $R=500$ replications and $T_\post=20$ post-treatment periods,
for each $(\kappa,\rho_u)$ cell.  Layout and color encoding match
Figure~\ref{fig:mc_rmse}.
\end{minipage}
\end{figure}

In the shared stochastic trend cells the bias is small for every
method and the variance is comparable across HSC, SC-INT, and SDID.
In the shared + idiosyncratic stochastic trend cells HSC carries a
slightly larger bias than
SC-INT and SDID but a meaningfully smaller variance, and the
variance reduction more than offsets the bias penalty in the
pooled RMSE of Figure~\ref{fig:mc_rmse}.  The mechanism is
consistent with the spectral interpretation of
Section~\ref{sec:spectral}: by selecting an interior $\hat\rho$,
HSC trades a small amount of low-frequency information against a
substantial reduction in spurious matching of the donor pool to
the treated unit's idiosyncratic drift.

\subsection{Cross-validation horizon}
\label{app:mc_cv_h}

Section~\ref{subsec:mc_rho} fixes the cross-validation horizon at
$h=1$.  Here we examine how the cross-validated $\hat\rho$ and the
post-period RMSE shift when the horizon is extended to $h=20$.

\paragraph*{Distribution of $\hat\rho$.}
Figures~\ref{fig:c_rho_h_q1} and~\ref{fig:c_rho_h_q2} report the
distribution of $\hat\rho$ at $h=1$ and $h=20$ for $q=1$ and $q=2$
respectively.  In cells where $\hat\rho$ is already near one
under $h=1$ (the top row $\kappa=0$ and the right column
$\rho_u=1$), the distribution is essentially unchanged at $h=20$.
In cells where $\hat\rho$ is interior under $h=1$ (the
shared + idiosyncratic stochastic trend region with $\kappa\ge 1$ and $\rho_u\le 0.5$),
the distribution shifts upward at $h=20$: the median moves
substantially toward one and the boxes widen.  The shift is
systematic across forecasters and cells.  Intuitively, a
longer-horizon CV objective penalizes forecaster extrapolation
error more heavily, so the optimizer reallocates predictive
responsibility from the time series branch toward the donor pool
by raising $\rho$.

\begin{figure}[!ht]
\caption{Distribution of $\hat\rho$ at $h=1$ versus $h=20$
cross-validation, $q=1$}\label{fig:c_rho_h_q1}
\centering
\begin{minipage}{1\linewidth}{
\centering
\includegraphics[width=\textwidth]{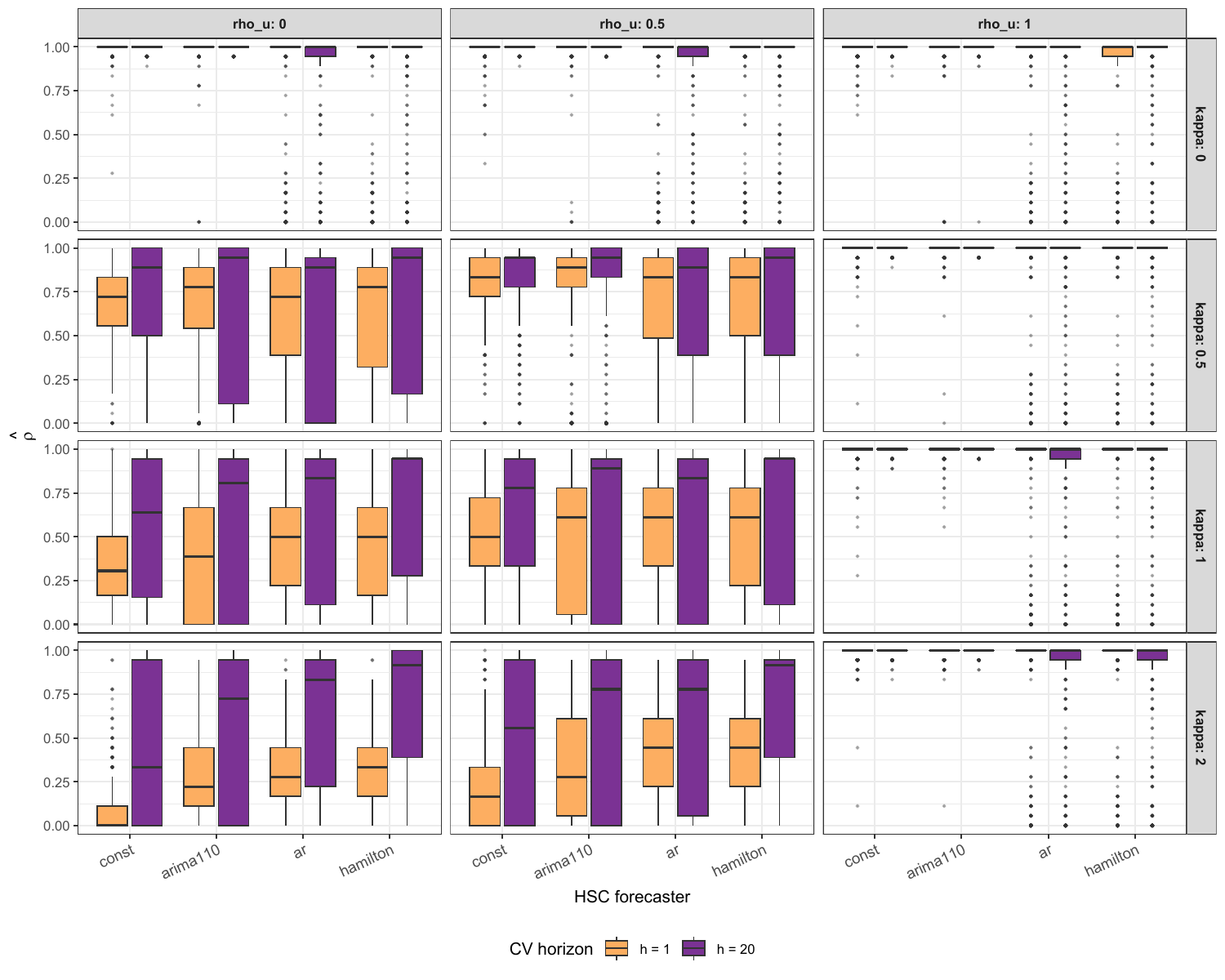}
}
\footnotesize\textit{Notes:} Boxplots report the distribution of
the cross-validated $\hat\rho$ across $R=500$ replications for each
$(\kappa,\rho_u)$ cell, at cross-validation horizons $h=1$ (orange)
and $h=20$ (purple), for the four HSC time series forecasters at
$q=1$.  $T_0=200$, $N_0=50$.
\end{minipage}
\end{figure}

\begin{figure}[!ht]
\caption{Distribution of $\hat\rho$ at $h=1$ versus $h=20$
cross-validation, $q=2$}\label{fig:c_rho_h_q2}
\centering
\begin{minipage}{1\linewidth}{
\centering
\includegraphics[width=\textwidth]{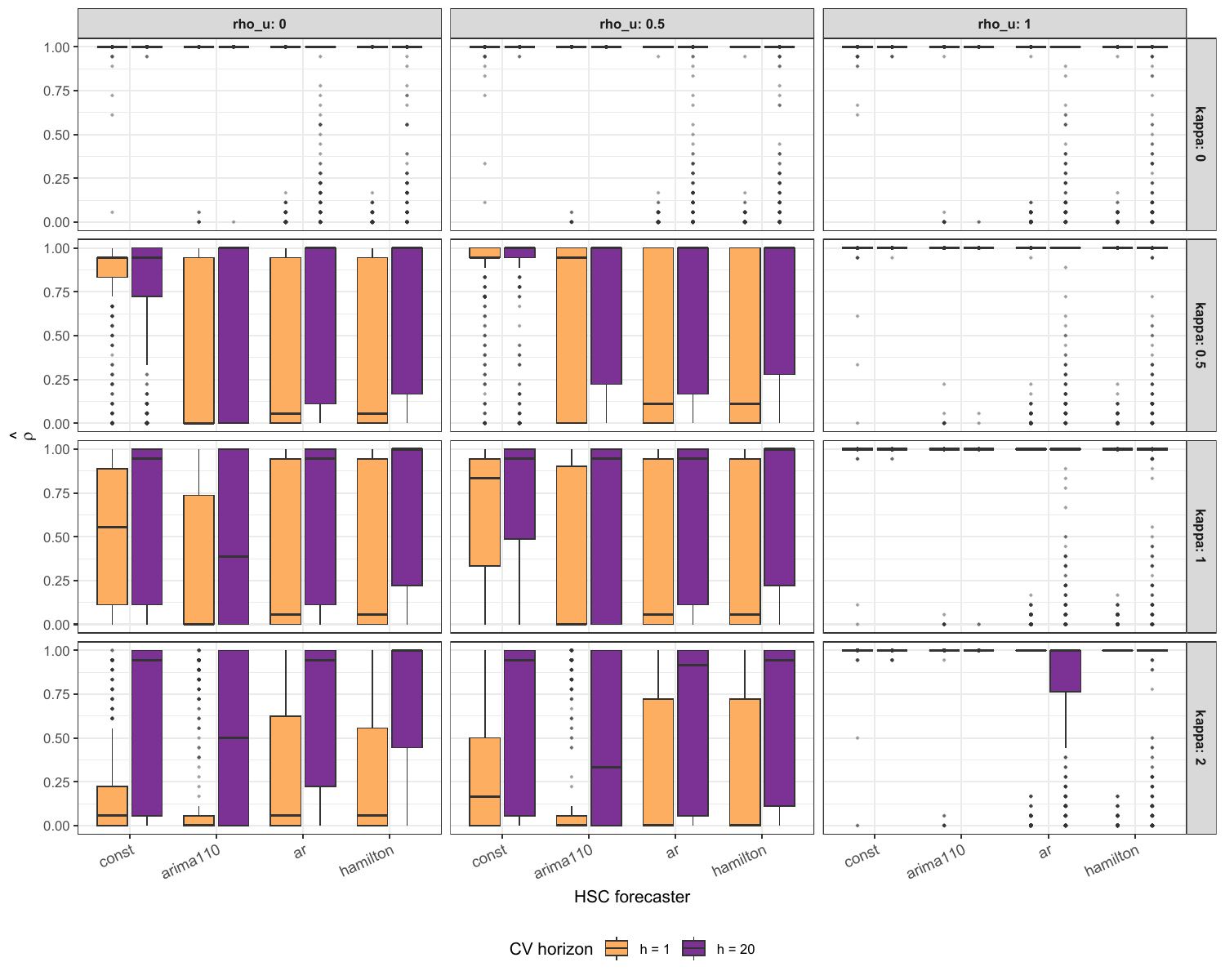}
}
\footnotesize\textit{Notes:} Same as Figure~\ref{fig:c_rho_h_q1}
for smoothness order $q=2$.
\end{minipage}
\end{figure}

\paragraph*{RMSE consequences of the horizon choice.}
The upward shift in $\hat\rho$ does not translate into uniformly
better RMSE.  Figure~\ref{fig:c_rmse_window} reports the
window-pooled RMSE at the representative shared + idiosyncratic
stochastic trend cell $(\kappa,\rho_u)=(2,0.5)$ for two windows: an
early window
covering post-treatment periods $\{1,\dots,10\}$ and a late
window covering $\{11,\dots,20\}$.  In the early window the $h=1$ HSC bars sit at or below the $h=20$
bars for every forecaster at both smoothness orders.  In the late
window this still holds for all $q=1$ configurations and for the $q=2$
\texttt{arima110}, \texttt{hamilton}, and \texttt{last\_constant}
configurations, but the $q=2$ \texttt{ar} configuration reverses, its
$h=20$ bar ($9.2$) falling below its $h=1$ bar ($9.6$).  Apart from
that single configuration the larger $\hat\rho$ selected by $h=20$
does not improve long-horizon predictive accuracy in this
representative cell; the
pattern is similar across other cells in the
shared + idiosyncratic stochastic trend region.

\begin{figure}[!ht]
\caption{Window-pooled RMSE at $h=1$ versus $h=20$
cross-validation, representative cell
$(\kappa,\rho_u)=(2,0.5)$}\label{fig:c_rmse_window}
\centering
\begin{minipage}{1\linewidth}{
\centering
\includegraphics[width=\textwidth]{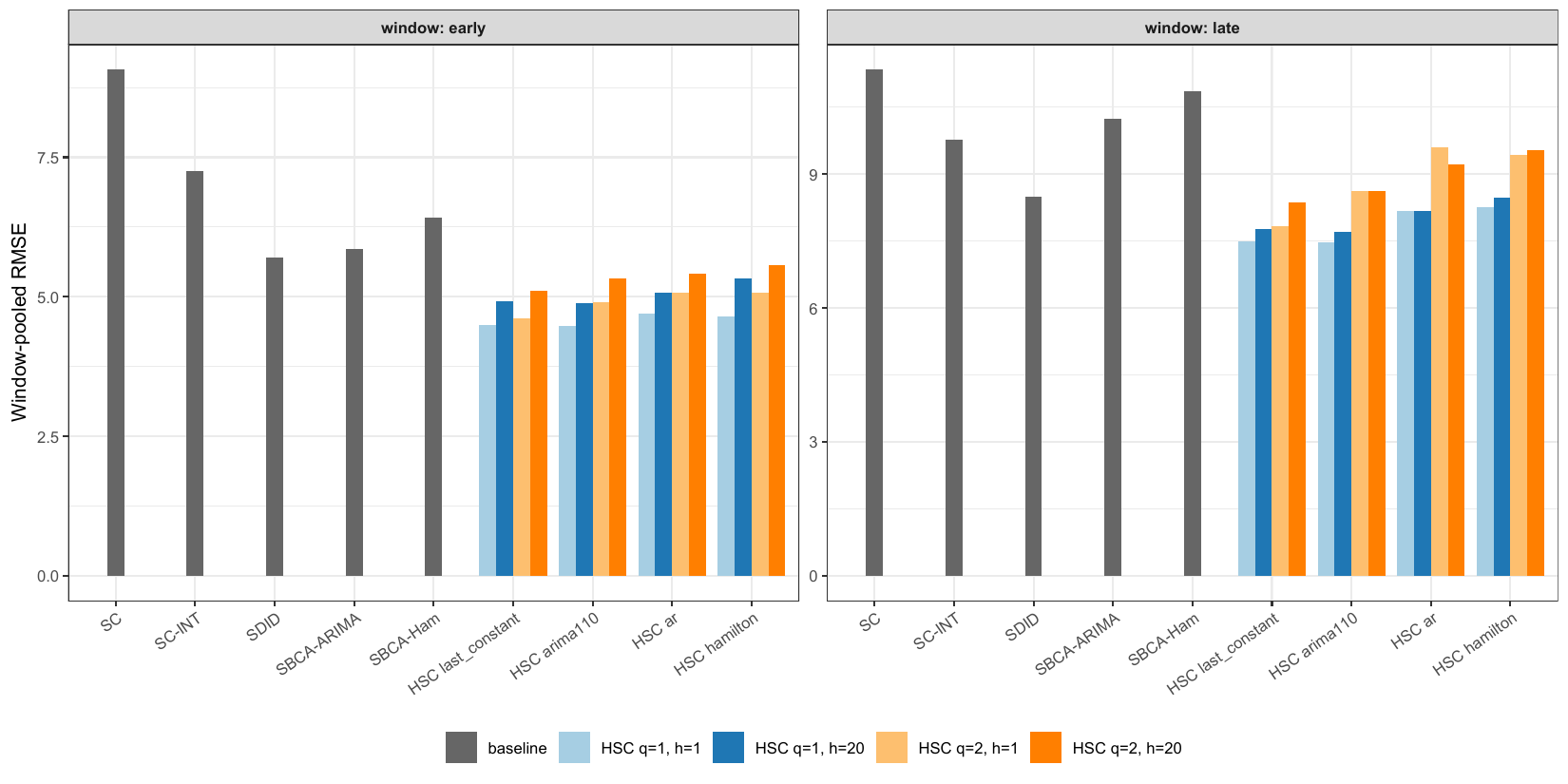}
}
\footnotesize\textit{Notes:} Bars report window-pooled RMSE
$\sqrt{R^{-1}\abs{W}^{-1}\sum_{r,h\in W}(\hat Y_{1,T_0+h}^{(r)}-Y_{1,T_0+h}^{(0,r)})^2}$
for the early window $W=\{1,\dots,10\}$ and the late window
$W=\{11,\dots,20\}$, restricted to the cell
$(\kappa,\rho_u)=(2,0.5)$.  Each HSC forecaster appears as four
bars: $q=1$ with $h=1$ (light blue), $q=1$ with $h=20$ (dark
blue), $q=2$ with $h=1$ (light orange), and $q=2$ with $h=20$
(dark orange).  The five baseline estimators do not depend on the
HSC cross-validation and appear as single grey bars.  $T_0=200$,
$N_0=50$.
\end{minipage}
\end{figure}

The takeaway is that the $h=1$ default used in the main text is a
robust choice: shortening the cross-validation horizon does not
sacrifice accuracy at the longer post-treatment horizons, and a
longer cross-validation horizon does not deliver an improvement
that compensates for the additional computational cost.

\clearpage
\setcounter{page}{1}
\setcounter{table}{0}
\setcounter{figure}{0}
\setcounter{equation}{0}
\setcounter{footnote}{0}
\setcounter{lemma}{0}
\setcounter{proposition}{0}

\renewcommand{\theassumption}{D\arabic{assumption}}
\renewcommand\thetable{D\arabic{table}}
\renewcommand\thefigure{D\arabic{figure}}
\renewcommand{\thepage}{D-\arabic{page}}
\renewcommand{\theequation}{D\arabic{equation}}
\renewcommand{\thefootnote}{D\arabic{footnote}}
\renewcommand{\thelemma}{D\arabic{lemma}}

\section{Hong Kong: robustness}
\label{app:hk_robust}

This appendix reports three robustness checks for the Hong Kong
application of Section~\ref{sec:empirical}: sensitivity to the choice
of HSC forecaster and roughness order, sensitivity to the
cross-validation horizon, and sensitivity to the donor pool.  All
three preserve the qualitative conclusions of the main text.

\subsection{All four HSC configurations}
\label{app:hk_allcfg}

Figure~\ref{fig:hk_d_cfg} overlays the counterfactuals of all four HSC
configurations, each evaluated at its own cross-validated $\hat\rho$.
The four trajectories stay close throughout the post-treatment
window---within about $2\%$ of one another in $1997$, widening to
about $6\%$ by $2003$---and all four imply a sizable negative effect,
so the headline result of Section~\ref{subsec:hk_cf} is not an
artifact of selecting the ARIMA$(1,1,0)$, $q=1$ configuration.

\begin{figure}[!ht]
\caption{Hong Kong: counterfactuals from all four HSC configurations
($h=1$)}\label{fig:hk_d_cfg}
\centering
\begin{minipage}{1\linewidth}{
\centering
\includegraphics[width=\textwidth]{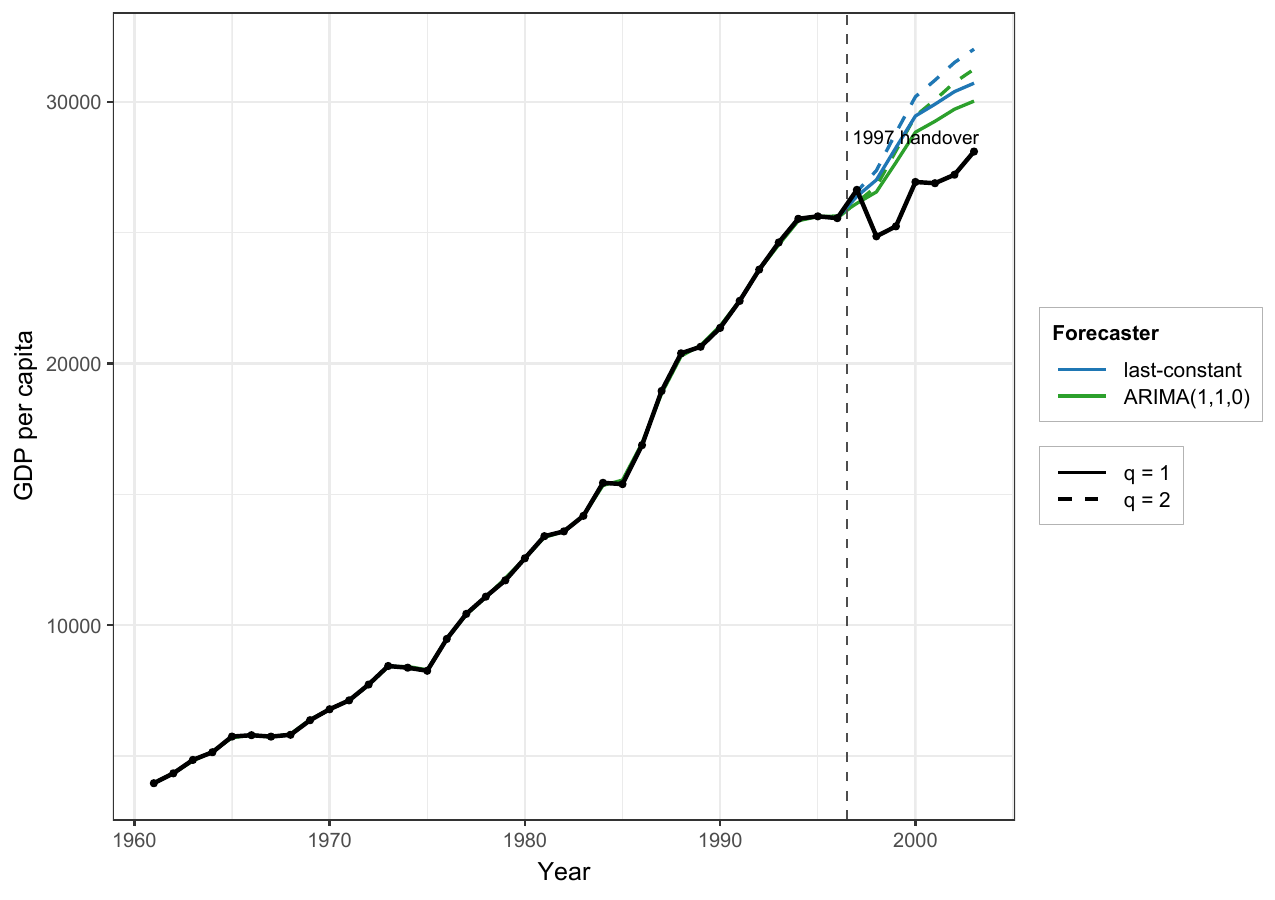}
}
\footnotesize\textbf{Note:} Estimated no-handover counterfactuals from
the four HSC configurations (\texttt{last\_constant} and
ARIMA$(1,1,0)$ forecasters at $q\in\{1,2\}$), each evaluated at its
own one-step-ahead cross-validated $\hat\rho$, with observed Hong Kong
per-capita GDP (solid black).  Sample: eleven developed donor
economies, annual per-capita GDP, $1961$--$2003$.
\end{minipage}
\end{figure}

\subsection{Cross-validation horizon}
\label{app:hk_h4}

The Monte Carlo study (Section~\ref{subsec:mc_rho}) shows that
lengthening the cross-validation horizon shifts $\hat\rho$ upward,
reallocating predictive responsibility from the treated-unit
forecaster toward the donor pool.  The same pattern appears in the
Hong Kong data.  Figure~\ref{fig:hk_d_cvh4_cv} shows that at $h=4$ the
cross-validated optimum of the selected configuration moves from
$\hat\rho=0.11$ to $\hat\rho=0.50$.  Figure~\ref{fig:hk_d_cvh4_cf}
shows the resulting counterfactual comparison: HSC continues to track
observed Hong Kong more closely than any baseline, and on the
four-step-ahead CV-MSPE all four HSC configurations
($1.6$--$2.0\times10^{6}$) again fall below every baseline
(SBCA-Hamilton $4.4\times10^{6}$, SDID $4.6\times10^{6}$, SC-INT
$9.2\times10^{6}$, plain SC $1.9\times10^{7}$).

\begin{figure}[!ht]
\caption{Hong Kong: cross-validated MSPE at the four-year horizon
($h=4$)}\label{fig:hk_d_cvh4_cv}
\centering
\begin{minipage}{1\linewidth}{
\centering
\includegraphics[width=\textwidth]{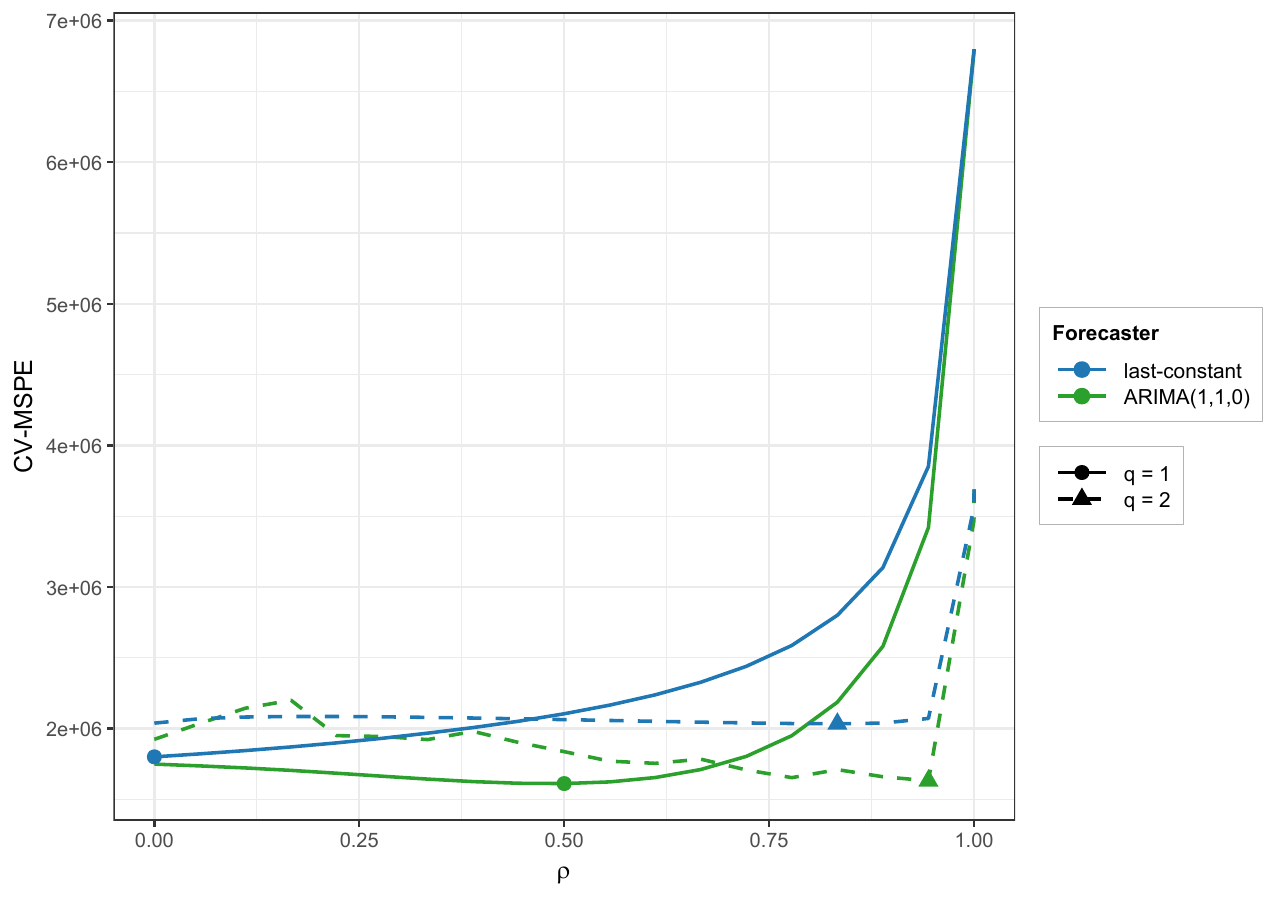}
}
\footnotesize\textbf{Note:} Cross-validated mean squared prediction
error at the four-step-ahead horizon ($h=4$) as a function of $\rho$,
for the four HSC configurations.  A marker on each curve denotes that
configuration's cross-validated $\hat\rho$.  Relative to the
one-step-ahead horizon of Figure~\ref{fig:hk_cv}, the selected
$\hat\rho$ shifts upward (from $0.11$ to $0.50$), reallocating
predictive responsibility toward the donor pool.  Sample: eleven
developed donor economies, $1961$--$1996$ pre-treatment.
\end{minipage}
\end{figure}

\begin{figure}[!ht]
\caption{Hong Kong: counterfactual per-capita GDP by estimator
($h=4$)}\label{fig:hk_d_cvh4_cf}
\centering
\begin{minipage}{1\linewidth}{
\centering
\includegraphics[width=\textwidth]{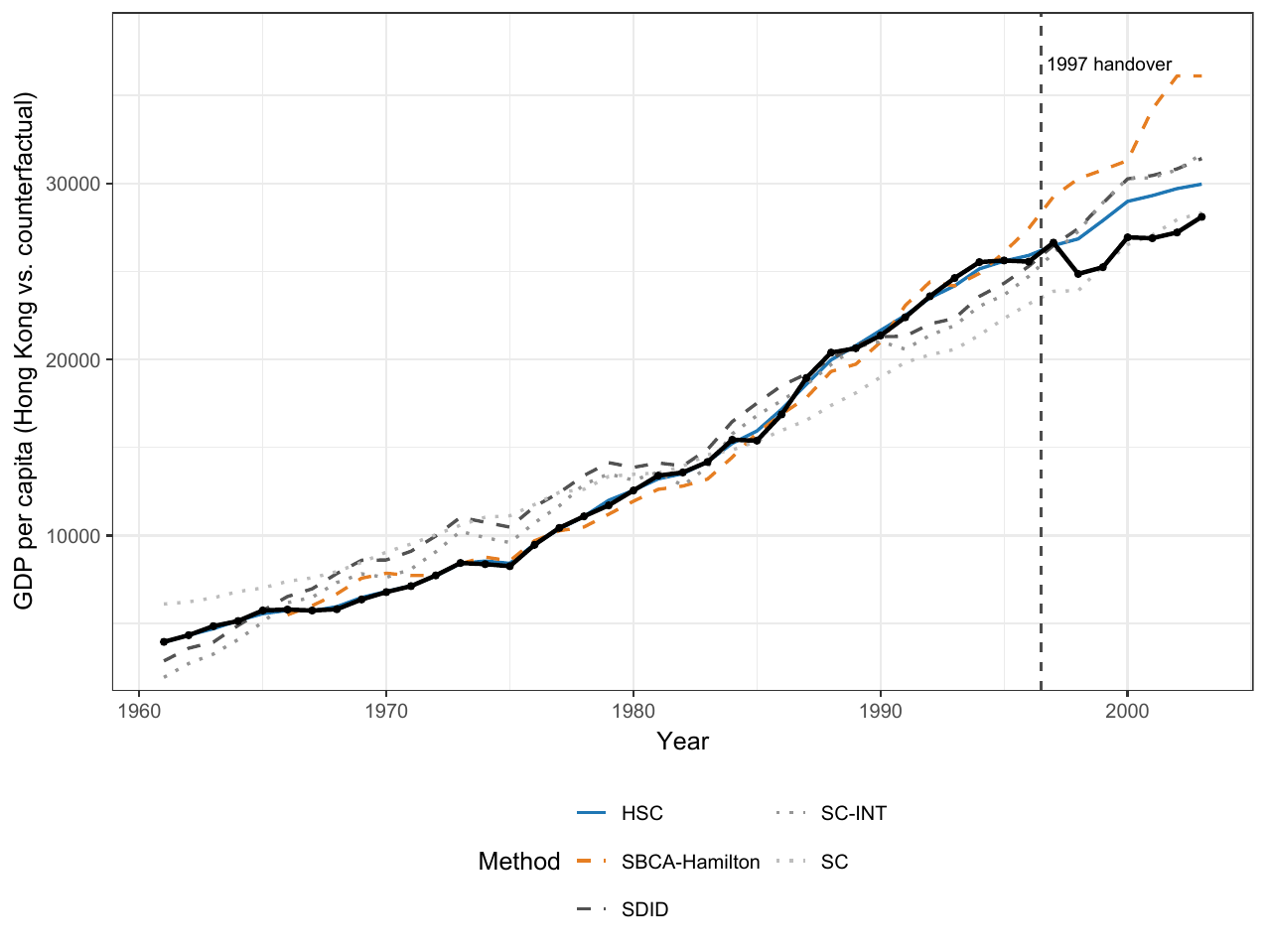}
}
\footnotesize\textbf{Note:} Observed Hong Kong per-capita GDP (solid
black) and estimated no-handover counterfactuals from the
cross-validation-selected HSC configuration at $h=4$, SBCA-Hamilton,
SDID, SC-INT, and plain SC.  The vertical dashed line marks the
$1997$ handover.  Sample: eleven developed donor economies, annual
per-capita GDP, $1961$--$2003$.
\end{minipage}
\end{figure}

\subsection{Geographic-neighbour donor pool}
\label{app:hk_hcw}

\citet{hsiao2012panel} select donors by geographic and economic
proximity rather than long-run comparability.  We re-estimate HSC and
the baselines on the nine economies from that pool that are available
in the \citet{shi2025synthetic} data release---mainland China,
Indonesia, Japan, Korea, Malaysia, the Philippines, Singapore,
Thailand, and the United States; Taiwan, also used by
\citet{hsiao2012panel}, is not in the release---over the same
$1961$--$2003$ window.  The cross-validation again selects
ARIMA$(1,1,0)$, $q=1$ with $\hat\rho=0.11$.
Figure~\ref{fig:hk_d_asianw} shows that HSC again distributes weight
broadly across the donor pool while the comparison estimators
concentrate, and Figure~\ref{fig:hk_d_asiancfg} shows that the four
HSC configurations again cluster.  On the one-step-ahead CV-MSPE all
four HSC configurations ($3.6$--$3.9\times10^{5}$) again fall below
every baseline (SDID $5.1\times10^{5}$, SC-INT $1.1\times10^{6}$,
plain SC $1.2\times10^{6}$, SBCA-Hamilton $1.6\times10^{6}$).  HSC's advantage over every baseline therefore survives replacing the
entire donor-selection philosophy.  The ordering among the baselines
does change: with a geographic-neighbour donor pool SBCA-Hamilton
becomes the weakest comparator rather than the strongest, while SDID,
SC-INT, and plain SC all improve substantially.

\begin{figure}[!ht]
\caption{Hong Kong: donor weights by estimator, geographic-neighbour
pool}\label{fig:hk_d_asianw}
\centering
\begin{minipage}{1\linewidth}{
\centering
\includegraphics[width=\textwidth]{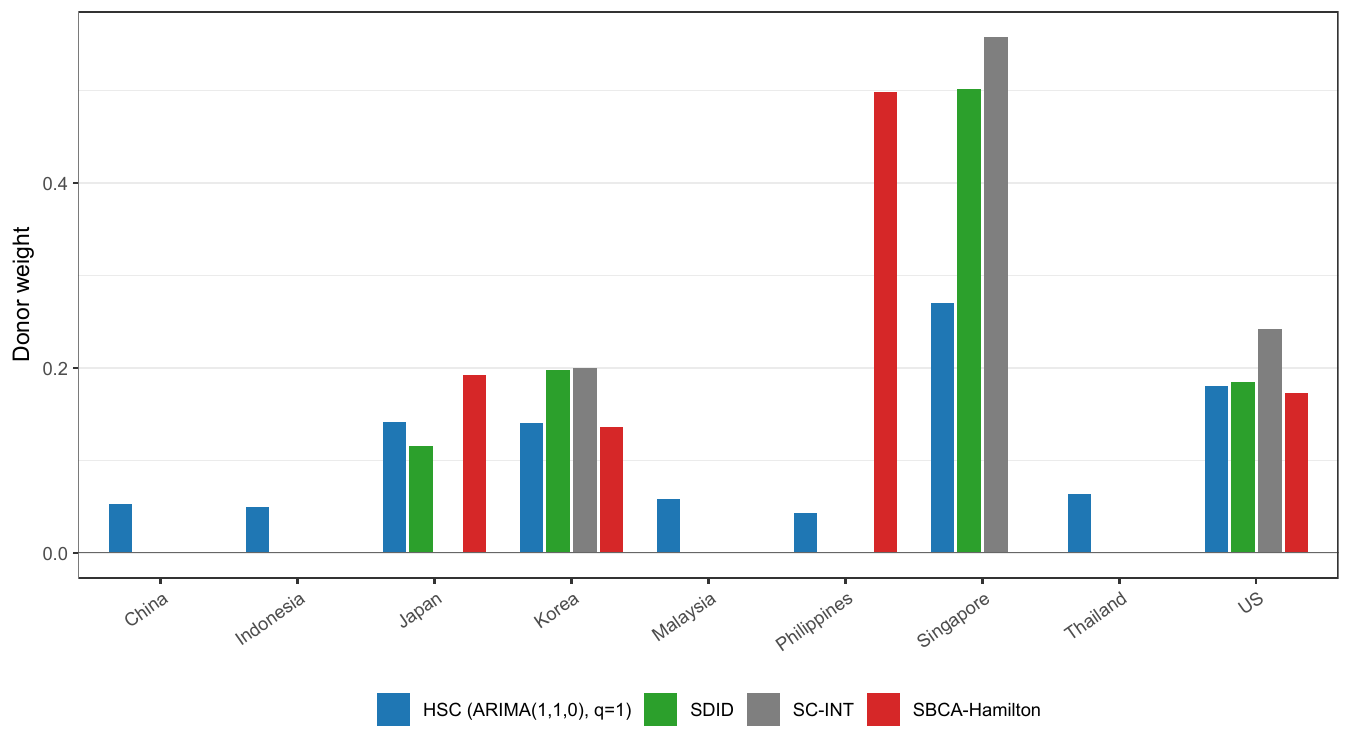}
}
\footnotesize\textbf{Note:} Donor weights assigned to the nine
geographic-neighbour donor economies of \citet{hsiao2012panel}
available in the \citet{shi2025synthetic} release, by the
cross-validation-selected HSC configuration (ARIMA$(1,1,0)$, $q=1$,
$\hat\rho=0.11$), SDID, SC-INT, and SBCA-Hamilton.  HSC again spreads
weight broadly while the comparison estimators concentrate.  Sample:
annual per-capita GDP, $1961$--$1996$ pre-treatment fitting window.
\end{minipage}
\end{figure}

\begin{figure}[!ht]
\caption{Hong Kong: counterfactuals from all four HSC configurations,
geographic-neighbour pool ($h=1$)}\label{fig:hk_d_asiancfg}
\centering
\begin{minipage}{1\linewidth}{
\centering
\includegraphics[width=\textwidth]{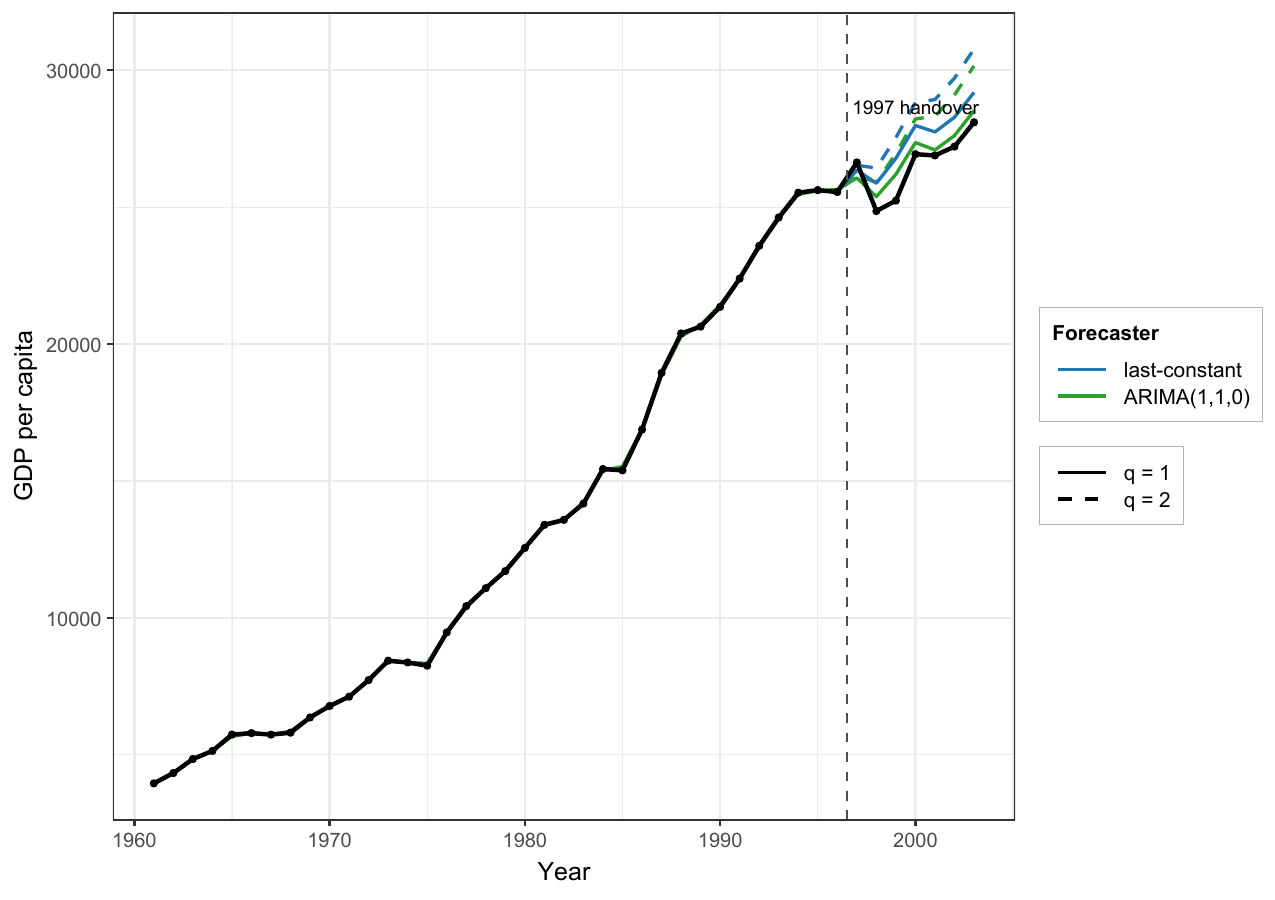}
}
\footnotesize\textbf{Note:} Estimated no-handover counterfactuals from
the four HSC configurations, each at its own one-step-ahead
cross-validated $\hat\rho$, with observed Hong Kong per-capita GDP
(solid black), estimated on the nine geographic-neighbour donors of
\citet{hsiao2012panel}.  Sample: annual per-capita GDP,
$1961$--$2003$.
\end{minipage}
\end{figure}

\end{document}